\documentclass{article}
\usepackage[utf8]{inputenc}
\usepackage{amsmath}
\usepackage{amsthm}
\usepackage{lipsum}
\usepackage{amsfonts}
\usepackage{amssymb}
\usepackage{graphicx}
\usepackage{subcaption}
\usepackage[numbers]{natbib}
\usepackage{verbatim}
\usepackage{upquote}
\usepackage{textcomp}
\usepackage{setspace}
\usepackage{ifthen}
\usepackage{authblk}
\usepackage{float}
\usepackage{multicol}
\usepackage[ruled,vlined]{algorithm2e}
\usepackage{algorithmicx}
\usepackage[noend]{algpseudocode}
\newcommand*{\affaddr}[1]{#1} 
\newcommand*{\affmark}[1][*]{\textsuperscript{#1}}
\newcommand*{\email}[1]{\texttt{#1}}
\newcommand{\etal}{{\it et al.}}

\newtheorem{corollary}{Corollary}
\newtheorem{lemma}{Lemma}
\newtheorem{theorem}{Theorem}

\newtheorem{definition}{Definition}

\usepackage[margin=0.9in]{geometry}
\title{\textbf{Servicing Timed Requests on a Line}}
\date{}

\author{%
Angelos Gkikas\affmark[1], Tomasz Radzik\affmark[1]\\
\affaddr{\affmark[1]Department of Informatics, Kings College London}\\
\email{\{angelos.gkikas,tomasz.radzik\}@kcl.ac.uk}\\
\affaddr{}%
}
\begin{document}

\maketitle
\begin{abstract}
We consider an off-line optimisation problem where $k$ robots must service $n$ requests on a single line. A request $i$ has weight $w_i$ and takes place at time $t_i$ at location $d_i$ on the line. A robot can service a request and collect the weight $w_i$, if it is present at $d_i$ at time $t_i$. The objective is to find $k$ robot-schedules that maximize the total weight. The optimisation problem is  motivated by a robotics application  \cite{DBLP:journals/dam/AsahiroHMOSY06} and can be modeled as a minimum cost flow problem with unit capacities in a flow network $\mathcal{N}$. Consequently, we ask for a collection of $k$ node-disjoint paths from the source $s$ to the sink $t$ in $\mathcal{N}$, with minimum total weight. It was shown in \cite{DBLP:journals/dam/AsahiroHMOSY06} that the flow network $\mathcal{N}$ can be implicitly represented by $n$ points on the plane which yields to an $O(n \log n)$-time algorithm for $k=1$ and the special case where all requests have the same weight. However, for $k \geq 2$ the problem can be solved in $O(kn^2)$ time with the successive shortest path algorithm which does not use this implicit representation. We consider arbitrary request weights and show a recursive $O(k^{2k}n \log^{2k} n)$-time algorithm  which improves the previous bound if $k$ is considered constant. Our result also improves the running time of previous algorithms for other variants of the optimisation problem. Finally, we show problem properties that may be useful within the context of applications that motivate the problem and may yield to more efficient algorithms.
\end{abstract}
\newcommand{\BCPgraph}{BCP graph}

\newcommand{\predecessor}{\mbox{\it pred}}

\newcommand{\relaxop}{\mbox{\it relax}}
\newcommand{\Relaxop}{\mbox{\it Relax}}
\newcommand{\polylog}{\mbox{polylog}}

 \section{Introduction}
We consider the following optimisation problem of servicing
timed requests on the line. 
For a given integer $k \ge 1$ and 
a set of $n$ timed requests $\{(x_i,t_i,w_i), 
1 \le i \le n\}$, where $x_i\in (-\infty, + \infty)$, 
$t_i \ge 0$ and $w_i \ge 0$ are the location 
(on the line),
the time and the weight of request $i$, respectively, 
maximise the total weight of requests 
which can be serviced 
by $k$ robots. Initially, at time $t=0$, all robots are 
at the origin point of the line and they can move 
freely along 
the line, changing direction and speed when needed, 
but never exceeding a given maximum speed $v$.
To service request $i$, one of the robots has to be 
at location $d_i$ exactly at time $t_i$.
Servicing a request is instantaneous and the robot can 
move immediately to serve another request.

This is an off-line optimisation problem
with all data about the requests known in advance, 
which appeared, 
for example,
in the context of 
the \emph{ball collecting problems} (BCPs) 
considered by Asahiro~\etal\cite{DBLP:journals/dam/AsahiroHMOSY06}. 
The basic BCP is essentially the optimisation 
problem stated in the previous paragraph.
There are $n$ weighted balls  approaching 
the line $L$ where the robots can move.
Each ball will cross $L$
at a specified time and  point, and
if a robot is there, then 
the ball is intercepted (collected). For the \textit{weighted} case the objective is to compute the movement of the robots so
that the total weight of the intercepted balls is maximised. For the \textit{unweighted} case (i.e. all balls have the same weight) the objective is to maximise the number of intercepted balls.

Asahiro~\etal\cite{DBLP:journals/dam/AsahiroHMOSY06}
studied a number of BCP variants,
putting them in the context of the \emph{Kinetic Travelling Salesman Problem}
(KTSP)
and establishing the tractability--intractability 
(polynomiality vs. NP-hardness) frontier through
the landscape of the studied variants. 
The literature of the KTSP consists of similar work, focusing on approximation algorithms \cite{Hammar:1999:ARK:646229.681565,Chalasani1996AlgorithmsFR,Asahiro2008}, polynomial time exact algorithms \cite{Asahiro2008,HELVIG2003153,RePEc:spr:annopr:v:289:y:2020:i:2:d:10.1007_s10479-019-03412-x}  
for special problem settings and real world applications\cite{Menezes2015158,6858759,5400538}.

Variants of the BCP are obtained 
by giving each robot $i$ its own 
line $L_i$ where it moves and intercepts balls,
or by not-fixing the position (i.e. the angle) of the common line $L$
(or the positions of the robots' individual lines $L_i$), asking instead for the optimal position of the line 
to be determined as part of the output,
or by considering different optimisation objectives
(e.g.,  minimizing the number of robots needed to 
collect all balls). Asahiro~\etal\
\cite{DBLP:journals/dam/AsahiroHMOSY06} showed that
maximising the total weight of collected balls
when robots move on a common line $L$ is polynomially
solvable, but $\mathcal{NP}$-Hard when 
each robot moves on its own line. 
They also showed that the BCP problems with 
a common line $L$, which is not fixed  
but part of the 
optimisation decision, 
can be solved by solving $O(n^2)$ instances
with a fixed line.

The problem 
of servicing timed requests on the line 
 corresponds 
to the  \textit{weighted} ball collecting problem,
so we will refer to it as BCP, or BCP$(k)$: 
a given number of $k$ robots, a single
given (fixed) line $L$,
and the objective of maximising 
the total weight. 
As shown in 
Asahiro \etal~\cite{DBLP:journals/dam/AsahiroHMOSY06}, 
for $k\geq 2$ 
the objective of maximising the total weight with 
$k$ robots can be modeled as a minimum cost flow problem 
in a directed acyclic graph ($DAG$)
$G^*$, which has $n$ nodes representing balls 
and two additional 
special source and sink nodes $s$ and $t$, respectively.
There is an edge in $G^*$ from a node $v'$,
representing ball $b'$, to a node $v''$, representing 
ball $b''$, if there is enough time for a robot 
to move from intercepting $b'$ to intercepting $b''$. An $s-t$ path in $G^*$ represents a schedule for one robot and its weight is equal to the total weight of the balls intercepted by the robot.

The corresponding minimum cost flow problem has 
unit node capacities (maximum one unit of flow 
through each node) and node weights equal to 
negations of the weights of balls,
so is equivalent to finding $k$ node-disjoint paths 
from $s$ to $t$
(the paths share only nodes $s$ and $t$)
such that the total weight of the selected paths 
is minimised. This problem 
can be solved in $O(kn^2)$ time by the successive shortest
path algorithm~\cite{Ahuja:1993:NFT:137406}. 
The quadratic dependence on $n$ is due to the fact 
that graph $G^*$ can have quadratic number of edges.
Looking into some technical details,
graph $G^*$ has actually $2n+2$ nodes since each node  
representing a ball 
is split into two nodes (connected by an edge) as
in the standard reduction from node capacities to 
edge capacities.

The $DAG$ $G^*$ can be implicitly represented
by a set ${\cal{P}}$ of $2n+2$ points 
in the 2-D Euclidean plane, illustrated
in Figure~\ref{fig1and2}.
The BCP input with $5$ requests given 
in Figure~\ref{fig1and2}(a) is shown in 
Figure~\ref{fig1and2}(b)
in the location-time coordinates
(the distances and times are normalised so that 
the maximum speed of a robot is equal to $1$).
The arrows show the edges of $G^*$. 
Vertex $t$, not shown in the diagram, 
is on the time axis
sufficiently high so that there are edges 
to $t$ from all other nodes.
For clarity, we also do not show the splitting of
nodes into two. For the \textit{unweighted} BCP and $k=1$, Asahiro \etal~\cite{DBLP:journals/dam/AsahiroHMOSY06} show an $O(n \log n)$-time algorithm
using the implicit plane representation of graph $G^*$ with the points in ${\cal{P}}$. However, for the weighted BCP and $k=1$, the problem is solved in the standard way of computing a longest path in a directed acyclic graph $G^*$, which requires $O(n^2)$ time.
For $k \ge 2$, 
\cite{DBLP:journals/dam/AsahiroHMOSY06} gives only 
the $O(kn^2)$ computation as indicated above, which applies 
to both \textit{unweighted} and \textit{weighted} BCP.

We show that the implicit plane
representation of graph $G^*$ can lead also to 
efficient algorithms for the {weighted} 
BCP for $k\ge 1$. More precisely, for the weighted BCP and the special case $k=1$ we show an iterative algorithm with running time of $O(n \log^3 n)$ which improves the previous bound of $O(n^2)$. For the weighted BCP and $k \geq 2$, we show a recursive algorithm for finding a minimum weight collection of $k$ node-disjoint $s-t$ paths in graph $G^*$ with the running time of $O(k^{3k}n\log^{2k+3} n)$, improving 
the previous bound of $O(kn^2)$ if $k$ is considered constant.  This result also gives an algorithm with the running time of  $O(k^{3k}n^3\log^{2k+3} n)$ for 
the BCP variant
where the placement of the line $L$ is to be chosen. A summary of the previous and new results for the $BCP$
 is shown in Table \ref{t1}. 

We also show properties of BCP solutions that may be useful
within the context of applications that motivate the problem.
Specifically, an $s-t$ path in $G^*$ is a schedule for one
robot in the BCP, and the representation of this path 
as a  concatenation of straight-line segments on the plane
({\it e.g.} path $(s,1,2,4)$ in Figure~\ref{fig1and2}(b))
gives the direction and the speed for each part of the
schedule.
If two paths on the plane cross, then the two robots
following these paths collide (are at the same point at the
same time). 
We show that for $k \geq 2$, there is at least one
minimum-weight
collection of $k$ node-disjoint non-crossing $s-t$ paths,
which ensures that the $k$ robots do not collide, and that
such a collection of optimal non-crossing paths can be
computed from any optimal collection of paths within 
$O(k n\log n)$ time.

\begin{table}[h]
\centering
\caption{(\textit{Weighted}) BCP: maximize the total weight}
\label{t1}
\begin{tabular}{cccccccccc}
\hline
\textbf{Line $L$} 
&  \textbf{$k= 1$} \cite{DBLP:journals/dam/AsahiroHMOSY06}
& \textbf{ $k=1$} 
& \textbf{$k\geq 2$} \cite{DBLP:journals/dam/AsahiroHMOSY06}
& \textbf{$k\geq 2$} \\ \hline
as part of input       & $O(n^2)$ & $O(n \log^3 n)$  & $O(kn^2)$     & $O(k^{3k}n\log^{2k+3} n)$  \\ 
as part of output      & $O(n^4)$  & $O(n^3 \log ^3 n)$  & $O(kn^4)$  & $O(k^{3k}n^3\log^{2k+3} n)$   \\ \hline
\end{tabular}
\end{table}

\begin{figure}
\centering
\begin{subfigure}{.5\textwidth}
\centering
 \includegraphics[scale=0.55]{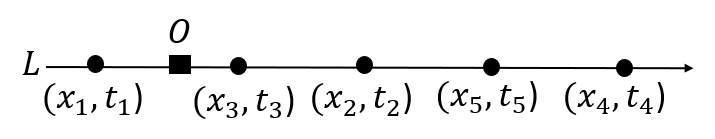}
    \caption{}
    \label{fig1}
\end{subfigure}%
\begin{subfigure}{.5\textwidth}
\centering
  \includegraphics[scale=0.62]{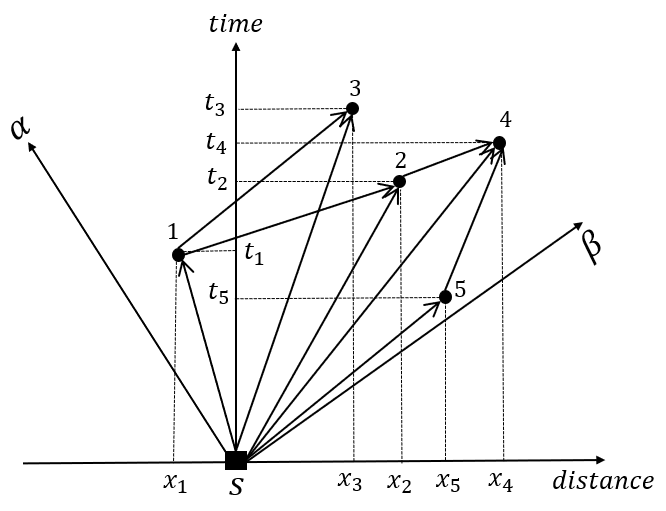}
    \caption{}
    \label{fig2}
\end{subfigure}
\caption{Figure \ref{fig1}: Five timed requests $1,2,3,4,5$ on the line $L$. Figure \ref{fig2}: Representation of the BCP input in the location-time coordinates and the $\alpha$-$\beta$
    coordinates. 
    }
    \label{fig1and2}
\end{figure}

The remaining part of the paper is organised in the following way. In section \ref{c3s2} we discuss the directed acyclic graph (DAG) model and its implicit planar representation. In section \ref{c3s2a} we describe the input and output of algorithm $\mathcal{A}_k$ and provide an overview of its recursive structure. In section \ref{k1sect} we consider the special case $k=2$ as an introduction to our recursive approach. In section \ref{knopi} we consider the general case $k \geq 3$. In section \ref{noncrossingpathsfork}  we show the additional property of BCP solutions which ensures the $k$ robots do not collide. 

 \section{Preliminaries}
\label{c3s2}

\subsection{DAG model of BCP}
\label{sub:dag}
The input of the BCP, as specified in \cite{DBLP:journals/dam/AsahiroHMOSY06}, consists of $n$ tuples $(x_1,y_1,v_1),..,(x_n,y_n,v_n)$ and two additional parameters $v$ and $k$. The parameter $k$ is the number of (identical) robots and $v$ specifies their maximum speed. The tuple $(x_i,y_i,v_i)$, for $1 \leq i \leq n$, specifies  speed $v_i$ of ball $b_i$ and the initial position $(x_i,y_i)$ in a $2D$ plane with $x$ and $y$ coordinates. We assume that $y_i \ge 0$. 
Starting at time $t=0$, ball $b_i \in B$ moves from 
$(x_i,y_i)$ with constant speed
of $v_i$ towards the $x$-axis, reaching the point $(x_i, 0)$ at time  $t_i={y_i}/{v_i}$.
A robot intercepts (or collects) ball $b_i$, if this 
robot is at time $t_i$ at point $(x_i, 0)$.
In this BCP model, to optimise the interception of balls, we need to know only the numbers $x_i$ and $t_i$, so we will 
assume that these numbers are given directly as the input. 

Notice that two or more balls can cross the line $L$ at the same time $t$ at the same distance $x$ from the origin. 
When a robot is at time $t$ at $x$, it can intercept all 
these balls. 
In the graph model, we assume that if $w$ balls 
are at the same time at the same place on the line, then they are represented by a single ball with weight $w$.
That is, the input to the problem is $n$ weighted timed requests $\{(x_i,t_i,w_i), 
1 \le i \le n\}$, where $x_i\in (-\infty, + \infty)$, 
$t_i \ge 0$ and $w_i \ge 0$ are the location 
(on the line), 
the time and the weight of request $i$, respectively.
We assume that the speed $v$ of the $k$ robots is equal to $1$ (this is achieved by dividing $x_i$ by $v$ for $i=1,2,\ldots,n$).

We model the input as a directed graph $G(V,E)$ with $n$ nodes $(1,2,..,n)$ representing the $n$ balls and two special nodes $s$ and $t$. 
For $1 \le i \le n$, $1 \le j \le n$, $i \neq j$, 
we have an edge $(i,j)$ in $G$, if and only if, $|x_j-x_i| \le t_j-t_i$ (recall that after normalising, $v=1$), which means that if a robot is at point $x_i$ at time $t_i$, having presumably just intercepted ball $b_i$, then it can arrive at point $x_j$ by time $t_j$ to intercept ball $b_j$.
For $1 \le i \le n$, we also have an edge $(s,i)$, if a robot starting at time $t=0$ from the origin $O$ of $L$ can reach point $x_i$ by time 
$t_i$ (to intercept ball $b_i$), and we have all edges 
$(i,t)$. 
Graph $G(V,E)$ is acyclic since an edge $(i,j)$ implies 
that $t_i < t_j$.
We assign weight $w_i$ to node $i$ and weight $0$ to nodes 
$s$ and $t$.
There are no edges $(s,i)$ for the balls $b_i$ which cannot be intercepted (because $x_i > t_i$). Such balls
can be removed from the input and they do not have to 
be included in graph $G$. We can therefore assume 
that graph $G$ has an edge $(s,i)$ for each ball~$b_i$. Figure \ref{Fig2} shows the directed acyclic graph $G$ constructed from the BCP input shown in Figure \ref{fig1}.

An $s-t$ path $(s, i_1, i_2, \ldots, i_p, t)$ in $G$ corresponds to a feasible movement
of one robot which intercepts balls $b_{i_1}, b_{i_2},  \ldots, b_{i_p}$, in this order.
The weight of this path (the sum of the weights of the nodes on this path) is equal to the total weight of the intercepted balls.
Consequently, we can find a schedule for one robot that maximizes the number of intercepted balls by finding the maximum weight path from $s$ to $t$ in the directed acyclic graph $G$. This can be done in the standard way by negating the weights and move from node weights to edge weights. That is, we construct graph $G'(V',E')$ with the same set of nodes and edges, such that the weight $w'(i,j)$ of an edge $(i,j) \in E'$ is equal to $w'(i,j)=-(w(i,j)+w_j)$. Finding the maximum weight path from $s$ to $t$ in $G$ is equivalent to finding a shortest (that is, minimum weight) path from $s$ to $t$ in $G'$.
Since graph $G$ (and subsequently $G'$) may have $\Theta(n^2)$ edges in the worst 
case, without referring to a special structure of~$G$, we can only conclude that such a path can be computed in $O(n^2)$ time.

\begin{figure}
\centering
    \includegraphics[bb=0 0 500 200,scale=0.7]{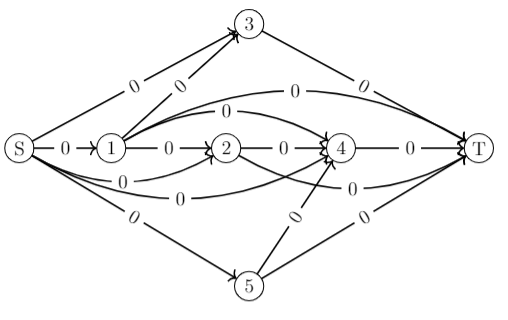}
    \caption{The Directed Acyclic Graph corresponds of the BCP input shown Figure \ref{fig1} for the appropriate values of $x_i,t_i$. The weight of each node $i \in V \setminus{\{s,t\}}$ is equal to $w_i$.}
    \label{Fig2}
\end{figure}

For $k \geq 2$ the problem asks for $k$ node-disjoint paths from $s$ to $t$ in $G$ such that the total weight of the selected paths is minimized. The condition of node-disjoint paths refers to the internal nodes and ensures that no intercepted ball is counted twice.
We change from node-disjoint paths to edge-disjoint paths
in the standard way by considering the following
modified graph $G^{*}$ obtained from $G$ by splitting 
nodes, as explained below.

Every node $i \in V\setminus\big\{s,t\big\}$ in $G$ is represented in $G^*$ by two nodes $i^{-},i^{+}$ connected by a \textit{short} edge from $i^{-}$ to $i^{+}$. 
The set $V^*$ of nodes in $G^*$ includes also nodes $s$
and $t$. Each edge $(i,j)$ in $G$, where $i,j\in V \setminus (s,t)$ is replaced by a \textit{long} edge $(i^{+},j^{-})$. Each edge $(s,i)$ is replaced by a \textit{long} edge $(s,i^-)$ and each edge $(i,t)$ is replaced by a \textit{long} edge $(i^+,t)$ for $i=1,2,\ldots,n$. We note that two paths from $s$ to $t$ share a node $i$ in $G$, if, and only if, the corresponding paths in $G^*$ share edge $(i^{-},i^{+})$. 
The weights are moved from nodes in $G$ onto the corresponding short edges in $G^*$: the weight of every \textit{short }edge $e$ connecting nodes $i^{-},i^{+}, \forall i \in V^{*}\setminus \{s,t\}$ is equal to $-w_i$. The weight of each \textit{long} edge is equal to zero. The capacity of every edge (long and short) is set equal to $1$.
Figure \ref{mcf} illustrates the obtained directed acyclic graph $G^{*}$  by applying the transformation described above to graph $G$ of Figure \ref{Fig2}. 

Each collection of $k$ node disjoint $s-t$ paths in $G$ 
corresponds in a natural way to a collection of $k$ edge-disjoint $s-t$ paths in $G^*$, with corresponding paths
having the same weight.
Finding $k$ edge-disjoint $s-t$ paths with the minimum 
total weight is equivalent to finding a minimum-cost
flow with source $s$, destination $t$ and demand $k$,
assuming unit edge capacities.  A collection of $k$ node-disjoint $s-t$ paths in $G^*$ 
(which must be also edge-disjoint) with minimum total weight gives an optimal schedule for $k$ robots in the BCP.

A collection of $k$ node-disjoint paths from $s$ to $t$ in $G^*$ with the minimum total weight can be found in $O(kn^2)$ time by the {\em successive shortest path algorithm}~\cite{Ahuja:1993:NFT:137406}.
For $k=1$ and the unweighted BCP, that is, when we are looking for a shortest 
$s-t$ path in $G$ and all nodes have weight equal to $1$, 
it is shown in~\cite{DBLP:journals/dam/AsahiroHMOSY06} that such a path can be found in $O(n \log n)$ time 
by using the special geometric representation of graph $G$
described in Section~\ref{sub:plane}. For $k=1$ and the weighted BCP  computing a shortest 
$s-t$ path in $G$ requires $O(n^2)$ time.
For $k\ge 2$, \cite{DBLP:journals/dam/AsahiroHMOSY06} gives only 
the straightforward $O(kn^2)$ computation as indicated above which applies both to the weighted and unweighted case.
The main contributions of our work is that the geometric
representation of graph $G$ can lead also to 
efficient algorithms for the weighted BCP and $k\ge 1$.

\begin{figure}
    \centering
    \includegraphics[bb=0 0 500 200,scale=0.7]{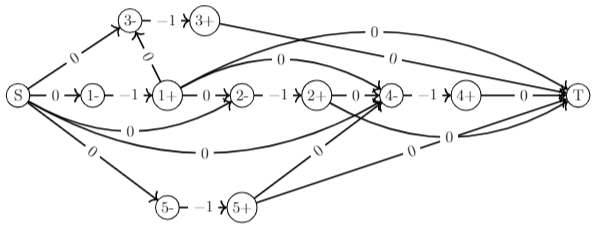}
    \caption{Directed Acyclic Graph  $G^{*}$ with $2n+2$ nodes. The weight of a short edge $(x^-,x^+)$ is equal to $-w_x$ (for clarity we assume that $w_x=1$ $ \forall x \in V^*$). The weight of a long edge $(x,y)$ is equal to $0$. }
    \label{mcf}
\end{figure}

\subsection{The plane representation of DAG $G^*$}
\label{sub:plane}
It was shown in \cite{DBLP:journals/dam/AsahiroHMOSY06} that the directed acyclic graph $G$ (and graph $G^*$) can be implicitly represented with a set of points ${\cal{P}}$ on the Euclidean 2-D plane. There are $2n+2$ points in ${\mathcal{P}}$ which correspond to the 
nodes in graph $G^*$.
There are two special points $s$ and $t$ 
which correspond to the special nodes in $G^*$. The remaining $2n$ "regular" points can be seen  as $n$ pairs of points $(1^-,1^+),...,(n^-,n^+)$. For $i=1,2,...,n$, a pair of points $(i^-,i^+)$ in $\cal{P}$ corresponds to pair of nodes $(i^-,i^+)$ in $G^*$ and therefore corresponds to 
ball $b_i$ of the BCP input.
Because of this correspondence, we will use the terms
ball, node and point (in ${\mathcal{P}}$) interchangeably
(remembering that the special nodes/points $s$ and $t$ 
do not correspond to any ball). 
The placement of a pair of points $(i^-,i^+) \in \cal{P}$ in the $2$-D plane is described with coordinates $\alpha$ and $\beta$ defined in the following way: \begin{itemize}
   \item $\alpha_i = t_i + x_i \quad$  and  $\quad    \beta_i= t_i - x_i$
    
    \item $\alpha_i^+ = \alpha_i, \quad \beta_i^+= \beta_i \quad  $  and  $\quad    \alpha_i^- =\alpha^+_i - \epsilon, \quad   \beta_i^-=\beta_i^+ - \epsilon$
\end{itemize}
where $\epsilon$ is an arbitrary small number to ensure that $i^-$ and $i^+$ are  sufficiently "close" to each other such that there is no point $j\neq i$ satisfying $\alpha_i^- \leq \alpha_j \leq \alpha_i^+$ or $\beta_i^- \leq \beta_j \leq \beta_i^+$. This transformation essentially consists of rotating the location-time coordinates by $45^{\circ}$ to the new system of $\alpha$-$\beta$
 coordinates -- see Figure~\ref{fig1and2}(b).

To simplify matters we will refer to pair of points $(i^-,i^+)$ by simply referring to point $i$.  Figure \ref{Figure3b} shows the $\alpha-\beta$ planar representation of the directed acyclic graph $G^*$ shown in Figure \ref{mcf} (for clarity we do not show the splitting of points into two).  The $\alpha-\beta$ planar representation (which can be constructed in $\Theta(n)$ time), implicitly represents the directed acyclic graphs $G$ and $G^*$.
    There is an edge in $G$ from node $i$ to node $j$, if, and only if, a robot can intercept 
ball $b_j$ after intercepting ball $b_i$. This means that $t_i-t_j\geq |x_i-x_j|$,  which is equivalent to having $\alpha_i \geq \alpha_j$ and $\beta_i \geq \beta_j$.

\begin{figure}
\centering
\begin{subfigure}{.5\textwidth}
  \centering
  \includegraphics[scale=0.42]{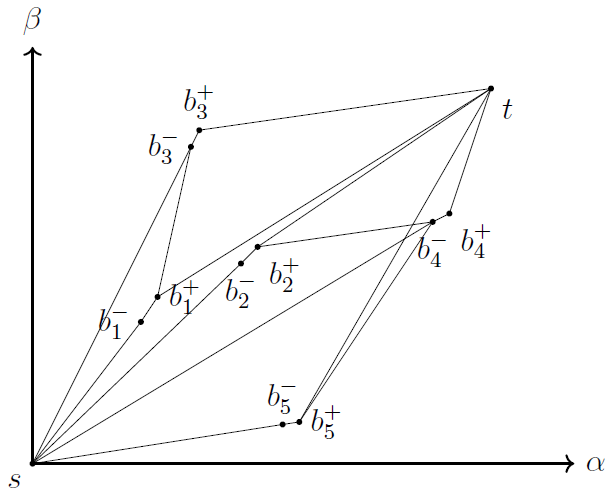}
    \caption{}
    \label{Figure3b}
\end{subfigure}%
\begin{subfigure}{.5\textwidth}
  \centering
  \includegraphics[scale=0.42]{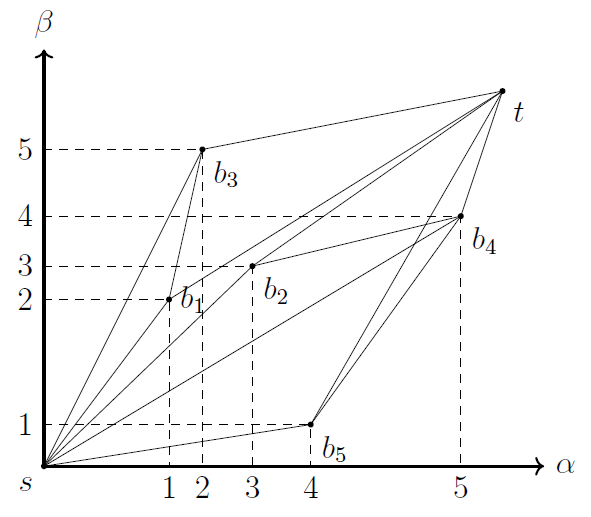}
    \caption{}
    \label{Figure3}
\end{subfigure}
\caption{Figure \ref{Figure3b} shows the implicit plane representation of $G^*$ in the $\alpha-\beta$ coordinate system. Figure \ref{Figure3} shows the transformation of the input such that all points have distinct integer $\alpha$ and $\beta$ coordinates.}
\end{figure}

We want the set of points $\cal{P}$ to represent
correctly the topology (the edges) of graph $G^*$, 
but otherwise the values of the 
coordinates of the points in $\cal{P}$ are not important.
We can therefore assume that $s=(0,0)$, $t = (n+1,n+1)$, each
regular point in $\cal{P}$ has integral coordinates 
and any two distinct points in $\cal{P}$ have 
both coordinates distinct. This can be achieved by sorting the $\alpha$ and $\beta$ coordinates of all  points in $\cal{P}$ and setting the value of the $m^{th}$ smallest $\alpha$ (resp. $\beta$) coordinate equal to $m$, where $m=1,2,..,n$.

Notice that two points $i$ and $j$ can not have the same $\alpha$ and $\beta$ coordinate because this implies that balls $b_i$ and $b_j$ cross the line $L$ at the same time at the distance from the origin and thus $b_i$ and $b_j$ correspond to the same point $i$. It is possible however for two points $i$ and $j$ in $\cal{P}$ to share the same $\alpha$ or $\beta$ coordinate. If for two points $i$ and $j$ in $\cal{P}$ we have $\alpha_i=\alpha_j$  (resp. $\beta_i=\beta_j$) but $\beta_i \leq \beta_j$ (resp. $\alpha_i \leq \alpha_j$)  then to ensure that $\cal{P}$ represents correctly the topology (the edges) of graph $G^*$ we set the value of $\alpha_i$ equal to $m$ and the value of $\alpha_j$ equal to $m+1$. Figure \ref{Figure3} illustrates an example of the replacement of the $\alpha$ and $\beta$ coordinates with integral values(for clarity we do not show the splitting of points into two).

For two points $i$ and $j$ in the plane,
with coordinates $(\alpha_i,\beta_i)$ 
and $(\alpha_j,\beta_j)$, respectively,
we write $i \prec j$ to denote that point $j$ dominates 
point $i$ in the sense that $i \neq j$,
$\alpha_i \le \alpha_j$ and $\beta_i \le \beta_j$.
We write $i \bowtie j$ to denote that points $i$ and $j$
are distinct and neither $i \prec j$ nor $j \prec i$.
We have $(0,0) = s\prec t$ and for 
each regular point $i$ in ${\cal{P}}$,
$s\prec i\prec t$.
Thus for any two points $i$ and $j$ in $\cal{P}$
(regular or special) $(i,j)$ 
is an edge in $G^*$ if, and only if, 
$i\prec j$.

\section{Algorithm $\mathcal{A}_k$}
\label{c3s2a}

\subsection{Input and Output}
\label{sec2:mincostflow}
Consider the implicit representation of the directed acyclic graph $G^*$ with the points in $\mathcal{P}$.
The edges of $G^*$  are represented by straight-line 
segments in the $\alpha$-$\beta$ plane. 

\begin{definition}
We say that two node-disjoint edges $(u,v)$ and $(x,y)$ 
in $G^*$ \emph{cross}, if the two (closed) 
segments $[u,v]$ and $[x,y]$ in the plane 
have a common point.
\end{definition}

Recall that we can assume w.l.o.g.\
that the $\alpha$ coordinates and the $\beta$ coordinates
of the $n$ nodes are distinct integers in $[1,n]$ (see
subsection~\ref{sub:plane}).
We also assume that all points are in
general position (the reduction to achieve this 
requires increasing the range of the integer
coordinates). Therefore, if two edges node-disjoint edges $(u,v)$ and $(x,y)$ 
in $G^*$ cross then the common point of the closed segments does not correspond to a point in $\mathcal{P}$. We say that two paths $Q$ and $Q'$ in $G^*$ \emph{cross}, 
if there is an edge $(u,v) \in Q$ crossing with an edge $(x,y) \in Q'$. We say that a path $Q$ in $G^*$ is \textit{non-self-crossing} if $Q$ does not traverse two edges that cross.

To provide an overview of our algorithm in the context of the minimum cost flow problem, 
we denote by $\mathcal{N}$ the flow network 
based on graph $G^*$, as discussed in subsection \ref{sub:dag}, with negative node weights (\textit{i.e.} weights of the short edges)
and all edges (short and long) having unit capacities.
For network $\mathcal{N}$ and $k-1$ node-disjoint $s$-$t$
paths $\mathcal{Y}_1, \mathcal{Y}_2,$ 
$\ldots,$ $\mathcal{Y}_{k-1}$ in $\mathcal{N}$, which
represent 
an integral flow of value $k-1$ in $\mathcal{N}$,
we define the residual network 
${\mathcal{N}}_{k-1}$ in the usual way, by 
reversing the edges of
 the paths $\mathcal{Y}_1, \mathcal{Y}_2, \ldots,$ $\mathcal{Y}_{k-1}$.
The base case is $\mathcal{N}_0 \equiv \mathcal{N}$.

We show an algorithm $\mathcal{A}_{k}$ which 
for an input 
$(\mathcal{N}; \mathcal{Y}_1, \mathcal{Y}_2,
\ldots,\mathcal{Y}_{k-1})$, where
$\mathcal{Y}_1, \mathcal{Y}_2,$ $\ldots,$ $\mathcal{Y}_{k-1}$
are $k-1$ node-disjoint \emph{non-crossing} $s$-$t$ paths 
minimizing the total weight of any collection of 
$k-1$ \emph{node-disjoint} $s$-$t$ paths,
computes 
a shortest path tree $T^*$ rooted at $s$ in the residual 
network $\mathcal{N}_{k-1}$. The paths $\mathcal{Y}_1, \mathcal{Y}_2,$ 
$\ldots,$ $\mathcal{Y}_{k-1}$ in $\mathcal{N}$   are given in a left-right order in their plane representation.  We maintain two global arrays $h$ and $predecessor$, which are 
indexed by the points $\mathcal{P}\setminus\{s\}$.
At the end of the computation, 
for each $x\in \mathcal{P}\setminus\{s\}$,
the values $h(x)$ and $predecessor(x)$
should be the shortest path weight from $s$ to $x$ and 
the predecessor of point $x$ in the tree $T^*$. 

When $k=2$, we have only one path $\mathcal{Y}_1$, so the condition that paths
$\mathcal{Y}_1, \mathcal{Y}_2,$ $\ldots,$ $\mathcal{Y}_{k-1}$
are \emph{non-crossing} is trivially satisfied.
For subsequent values of $k$, this condition will be 
ensured inductively.
From now on, when we refer to paths $\mathcal{Y}_1, \mathcal{Y}_2,$ $\ldots,$ 
$\mathcal{Y}_{k-1}$, we assume that they are non-crossing
paths representing a minimum-cost flow value of $k-1$. The paths $\mathcal{Y}_1, \mathcal{Y}_2,$ $\ldots,$ $\mathcal{Y}_{k-1}$ in network $\mathcal{N}$ and
the computed $s$-$t$ shortest path in the residual network
$\mathcal{N}_{k-1}$ give in the usual way 
a minimum-cost flow of value $k$ in $\mathcal{N}$.
This flow is represented by $k$ node-disjoint 
$s$-$t$ paths $Y_1,Y_2,..,Y_k$ in $\mathcal{N}$, 
which are not necessarily non-crossing.
Let $\mathcal{P}_k \subseteq \mathcal{P}$ the set of all
points covered by paths $Y_1,Y_2,..,Y_k$. 
The following theorem states that
a valid input for algorithm $\mathcal{A}_{k+1}$ exists
and can be computed in an efficient way.
\begin{theorem}
\label{maintheoremuncrossing}
Given a point set $\mathcal{P}_k$ such that all points $\mathcal{P}_k$ can be covered with $k$ paths, there is an $O(kn\log n)$ algorithm $\widetilde U_k$ which computes a collection of $k$ node-disjoint non-crossing $s$-$t$ paths $\mathcal{Y}_1,\mathcal{Y}_2,\ldots,\mathcal{Y}_k$ covering all points in $\mathcal{P}_k$.
\end{theorem}
 
 The proof of Theorem \ref{maintheoremuncrossing} is given separately in Section \ref{noncrossingpathsfork}. Starting with the network $\mathcal{N}$, 
we compute a minimum-cost integral flow of value $k$
in $\mathcal{N}$,
which gives a solution for BCP,
by iterating algorithm $\mathcal{A}_i$ followed by 
    algorithm $\widetilde{U}_i$, for $i = 1,2, \ldots, k$. Algorithm $\mathcal{A}_k$ is an instance of the \emph{relaxation technique} for 
the single-source shortest paths 
problem~\cite{Cormen2001introduction} 
in the residual network $\mathcal{N}_{k-1}$. Arrays  $h$ 
and $predecessor$ are only updated by
the following $\relaxop(y,x)$ operation, 
where $(y,x)$ is an edge and $w(x)$ is the weight of node
$x$: 
If $h(x) > h(y) + w(x)$, then
$h(x) \leftarrow h(y) + w(x)$ and 
$predecessor(x) \leftarrow y$. 

In algorithm $\mathcal{A}_k$ $\relaxop$
operations occur in groups: 
$\Relaxop(x) \equiv 
\{\relaxop(y,x): y \prec x\}$. The detailed description of operation $\Relaxop$ and its implementation details are given in Subsection \ref{section:arrays}.
For the worst-case running-time efficiency, 
we implement operation $\Relaxop(x)$ not by 
performing explicitly all operations $\relaxop(y,x)$ (this would take $O(n)$ time), 
but by finding a point $z\prec x$ such 
that $h(z) =\min\{ h(y): y\prec x\}$ and 
performing only $\relaxop(z,x)$.
Finding point $z$ takes $O(\log^3(n))$ time 
using a data structure
introduced for the two-dimensional orthogonal-search problem\cite{Willard:1985:NDS:3674.3690}. 

\subsection{Overview of Algorithm $\mathcal{A}_k$}
In this section we give an overview of algorithm $\mathcal{A}_k$ for $k\geq 3$. The detailed description and the analysis of algorithm $\mathcal{A}_k$ for
$k  \geq 3$ is given in section \ref{knopi}. The analysis of algorithm $\mathcal{A}_k$ for the special case  
$k = 2$ is given separately in section \ref{k1sect}. 
This special case does not   refer to
some of the elaborations of the general case, so
the arguments are simpler and shorter, and can be treated 
as preliminaries to the general case.

To facilitate the recursive structure of algorithm
${\mathcal{A}_{k}}$, we extend the input specification
to a sub-network
of $\mathcal{N}_{k-1}$
induced by the points in $\mathcal{P}$
with the $\beta$ coordinates in the interval 
$(\beta_1, \beta_2]$, 
for given $\beta_1 < \beta_2$. 
We denote this sub-network by 
$(\mathcal{N}; 
\mathcal{Y}_1, \mathcal{Y}_2,
\ldots,\mathcal{Y}_{k-1})[\beta_1, \beta_2]$,
or $\mathcal{N}[\beta_1, \beta_2]$ for short.
The initial input, that is, the input to the initial 
call to algorithm ${\mathcal{A}_{k}}$, is the  whole residual network $\mathcal{N}_{k-1}$  which is defined 
by the interval $(0, n+1]$.

Arrays $h$ and $predecessor$ are global
and initialized outside of the computation of 
algorithm~${\mathcal{A}_{k}}$
(details of this global initialisation are 
in subsection \ref{section:arrays}).
The subsequent recursive calls to~${\mathcal{A}_{k}}$
continue from the current state of these arrays, without
re-initialising.
More precisely,
when algorithm $\mathcal{A}_k$ is applied to a sub-network
$\mathcal{N}' = \mathcal{N}[\beta_1,\beta_2]$
(a recursive call), then the computation starts 
with each point $x$ in $\mathcal{N}'$ 
having some value $h(x) \leq 0$, and array $predecessor$ 
restricted to $\mathcal{N}'$
representing a forest in $\mathcal{N}'$. 
At the end of the computation,
for each point $x$ in the sub-network,
$h(x)$ is equal to the weight of some path  
to $x$, hopefully smaller than its starting 
value, and array $predecessor$  
represents a new forest.

A call to algorithm ${\mathcal{A}_{k}}$ for a sub-network $\mathcal{N}[\beta_1, \beta_2]$ 
includes two recursive calls to  
${\mathcal{A}_{k}}$
applied to sub-networks 
$\mathcal{N}[\beta_1,(\beta_1+\beta_2)/2]$ and
$\mathcal{N}[(\beta_1+\beta_2)/2,\beta_2]$.
We denote by $\mathcal{N}_1$ and $\mathcal{N}_2$
these two sub-networks, respectively, or the sets of 
nodes (points) in these sub-networks, 
depending on the context.
The  sub-network $\mathcal{N}_1$ (the lower half)  has $\lceil 2n/2 \rceil$ 
 points from $\mathcal{P}$ and the 
sub-network $\mathcal{N}_2$ (the upper half) has $\lfloor 2n/2 \rfloor$  points 
from $\mathcal{P}$. 

The base case of the recursion are sub-problems of size smaller than some constant threshold. 
Algorithm $\mathcal{A}_k$ also includes a \textit{coordination} phase which takes place between the two recursive calls and consists of calling a coordination algorithm $\mathcal{C}_k$ on $\mathcal{N}[\beta_1,\beta_2]$. 

While the recursive calls to algorithm $\mathcal{A}_k$ on $\mathcal{N}_1$ and $\mathcal{N}_2$ consider paths 
which are wholly either in $\mathcal{N}_1$ or $\mathcal{N}_2$,  
the coordination algorithm $\mathcal{C}_k$ is responsible for considering paths which have points both in $\mathcal{N}_1$ and $\mathcal{N}_2$. Putting everything together, when algorithm $\mathcal{A}_k$ is applied to a sub-network $\mathcal{N}[\beta_1,\beta_2]$ the computation consists of three phases. The first phase is the recursive call of algorithm $\mathcal{A}_k$ to sub-network $\mathcal{N}_1$, the second phase is the coordination of $\mathcal{N}_1$ and $\mathcal{N}_2$ by algorithm $\mathcal{C}_k$ and the third phase is the recursive call of algorithm $\mathcal{A}_k$ to $\mathcal{N}_2$.

\begin{definition}
We say that the computation of a shortest-path algorithm, 
or a part of such algorithm,
follows a given path 
$Q = (x_0,x_1,\ldots,x_m)$,
if the computation includes all relax operations $\relaxop(x_i,x_{i+1})$,
$i = 1,2,\ldots,m-1$ in this order.
\end{definition}
Note that each operation $\relaxop(x_i,x_{i+1})$
may be implicitly included within operation
$\Relaxop(x_{i+1})$. Recall that for a sub-network $\mathcal{N}[\beta_1,\beta_2]$ a path $Q$ in the sub-network is non-self-crossing if $Q$ does not traverse two edges that cross.

\begin{definition}
For a sub-network 
$(\mathcal{N}; \mathcal{Y}_1, \mathcal{Y}_2,
\ldots,\mathcal{Y}_{k-1})[\beta_1,\beta_2]$
and a point $v$ in this sub-network, we define 
path $\widetilde Q_v$ as the minimum weight path 
among all non-self-crossing paths in this sub-network which end at $v$.
We denote by $\widetilde h(v)$ 
the weight of path $\widetilde Q_v$.
\end{definition}

The following theorem describes the specification
of algorithm~$\mathcal{A}_k$.
\begin{theorem}\label{Thm1noPi}
When algorithm~$\mathcal{A}_k$ is applied 
to a sub-network 
$(\mathcal{N}; \mathcal{Y}_1, \mathcal{Y}_2,
\ldots,\mathcal{Y}_{k-1})[\beta_1,\beta_2]$, the computation follows every non-self-crossing path in this sub-network and the running time is $O(k^{3k}n \log^{2k+3} n)$ where $n$ is the size of the sub-network.
\end{theorem}

If the computation follows a path $Q = (x_0,x_1,\ldots,x_m)$, then 
at the end of this computation, the computed shortest path weight $h(x_m)$ is at most the weight of $Q$.  Theorem \ref{Thm1noPi} implies the following corollary.
\begin{corollary}
\label{Thm1CornoPi}
When the call of algorithm $\mathcal{A}_{k-1}$ on a sub-network $(\mathcal{N}; \mathcal{Y}_1, \mathcal{Y}_2,
\ldots,\mathcal{Y}_{k-1})[\beta_1,\beta_2]$ terminates for every point $v$ in the sub-network we have $h(v) \leq \widetilde h(v)$.
\end{corollary}

 The proof of Theorem \ref{Thm1noPi} consists of
showing that the computation of 
$\mathcal{A}_k(\,\mathcal{N}[\beta_1,\beta_2]\,)$ 
follows every non-self-crossing path in
$\mathcal{N}[\beta_1,\beta_2]$.
First we analyse the combinatorial structure of a non-self crossing path $Q$ by considering its geometric representation on the $\alpha-\beta$ plane and then we show how the consecutive computational phases of 
algorithm $\mathcal{A}_k$ follow the consecutive 
sections of path $Q$. 

To use Theorem~\ref{Thm1noPi} to conclude that algorithm $\mathcal{A}_k$ applied to the whole residual network
$\mathcal{N}_{k-1}$ is
correct, 
that is, that the computed tree is indeed a shortest path tree in $\mathcal{N}_{k-1}$,
we need Theorem \ref{mainThm2nopi} (given below) which 
asserts that there are non-self-crossing shortest paths in $\mathcal{N}_{k-1}$. To simplify the presentation of a non-self-crossing path followed  by algorithm $\mathcal{A}_k$ (not necessarily a shortest path) we  distinguish between \emph{red} and \emph{black} points and edges.

The {red} points and {red} edges are the points and edges on the paths $\mathcal{Y}_1, \mathcal{Y}_2, \ldots,\mathcal{Y}_{k-1}$. All other points and edges are black. Recall that in the residual network $\mathcal{N}_{k-1}$, the red edges (short and long) of the paths  $\mathcal{Y}_1, \mathcal{Y}_2, \ldots,\mathcal{Y}_{k-1}$ have reversed direction and negated weights, as in the standard way. That is, a long red edge $(u,v)\in \mathcal{Y}_j$ where $j \in [1,k-1]$ such that $u \prec v$ has weight equal to $0$ and reversed direction from $v$ to $u$. A short edge $(u^-,u^+)$ has direction from $u^+$ to $u^-$ and weight equal to $w_u$. 

\begin{theorem}\label{mainThm2nopi}
For a sub-network $\mathcal{N}[\beta_1,\beta_2]$,
there exists 
a non-self-crossing shortest path
to every  point $v$ in the sub-network.
\end{theorem}
\begin{proof}
For a sub-network $\mathcal{N}[\beta_1,\beta_2]$ consider a point $v$ in this sub-network. Among all shortest paths to point $v$ let $Q^*$ be the shortest path with the minimum number of edges. We claim that $Q^*$ is non-self-crossing. Assume towards contradiction that $Q^*$ is self-crossing.

Recall that each point $x$ is a pair of points $(x^-,x^+)$ connected with a short edge of capacity $1$. We denote by $h(x)$ the weight of the sub-path of $Q^*$ to point $x^+$ and by $h^-(x)$ the weight of the sub-path of $Q^*$ to point $x^-$. Notice that if $x$ is a black point then the path to $x^+$ must traverse the short residual edge $(x^-,x^+)$ with weight $-w_x<0$. Therefore we have that $h(x)=h^-(x)-w_x$. If $x$ is a red point then the short edge $(x^+,x^-)$ is not residual and its weight is equal to $w_x>0$, which means that the path to $x^+$ can not traverse the short red edge $(x^+,x^-)$ and therefore we have $h(x) \leq h^-(x) \leq h(x)+w(x^+,x^-)$.

It is easy to see that if $Q^*$ is self-crossing then it must traverse at least on red edge of a path $\mathcal{Y}_j$ where $j \in [1,k-1]$.  Notice that since paths $\mathcal{Y}_1,\mathcal{Y}_2,\ldots,\mathcal{Y}_{k-1}$ are non-crossing pairwise, if $Q^*$ is self-crossing then either $Q^*$ has two black edges $(u,v)$ and $(x,y)$ that cross (see Figure \ref{nonselfcrossingtheorem1}) or a black edge $(u,v)$ crossing with a red edge $(y,x)$ (see Figure \ref{nonselfcrossingtheorem2}).

Without loss of generality, we assume that edge $(u,v)$ appears before edge $(x,y)$ in $Q^*$. Let $\pi$ be the crossing point of edge $(u,v)$ and edge $(x,y)$. Since $\pi$ is a point on the closed segment $[u,v]$ of the black edge $(u,v)$ we have that $u \prec \pi$. Similarly, since $\pi$ is a point on the closed segment $[x,y]$ of the black edge $(x,y)$ (resp. red edge $(y,x)$)  we have that $\pi \prec y$. Thus, we conclude that $u \prec y$. Symmetrically, we obtain that $x \prec v$.

 We first claim that $h(u)=h^-(y)$. If $h(u)<h^-(y)$, then consider the path $Q_u^ \cup \{(u^+,y^-)\}$ to point $y^-$ where $Q_{u}$ is the sub-path of $Q^*$ to point $u^+$. The weight of the long edge $(u^+,y^-)$ is equal to zero and since $h(u)<h^-(y)$, the weight of path $Q_{u}\cup \{(u^+,y^-)\}$ is smaller than the weight of path $Q_{y}$ where $Q_y$ is the sub-path of $Q^*$ to point $y^-$.  However, this makes a contradiction that $Q_{y}$ is a shortest path to point $y^-$. If $h(u)>h^-(y)$, then we obtain that $h^-(v)>h(x)$ since the weight of the long edges $(u^+,v^-)$ and $(x^+,y^-)$ is equal to zero. Consider the cycle $C=Q_{vx} \cup \{(x^+,v^-)\}$ where $Q_{vx}$ is the sub-path of $Q^*$ from $v^-$ to $x^+$. If $h^-(v)>h(x)$ then the total weight of cycle $C$ is negative. This makes a contradiction since there are no negative cycles in the residual network.

Let $m^*$ be the number of edges in path $Q^*$. Consider the decomposition of $Q^*$ into $Q_u \cup Q_{uy} \cup Q_{y}$ where $Q_u$ is the sub-path of $Q^*$ from its starting point to point $u^+$, $Q_{uy}$ is the sub-path of $Q^*$ from $u^+$ to $y^-$ and $Q_{y}$ is the sub-path of $Q^*$ from $y^-$ to $v$. Consider the path $Q=Q_u \cup \{(u^+,y^-)\} \cup Q_{y}$ and let $m$ be the number of edges in path $Q$. Path $Q$ has the same weight as $Q^*$ since $h(u)=h^-(y)$. Further, $m < m^*$ since the sub-path $Q_{uy}$ consists of at least two edges. However, this makes a contradiction since  $Q^*$ is chosen as the shortest path to $v$ with the minimum number of edges.
\end{proof}

 \begin{figure}
\centering
\begin{subfigure}{.5\textwidth}
  \centering
  \includegraphics[scale=0.5]{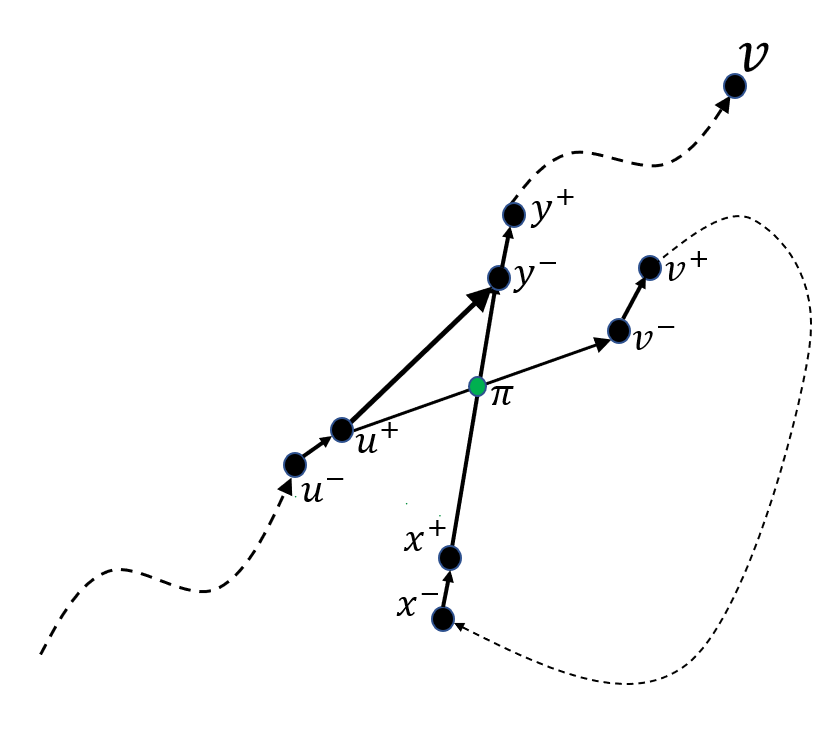}
    \caption{}
    \label{nonselfcrossingtheorem1}
\end{subfigure}%
\begin{subfigure}{.5\textwidth}
  \centering
  \includegraphics[scale=0.5]{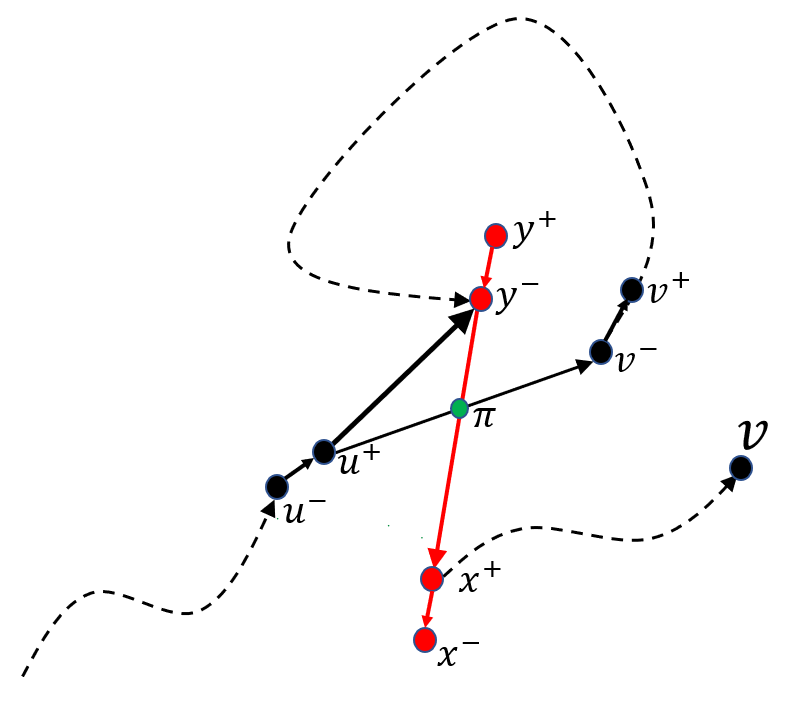}
    \caption{}
    \label{nonselfcrossingtheorem2}
\end{subfigure}
\caption{Figures \ref{nonselfcrossingtheorem1} and \ref{nonselfcrossingtheorem2}: The schematic representation of the proof for Theorem \ref{mainThm2nopi}.}
\end{figure}

Consider the computation of 
$\mathcal{A}_k(\mathcal{N}_{k-1})$, that is, the initial call of algorithm $\mathcal{A}_{k}$ to the whole residual network $\mathcal{N}_{k-1}$.
The array $predecessor$ is initialised to some tree 
rooted at $s$. The details of this initialisation are given in subsection \ref{section:arrays}. Array $predecessor$ is updated only by 
the $\relaxop$ operation. Therefore, by the general properties of the shortest-paths relaxation technique, 
since there are no negative cycles in $\mathcal{N}_{k-1}$, the array $predecessor$ always represents some tree. 
Let $h$ and $predecessor$ be the arrays when the computation terminates. 
From Corollary~\ref{Thm1CornoPi},
for every point $v$ in $\mathcal{N}_{k-1}$, 
we have $h(v) \leq \widetilde h(v)$.
From Theorem \ref{mainThm2nopi},  
$\widetilde h(v) = h^*(v)$, 
where $h^*(v)$ is the weight of a shortest path 
from $s$ to $v$ in $\mathcal{N}_{k-1}$. 
Therefore we have that 
$h^*(v) \leq  h(v) \leq \widetilde h(v) = h^*(v)$, 
so $h(v) = h^*(v)$. Thus, the computed tree must be a shortest-path tree (from the general properties of 
the relaxation technique: $h(v)$ is never smaller than 
the weight of the current tree path from $s$ to $v$).

For the special case $k=1$, in Section \ref{k1sect} we show an iterative algorithm $\mathcal{A}_1$ with the running time of
$O(n \log^3 n)$ which considers all points 
in topological order (two points $x$ and $x'$ are in topological order if $x \prec x'$) and for each point $x$ performs operation $\Relaxop(x)$. Algorithm $\mathcal{A}_1$ essentially implements the standard methodology\footnote{For any directed acyclic graph (DAG) $G$, a shortest path between two points in $G$ can be computed by traversing the nodes in topological order and for each node $v$ perform operation relax in all edges outgoing from $v$.}  of computing a shortest path in a directed acyclic graph, but accounts for incoming edges (instead of outgoing edges) using operation $\Relaxop(x)$.  A topological order of the points can be found in $O(n \log n)$-time as shown in \cite{DBLP:journals/dam/AsahiroHMOSY06}.
Operation $\Relaxop(x)$ takes $O(\log^3(n))$ amortized time 
using a data structure for orthogonal-search queries \cite{Willard:1985:NDS:3674.3690}. 

For $k \geq 2$ the proof of Theorem \ref{Thm1noPi} is outlined below. For some $\beta_1,\beta_2$ such that $0 \leq \beta_1 \leq \beta_2 \leq n+1$ consider a sub-network $\mathcal{N}[\beta_1,\beta_2]$ and let $Q$ be a non-self-crossing path in this sub-network. Without loss of generality, we assume that $Q$ has points both in $\mathcal{N}_1$ and $\mathcal{N}_2$. 
\begin{definition}
\label{Def:x}
If path $Q$ starts in $\mathcal{N}_1$ then we define  $x \in \mathcal{N}_1$ to be the last point in $ Q$ such that all points before $x$ are in $\mathcal{N}_1$. If path $Q$ starts in $\mathcal{N}_2$ we define $x \in \mathcal{N}_2$ to be the starting point of $Q$.
\end{definition}

\begin{definition}
\label{Def:x'}
 If path $Q$ ends in $\mathcal{N}_1$ then we define  $x' \in \mathcal{N}_1$ to be the last point of $ Q$. If path $Q$ ends in $\mathcal{N}_2$ we define $x'$ to be the first point of $Q$  in $\mathcal{N}_2$ such that all points after $x'$ are in $\mathcal{N}_2$.
\end{definition}

Notice that points $x$ and $x'$ are always unique and well-defined for any path $Q$ in a sub-network $\mathcal{N}[\beta_1,\beta_2]$. Path $Q$ can be decomposed into three parts $(Q_x,q_{xx'},Q_{x'})$ where $Q_x$ is the sub-path of $Q$ from its starting point to point $x$, $q_{xx'}$ is the sub-path of $Q$ from $x$ to $x'$ and $Q_{x'}$ is the sub-path of $Q$ from point $x'$ to its end point. Following Definitions \ref{Def:x} and \ref{Def:x'}, observe that if $x$ is in $\mathcal{N}_1$ then the sub-path $Q_x$ has only points in $\mathcal{N}_1$ and if $x$ is in $\mathcal{N}_2$ then $Q_x$ is empty. Similarly, if $x'$ is in $\mathcal{N}_2$ then all points in the sub-path $Q_{x'}$ are in $\mathcal{N}_2$  and if $x'$ is in $\mathcal{N}_1$ then $Q_{x'}$ is empty. The sub-path $q_{xx'}$ has points both in $\mathcal{N}_1$ and $\mathcal{N}_2$ and it is empty if $x=x'$.

Theorem \ref{Thm:coordinationnoPi} describes the specification of the coordination algorithm $\mathcal{C}_k$. Using Theorem \ref{Thm:coordinationnoPi}, the proof of Theorem \ref{Thm1noPi} follows by double induction on both parameter $k$ and the size of the network~$n$. 

\begin{theorem}
\label{Thm:coordinationnoPi}
For $k \geq 2$, assuming that Theorem \ref{Thm1noPi} is true for $k-1$, when algorithm $\mathcal{C}_k$ is applied to a sub-network $(\mathcal{N}; \mathcal{Y}_1, \mathcal{Y}_2,
\ldots,\mathcal{Y}_{k-1})[\beta_1,\beta_2]$ the computation follows the sub-path $q_{xx'}$ of every non-self-crossing path $Q$ in this sub-network.
\end{theorem}

We say that the sub-path $q_{xx'}$ of $Q$  \textit{crosses} from $\mathcal{N}_2$ to $\mathcal{N}_1$ if it traverses a red edge $(u,u') \in \mathcal{Y}_j, j \in [1,k-1]$ such that $u \in \mathcal{N}_2$ and $u' \in \mathcal{N}_1$.   Notice that for $k \geq 2$ there are exactly $k-1$ red edges $(u_1,u'_1),\ldots,(u_{k-1},u'_{k-1})$ that cross from $\mathcal{N}_2$ to $\mathcal{N}_1$. The proof of Theorem \ref{Thm:coordinationnoPi} depends on the fact that for a sub-network $\mathcal{N}[\beta_1,\beta_2]$, the sub-path $q_{xx'}$ can cross at most $(k-1)$ times from $\mathcal{N}_2$ to $\mathcal{N}_1$ (as each such crossing traverse one of the $k-1$ red 
edges from $\mathcal{N}_2$ to $\mathcal{N}_1$)
and on the analysis of the structure of a non-self-crossing path $Q$.

\paragraph{Computational Example}\textit{} \newline
To resolve any ambiguity, in Figures \ref{originalInstance}, \ref{shortestpathtreecross}, \ref{shortestpathtree} and \ref{finalinstance}  we show an example of the input and output of algorithm $\mathcal{A}_k$ for the special case $k=2$. Figure \ref{originalInstance} shows the implicit representation of the residual network $\mathcal{N}_{k-1}$ for $k=2$. The red segment represents path $\mathcal{Y}_1=(s,y_2,y_3,y_4,y_5,t)$. 

For clarity we do not show the black edges (long and short). Further, to simplify matters, we assume that the weight of each point $i$ is equal to $1$. That is, the weight of a black short edge $(i^-,i^+)$ is equal to $-1$ and the weight of a red short edge is equal to $1$. 

Figure \ref{shortestpathtreecross} shows the  shortest path tree $T$ computed by algorithm $\mathcal{A}_2$ in the residual network $\mathcal{N}_1$. The computed shortest path from $s$ to $t$ in $T$ is path $Q=(s,b_3,b_4,y_4,y_3,y_2,b_2,b_5,b_6,t)$.  Figure \ref{shortestpathtree} shows the two optimal node-disjoint paths $Y_1$ and $Y_2$ (which can cross) in $\mathcal{N}$ if we obtain the flow for $k=2$ in the usual way.

Finally, Figure \ref{finalinstance} shows the resulting collection of two optimal node-disjoint, non-crossing paths  $\mathcal{Y}_1$ and $\mathcal{Y}_2$ in $\mathcal{N}$ obtained by the additional post-processing algorithm $\widetilde{U}_2$, which will be the input for algorithm $\mathcal{A}_3$.\footnote{Observe that we need at least 3 robots to collect balls $b_3,b_4$ and $y_3$ and the schedule shown for two robots collects every ball except $y_3$ so it must be optimal.}

To conclude that algorithm $\mathcal{A}_2$ computes a shortest path from $s$ to $t$ in the residual
network, we will show that 
the sequence of $\relaxop$ operations executed during 
the computation includes a sub-sequence of $\relaxop$ operations which
corresponds, or 'follows', a non-self-crossing shortest path. 
For the example shown in Figures \ref{originalInstance},\ref{shortestpathtreecross},\ref{shortestpathtree} and \ref{finalinstance}, this sub-sequence of $\relaxop$ operations is $(s,b_3),(b_3,b_4),(b_4,y_4),(y_4,y_3),(y_3,y_2),(y_2,b_2),(b_2,b_5),(b_5,b_6),(b_6,t)$. Notice that only the  relative order of these relax operations is important.

\begin{figure}
\centering
\begin{subfigure}{.5\textwidth}
  \centering
  \includegraphics[scale=0.55]{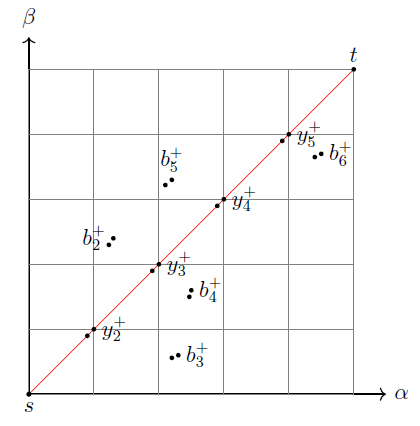}
  \caption{}
  \label{originalInstance}
\end{subfigure}%
\begin{subfigure}{.5\textwidth}
  \centering
  \includegraphics[scale=0.55]{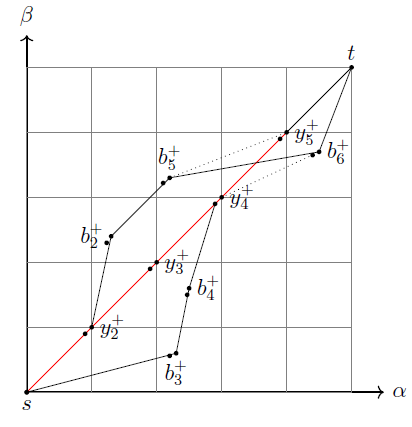}
    \caption{}\label{shortestpathtreecross}
\end{subfigure}
\begin{subfigure}{.5\textwidth}
  \centering
  \includegraphics[scale=0.55]{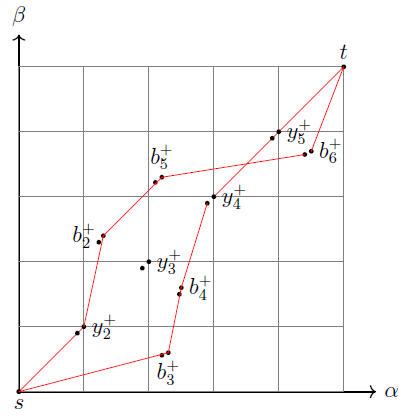}
    \caption{}\label{shortestpathtree}
\end{subfigure}%
\begin{subfigure}{.5\textwidth}
  \centering
  \includegraphics[scale=0.55]{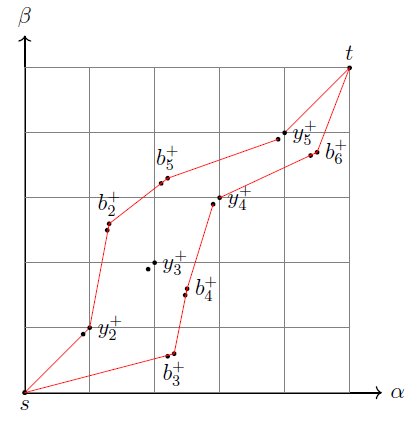}
    \caption{}\label{finalinstance}
\end{subfigure}
\caption{Figure \ref{originalInstance},\ref{shortestpathtreecross},\ref{shortestpathtree} and \ref{finalinstance} show a computational example of algorithm $\mathcal{A}_k$ for $k=2$.}
\end{figure}

\subsection{Implementation Details}
\label{section:arrays}
 Before we discuss the implementation details of algorithm $\mathcal{A}_k$ we remind the reader the structural details of the residual network $\mathcal{N}_{k-1}$. Recall that all nodes and edges of the paths $\mathcal{Y}_1, \mathcal{Y}_2, \ldots,$ $\mathcal{Y}_{k-1}$ are red. All other nodes and edges are black. As discussed in subsection \ref{sub:dag}, the node set $V(\mathcal{N}_{k-1})\setminus{\{s,t\}}$ consists of $n$ pairs of nodes $(1^-,1^+),(2^-,2^+),\ldots,(n^-,n^+)$ connected with a short edge.  To simplify matters, we refer to a pair of nodes $(x^-,x^+)$  as a \textit{pair} node $x$ in $\mathcal{N}_{k-1}$ or as a \textit{pair} point $x$ in $\mathcal{P}$, depending on the context. 

 For two pair nodes $x$ and $y$, if the residual network $\mathcal{N}_{k-1}$  has an edge $(y,x)$, then for the corresponding  pair points $x$ and $y$ in $\mathcal{P}$ we say that $x$ dominates  $y$ which is denoted by $y \prec x$. 
 For two pair points $y$ and $x$ such that $y \prec x$, the residual network $\mathcal{N}_{k-1}$ has either a long black edge $(y^+,x^-)$ or 
 a long red edge $(x^-,y^+)$ (the latter if 
 edge $(y,x) \in \mathcal{Y}_j$ where $j \in [1,k-1]$), 
 with weight equal to zero. 

For a black pair node $x$, the weight of the short edge $(x^-,x^+)$ is equal to $-w_x<0$. For a red pair node $x$ the  
 short edge $(x^-,x^+)$  has reversed direction from $x^+$ to $x^-$
 and weight equal to $w_x>0$.  The capacity of every edge regardless of colour (red or black) or type (long or short) is equal to $1$.

\paragraph{Initialisation of arrays $h$ and $pred$} \textit{ }

Consider the  two arrays $h$ and $\predecessor$, which are 
indexed by the nodes $V(\mathcal{N}_{k-1})\setminus\{s\}$.
For a pair node $x \in V(\mathcal{N}_{k-1})\setminus\{s\}$, terms $h(x)$ and $h^-(x)$ denote the current shortest path weight from the source $s$ to point $x^+$ and $x^-$, respectively. Similarly, we denote by $\predecessor(x)$ and $\predecessor^- (x)$ the predecessor of node $x^+$ and $x^-$ in the current tree. We initialize arrays $h$ and $\predecessor$
in the following way. 

For each black pair node 
$x\in V(\mathcal{N}_{k-1})$,
we set $\predecessor^-(x) = s$, $h^-(x) = 0$,
$\predecessor(x) = x^-$ and $h(x) = -w_x$.
For each red edge 
$(x^-, y^+)$ in $\mathcal{N}_{k-1}$ 
we set $\predecessor^-(x) = s$, $h^-(x) = 0$,
$\predecessor(y) = x^-$ and $h(y) = 0$.
For node $t$, we set $\predecessor(t) = s$ and 
$h(t) = 0$. 
Finally, to have the initial tree which reaches 
all nodes in $\mathcal{N}_{k-1}$, 
for each red edge $(s,x^-)$ in $\mathcal{N}_{k-1}$ ,
we set $\predecessor^-(x) = x^+$ and $h^-(x) = w_x$, 
and for each red edge $(t,x^+)$ in $\mathcal{N}_{k-1}$,
we set $\predecessor(x) = t$ and $h(x) = 0$.

The initialization of arrays 
$h$ and $\predecessor$ described above is valid for the relaxation 
technique since array $\predecessor$ defines a tree in 
$\mathcal{N}_{k-1}$ which is rooted at $s$
and for each node $x$ in $\mathcal{N}_{k-1}$ 
other than $s$,
$h(x)$ is the weight of the tree path from $s$ to $x$. An algorithm based on the relaxation technique
updates
arrays $h$ and $\predecessor$ only by
the following classic $\relaxop(y,x)$ operation \cite{Cormen2001introduction}: if $h(x) > h(y) + w(y,x)$, then $h(x) \leftarrow h(y) + w(y,x)$ and $\predecessor(x) \leftarrow y$, where 
$(y,x)$ is an edge in the input graph (or equivalently $y \prec x$), 
and $w(y,x)$ is the weight of this edge. Notice that for an operation $\relaxop(x,y)$ edge $(x,y)$ can be either short or long. 

At the end of the computation, for each node $x$,
$h(x)$ (resp. $h^-(x)$) should be equal to the shortest-path weight 
from $s$ to $x$ (resp. $x^-$), and array $\predecessor$ should 
represent a shortest path tree from the source $s$ to all
reachable nodes.
Since we want to compute a shortest path 
from $s$ to $t$, we will only require 
(and we will verify in the proofs) that 
at the end of the computation array $\predecessor$ 
includes a shortest path from $s$ to $t$ and that 
values $h(x)$ and $h^-(x)$ are correct for each node $x$ on this path.

An algorithm based on operation $\relaxop(x,y)$ 
computes a shortest $s$-$t$ path for a given input network,
if there is a shortest $s$-$t$ path
$(s=x_0,x_1,x_2,\ldots,x_q=t)$ such that 
the sequence of $\relaxop$ operations
executed by the algorithm
includes as a sub-sequence $\relaxop(x_{i-1},x_i)$,
$i = 1,2,\ldots,q$.
Only the relative order of 
such operations $\relaxop(x_{i-1},x_i)$ is important,
but they do not have to be
consecutive. They can be interleaved 
in arbitrary way with any number of other 
$\relaxop$ operations.

\paragraph{Two-Dimensional Orthogonal Search Problem} \textit{ }

 In the Two-Dimensional Orthogonal Search Problem we are given a set $S$ of $n$ points in a two dimensional plane where each point $i$ for $i=1,2,\ldots,n$ is identified with two coordinates $(x_i,y_i)$ and a weight value $u_i$. Given a  rectangle query $R=[x_1,x_2]$x$[y_1,y_2]$ the orthogonal search problem asks for the point  $j$ within $R$ that has the minimum weight value $u_j$. The operation of retrieving the point with the minimum $u$ value within $R$, can be seen as an answer to a query which has to be completed relatively fast. 
 
 We want to store all points in $S$ in a data structure such that given a query (\textit{i.e.} a rectangle $R$) we can perform the two basic operations: (i) Report the point with the minimum weight value within rectangle $R$ and (ii) Update the weight value of a given point in $S$. In such data structures the operations would usually be either only queries (the static version of the problem) or queries and insertions and deletions of points (the dynamic version of the problem). When we update the weight value of a point $i$ from $u_i$ to $u^{\prime}_i$ in $S$, we assume that we delete point $i$ and add a new point $i^{'}$ with weight $u^{\prime}_i$. 
 
 A variety of dynamic data structures such as range trees \cite{Lueker197828} \cite{x} \cite{Lee:1980:QTF:320613.320618}, layered range trees\cite{Gabow:1984:SRT:800057.808675} \cite{Willard:1985:NDS:3674.3690} and weight balanced trees\cite{Nievergelt:1972:BST:800152.804906} have been designed for dynamic and static versions of the Orthogonal Searching Problem. In \cite{Lueker197828} it was shown that the asymptotic upper bound of the time to respond to one query (i.e report the minimum weight point within a rectangle $R$) is $O(\log^{3}(n))$ in the case of a two dimensional space. Furthermore, it was shown that the upper bound on the running time of a sequence of $n$ operations which can be queries, insertions and deletions is $O(n \log^2 n)$.
 
 \paragraph{Operation $\Relaxop$} \textit{ }
\label{sub:Relax}

When algorithm ${\mathcal{A}_{k}}$ is applied to a sub-network $\mathcal{N}[\beta_1,\beta_2]$ of the residual network $\mathcal{N}_{k-1}$, 
operations $\relaxop$ are grouped together 
for the edges incoming to the same pair node.
For a pair node $x$,
we define operation $\Relaxop(x)$ as
 a sequence of all operations 
$\{ \relaxop(y^+,x^-): 
(y^+,x^-) \in \mathcal{N}[\beta_1,\beta_2]\}$, 
in arbitrary (because not relevant) order,
followed by the $\relaxop$ operation applied to the
residual edge outgoing from $x^-$ (if any). That edge
is either 
$(x^-,x^+)$, for a black pair node $x$, 
or $(x^-,\pi_x^+)$, 
for a red pair node $x$ on a path $\mathcal{Y}_j, j\in [1,k-1]$ with predecessor $\pi_x$.

For the worst-case running-time efficiency, 
we implement operations $\{ \relaxop(y^+,x^-): 
(y^+,x^-) \in \mathcal{N}[\beta_1,\beta_2]\}$
not by performing all of them explicitly
(this would take $O(n)$ time)
but by finding the pair node $y_{min}$ such 
that $h(y_{min}) =\min\{ h(y): 
(y^+,x^-) \in \mathcal{N}[\beta_1,\beta_2]\}$ and
performing only operation $\relaxop(y^+_{min},x^-)$.
We keep all pair nodes  in 
$\mathcal{N}[\beta_1,\beta_2]$
in a data structure 
for answering rectangle queries~\cite{Lueker197828}. 
The weight-value of a pair node $y$
in this data structure
is equal to the  current shortest path weight  
$h(y)$. 

For a sub-network $\mathcal{N}[\beta_1,\beta_2]$
and a pair node $x$ in this sub-network, we denote by $D_x[\beta_1,\beta_2]$ 
the set of all pair nodes $y$ in the sub-network such that there is an edge $(y^+,x^-)$, or equivalently $y \prec x$.
Notice that for a pair node $y \in D_x[\beta_1,\beta_2]$ the corresponding pair point $y$ in $\mathcal{P}$ must be within the rectangle $R_x=[0,\alpha_x]\times[\beta_1,\beta_x]$ since $y \prec x$ (\textit{i.e.} $\alpha_y \leq \alpha_x$ and $\beta_y \leq \beta_x$).
Thus, for a pair node $x$ finding pair node $y_{min}$ in $D_x[\beta_1,\beta_2]$ amounts to finding the minimum value pair point in rectangle~$R_x$.

For a black pair node $x$, finding pair node $y_{min}$ consists of answering the rectangle query for $R_x$ since every edge $(y^+,x^-) \in  \mathcal{N}_{k-1}[\beta_1,\beta_2]$ is a residual edge. For a red pair node $x$ on some path $\mathcal{Y}_j$ where $j \in [1,k-1]$, we first remove from the data structure the predecessor pair node $\pi_x$ of $x$ on $\mathcal{Y}_k$
(since $(\pi_x^+,x^-)$ is not a residual edge), then 
find $y_{min}$ by answering the rectangle query for $R_x$,
and finally re-insert $\pi_x$ back to the data
structure.
Each single operation on the data structure from~\cite{Lueker197828} (rectangle query,
update of the value of a given element, deleting
a given element, or inserting a new element) takes
$O(\log^3 n)$ time, so the running time of 
operation $\Relaxop$ is $O(\log^3 n)$.
\section{Shortest Path Algorithm $\mathcal{A}_k$ for $k =2$}
\label{k1sect}
In this section we consider the special case $k=2$ as an introduction to our recursive approach. For a  sub-network $(\mathcal{N}; \mathcal{Y}_1)[\beta_1,\beta_2]$ or $\mathcal{N}[\beta_1,\beta_2]$ for short, when algorithm $\mathcal{A}_2$ is applied to this sub-network, the computation consists of three phases. Consider the two sub-networks  $\mathcal{N}[\beta_1,(\beta_1+\beta_2)/2]$ and $\mathcal{N}[(\beta_1+\beta_2)/2,\beta_2]$ which we denote by $\mathcal{N}_1$ and $\mathcal{N}_2$, respectively. The first phase is the recursive call of algorithm $\mathcal{A}_2$ on  $\mathcal{N}_1$.  The second phase calls  the coordination algorithm $\mathcal{C}_2$ on sub-network $\mathcal{N}[\beta_1,\beta_2]$. The third phase is the recursive call of algorithm $\mathcal{A}_2$ on $\mathcal{N}_2$.
The description of algorithm $\mathcal{A}_2$ is shown in pseudo-code in Algorithm \ref{PseudoA2main}.

\begin{algorithm}[h]
 \caption{Algorithm $\mathcal{A}_2$ on input $(\mathcal{N}; \mathcal{Y}_1)[\beta_1,\beta_2] \equiv \mathcal{N}[\beta_1,\beta_2]$ }
\SetAlgoLined
 \begin{algorithmic}
 \State $\mathcal{N}_1 \leftarrow \mathcal{N}[\beta_1,(\beta_1+\beta_2)/2];$ $\mathcal{N}_2 \leftarrow \mathcal{N}[(\beta_1+\beta_2)/2,\beta_2];$
\State $\mathcal{A}_2 (\mathcal{N}_1);$
  \State $\mathcal{C}_{2}(\mathcal{N}[\beta_1,\beta_2]);$
\State $\mathcal{A}_2 (\mathcal{N}_2);$
 \end{algorithmic}
 \label{PseudoA2main}
\end{algorithm}

When algorithm $\mathcal{C}_2$ is applied to a sub-network $\mathcal{N}[\beta_1,\beta_2]$ the computation consists of three steps. The first and third step call algorithm $\mathcal{A}_1$ on the sub-network $\mathcal{N}[\beta_1,\beta_2] \setminus \mathcal{Y}_1$ which denotes the sub-network without the red edges of path $\mathcal{Y}_1$. The second step calls algorithm $\Delta(\mathcal{Y}_1)$ on the sub-network $\mathcal{N}[\beta_1,\beta_2]$. 
The computational steps of algorithm $\mathcal{C}_2$ are described in pseudo-code  in Algorithm \ref{PseudoC2main}.

For an input sub-network $\mathcal{N}[\beta_1,\beta_2]$ algorithm $\mathcal{A}_1$ consists of two steps. The first step computes  a topological order of all points in the sub-network using the $O(n \log n)$-time algorithm of Asahiro. \etal~\cite{DBLP:journals/dam/AsahiroHMOSY06}. The second step considers the points of the sub-network in topological order, that is, for two points $v$ and $v^{\prime}$ such that $v \prec v^{\prime}$ point $v$ is considered first and when a point $v$ is considered  it performs operation $\Relaxop(v)$ as described in Sub-section \ref{sub:Relax}.

For a sub-network $\mathcal{N}[\beta_1,\beta_2]$, algorithm  $\Delta(\mathcal{Y}_1)$ traverses the red edges of path $\mathcal{Y}_1$ in the sub-network (if any) and performs operation $\relaxop$ on the red edges(long and short).  Specifically, let $y_i$ for $i=1,2, \ldots,m$ be the $i^{th}$ red point on path $\mathcal{Y}_1$  in the sub-network, such that $y_1 \prec y_2 \prec \ldots \prec y_m$.   Algorithm $\Delta(\mathcal{Y}_1)$ performs operation $\relaxop(y^+_j,y^-_j)$ and operation $\relaxop(y^-_j,y^+_{j-1})$ for $j=m,m-1,\ldots,2$. For $j=1$ only  operation $\relaxop(y^+_j,y^-_j)$ is performed since edge $(y^-_j,y^+_{j-1})$ does not exist.

\begin{algorithm}[h]
 \caption{Algorithm $\mathcal{C}_2$ on input  $\mathcal{N}[\beta_1,\beta_2]$  }
\SetAlgoLined
 \begin{algorithmic}
  \State $\mathcal{A}_{1}(\mathcal{N}[\beta_1,\beta_2]\setminus \mathcal{Y}_1);$
 \State $\Delta(\mathcal{Y}_1);$
 \State $\mathcal{A}_{1}(\mathcal{N}[\beta_1,\beta_2]\setminus \mathcal{Y}_1);$
 \end{algorithmic}
 \label{PseudoC2main}
\end{algorithm}

For a sub-network $\mathcal{N}[\beta_1,\beta_2]$ we say that a path $Q$ is non-chromatic (red-chromatic) if it traverses only black (red) edges. 
Lemmas \ref{lemma:onlyblack} and \ref{lemma:onlyred} describe the specification of algorithms $\mathcal{A}_1$ and $\Delta(\mathcal{Y}_1)$, respectively.
\begin{lemma}
\label{lemma:onlyblack}
When algorithm $\mathcal{A}_1$ is applied to a sub-network  $\mathcal{N}[\beta_1,\beta_2]$ the computation follows every non-chromatic path in the sub-network  and the running time is $O(n\log^3 n)$  where $n$ is the size of the sub-network.
\end{lemma}
\begin{proof}
Let $Q$ be a path in the sub-network $\mathcal{N}[\beta_1,\beta_2]$ such that $Q$ traverses only black edges.  Consider the ordering of the edges $(y_1,y_2)(y_2,y_3),...,(y_{m-1},y_m)$ in $Q$.
Recall that every point $v$ in sub-network  $\mathcal{N}[\beta_1,\beta_2]$ is a pair of points $(v^-,v^+)$ which are connected with a short edge of capacity $1$.  We  show that the computation of algorithm $\mathcal{A}_1$ includes a sequence of $\relaxop$ operations on edges $(y_1^-,y_1^+),(y_{1}^+,y^-_2),\ldots,(y_{m-1}^+,y^-_m),(y_1^-,y_1^+)$ in this relative order.

For a black edge $(x,y)$ points $x$ and $y$ can be either black or red. Therefore, if $Q$ includes a red point $v$ (\textit{i.e.} a point on path $\mathcal{Y}_1$) then $v$ must be either the starting or ending point of $Q$. In detail, if the starting point $y_1$  of $Q$ is red, then the first  edge of $Q$ is edge $(y_1^+,y_2^-)$ and if the ending point $y_m$ of $Q$ is red, then the last edge of $Q$ is edge $(y_{m-1}^+,y^-_m)$. This is because the short red edges $(y^-_1,y^+_1)$ and  $(y^-_m,y^+_m)$ are not residual edges. 

Because $Q$ traverses only black edges, for edge $(y_{i},y_{i+1}) \in Q$ where $i=1,2,\ldots,m-1$ it holds that $y_{i} \prec y_{i+1}$. Algorithm $\mathcal{A}_1$ considers all points in the sub-network in topological order and when a point $x$ is considered it performs operation $\Relaxop(x)$. Therefore, for any $i \in [1,m-1]$ and an edge $(y_i,y_{i+1}) \in Q$ operation $\Relaxop(y_i)$ precedes operation $\Relaxop(y_{i+1})$.

 For a point $x$ operation  $\Relaxop(x)\equiv 
\{ \relaxop(y^+,x^-): y \in D_x[\beta_1,\beta_2]\}$ is equivalent to sequence of operations $\relaxop(y^+,x^-)$ for every point $y$ in the sub-network such that $y \prec x$ (as defined in sub-section \ref{sub:Relax}). Thus, operation $\relaxop(y_{i-1}^+,y_{i}^-)$ is implicitly included in operation $\Relaxop(y_i)$  for $i=2,3,\ldots,m$. Further, for $i=1,2,\ldots,m$ if point $y_i$ is black then operation $\Relaxop(y_i)$ also includes operation $\relaxop(y_i^-,y^+_i)$.
For the special case $i=1$(resp. $i=m$), if point $y_i$ is red then the first (resp. last) $\relaxop$ operation in the sequence is on edge  $(y_{1}^+,y^-_2)$ (resp. $(y_{m-1}^+,y^-_m)$). 

We conclude that the computation of algorithm $\mathcal{A}_1$ includes all operations $\relaxop(y_{i-1},y_i)$ for $i=2,3,\ldots,m$ and all operations $\relaxop(y_i^-,y_i^+)$ for $i=1,2,\ldots,m$, in this relative order.
One operation $\Relaxop$ requires $O(\log^3n)$  where $n$ is the size of the sub-network $\mathcal{N}[\beta_1,\beta_2]$ and therefore the total running time of algorithm $\mathcal{A}_1$ is $O(n\log^3 n)$.
\end{proof}

\begin{lemma}
\label{lemma:onlyred}
When algorithm $\Delta(\mathcal{Y}_1)$ is applied on a sub-network $\mathcal{N}[\beta_1,\beta_2]$ the computation follows every red-chromatic path in the sub-network and the running time is $O(n)$ where $n$ is the size of the sub-network. 
\end{lemma}
\begin{proof}
 Algorithm let $m'$ be the number of red points  on path $\mathcal{Y}_1$ and denote by $y_i$ for $i=1,2,\ldots,m$ the $i^{th}$ red point such that $y_1 \prec y_{2}\prec \ldots \prec y_m$. Let $Q$ be a path in the sub-network such that $Q$ traverses only red edges. Denote by $(y_i,y_{i+1}),\ldots,(y_{j-1},y_j)$ the ordering of the red edges in $Q$ where $i,j \in [1,m]$ and $j \leq i$. Algorithm $\Delta(\mathcal{Y}_1)$ performs all relax operations $(y^+_m,y^-_m)(y^-_m,y^+_{m-1}),\ldots,(y^-_2,y^+_1),(y^+_1,y^-_1)$.

The proof simply follows by induction for $k=i,i-1,\ldots,j$. Path $\mathcal{Y}_1$ can have at most $n$ points and therefore it can have at most $(n-1)$ edges. Operation $\relaxop$ takes constant time and therefore the total time needed is $O(n)$.
\end{proof}

Recall that for sub-network $\mathcal{N}[\beta_1,\beta_2]$ and  a given path $Q = (x_0,x_1,\ldots,x_m)$, in this sub-network  (not necessarily a shortest
path or an $s$-$t$ path) we say that the computation follows path $Q$, if the computation includes all relax operations $\relaxop(x_i,x_{i+1})$, $i = 1,2,\ldots,m-1$ in this order.
Each operation $\relaxop(x_i,x_{i+1})$ may be implicitly included within operation $\Relaxop(x_{i+1})$. The following theorem describes the specification of algorithm $\mathcal{A}_2$.

\begin{theorem}
\label{Theorem:specialcasek=2}
When algorithm $\mathcal{A}_2$ is applied to a  sub-network $\mathcal{N}[\beta_1,\beta_2]$ the computation follows every path in the sub-network and the running time is $O(n \log^4 n)$ where $n$ is the size of the sub-network.
\end{theorem}

Recall that if the computation follows path $Q=(x_0,x_1,\ldots,x_m)$, then 
at the end of this computation, the computed shortest path weight $h(x_m)$  is at most the weight of $Q$.  Therefore, Theorem \ref{Theorem:specialcasek=2} implies that at the termination of  the computation of algorithm $\mathcal{A}_2$ on a  sub-network $\mathcal{N}[\beta_1,\beta_2]$, for every point $v$ in the sub-network we have $h(v)=h^*(v)$ where $h^*(v)$ is the  weight of a shortest path to $v$.

 \subsection{Proof of Theorem \ref{Theorem:specialcasek=2}}

For a sub-network $\mathcal{N}[\beta_1,\beta_2]$ let $Q$ be a path in the sub-network. Denote by $\mathcal{N}_1$ and $\mathcal{N}_2$ the sub-networks $\mathcal{N}[\beta_1,(\beta_2+\beta_1)/2]$ and $\mathcal{N}[(\beta_2+\beta_1)/2,\beta_2]$, respectively. Without loss of generality, we assume that $Q$ has points both in $\mathcal{N}_1$ and $\mathcal{N}_2$. 

Recall that according to Definition \ref{Def:x}, if path $Q$ starts in $\mathcal{N}_1$ then we define  $x \in \mathcal{N}_1$ to be the last point in $Q$ such that all points before $x$ are in $\mathcal{N}_1$. If path $Q$ starts in $\mathcal{N}_2$ we define $x$ to be the starting point of $ Q$. 
Similarly, according to Definition \ref{Def:x'}, if path $Q$ ends in $\mathcal{N}_1$ then we define  $x' \in \mathcal{N}_1$ to be the ending point of $ Q$. If path $Q$ ends in $\mathcal{N}_2$ we define $x'$ to be the first point of $Q$  in $\mathcal{N}_2$ such that all points after $x'$ are in $\mathcal{N}_2$.
 
 We can decompose  $Q$ into the following parts $Q_x,q_{xx'},Q_{x'}$ where $Q_x$ is the sub-path of $Q$ from its starting point to point $x$, $q_{xx'}$ is the sub-path of $Q$ from $x$ to $x'$ and $Q_{x'}$ is the sub-path of $Q$ from point $x'$ to the ending point of $Q$.  Definition \ref{Def:x} implies that if sub-path $Q_x$ is not empty  then all points in $Q_x$ are in $\mathcal{N}_1$. Similarly, Definition \ref{Def:x'} implies that if $Q_{x'}$ is not empty then all points in $Q_{x'}$ are in $\mathcal{N}_2$. The proof of Theorem \ref{Theorem:specialcasek=2} is outlined below.
 
When algorithm $\mathcal{A}_2$ is applied to a sub-network $\mathcal{N}[\beta_1,\beta_2]$ the first phase of the computation is the recursive call of algorithm $\mathcal{A}_2$ on $\mathcal{N}_1$. 
The second phase of the computation calls algorithm $\mathcal{C}_2$ on sub-network $\mathcal{N}[\beta_1,\beta_2]$. Finally, the third phase of the computation is the recursive call of algorithm $\mathcal{A}_2$ on $\mathcal{N}_2$.
The following theorem describes the specification of algorithm $\mathcal{C}_2$.

 \begin{theorem}
 \label{lemma:coordinationfork=2}
 When algorithm $\mathcal{C}_2$ is applied to a sub-network $\mathcal{N}[\beta_1,\beta_2]$  the computation follows the sub-path $q_{xx'}$ of every path $Q$ in the sub-network and the running time is $O(n \log^3 n)$ where $n$ is the size of the sub-network. 
 \end{theorem}

\begin{figure}
    \centering
    \includegraphics[scale=0.32]{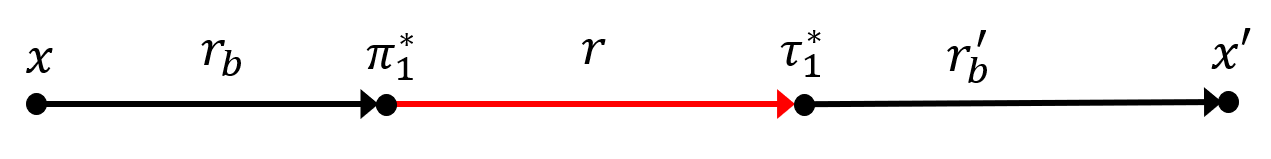}
    \caption{The structure of the sub-path $q_{xx'}$ of $Q$ from $x$ to $x'$ if $q_{xx'}$ crosses from $\mathcal{N}_2$ to $\mathcal{N}_1$.}
    \label{fig:combk1}
\end{figure}

It remains to show the proof of Theorem \ref{lemma:coordinationfork=2} and to conclude the proof of Theorem \ref{Theorem:specialcasek=2}. The following definitions  facilitate the analysis of the combinatorial structure of a path $Q$ in a sub-network $\mathcal{N}[\beta_1,\beta_2]$, followed by algorithms $\mathcal{C}_2$ and $\mathcal{A}_2$.
 \begin{definition}
\label{def:run}
For a path $Q$,  
a \textit{run} $r$ is a maximal sub-path of $Q$ such that all edges in $r$ are of the same colour. 
\end{definition}

 A run is \textit{non-chromatic} if it consists of black edges.  A run is \textit{chromatic} if all of its edges are of red colour. Notice that an edge $e(u,u^{\prime})$ can be traversed at most once by a (simple) path $Q$ and therefore we have the following corollary.
 
 \begin{corollary}
\label{cor:runsdisjoint}
Two chromatic runs $r$ and $r^{\prime}$ in $Q$ 
of the same colour $c_i$ are edge-disjoint.
\end{corollary}

\begin{proof}[Proof of Theorem \ref{lemma:coordinationfork=2}]
  When algorithm $\mathcal{C}_2$ is applied to sub-network $\mathcal{N}[\beta_1,\beta_2]$ the computation consists of the following three steps $\mathcal{A}_1(\mathcal{N}[\beta_1,\beta_2]),\Delta(\mathcal{Y}_1),\mathcal{A}_1(\mathcal{N}[\beta_1,\beta_2])$.
Let $Q$ be a path in the sub-network $\mathcal{N}[\beta_1,\beta_2]$ and let $q_{xx'}$ be the sub-path of $Q$ from $x$ to $x'$.  
We show that algorithm $\mathcal{C}_2$ follows path $q_{xx'}$.

 For $k=2$ there is a unique red edge $(u,u') \in \mathcal{Y}_1$ such that  $u \in \mathcal{N}_2$ and $u' \in \mathcal{N}_1$ crossing from $\mathcal{N}_2$ to $\mathcal{N}_1$. We say that the sub-path $q_{xx'}$ of $Q$ \textit{crosses} from $\mathcal{N}_2$ to $\mathcal{N}_1$ if it has a chromatic run that includes the red edge $(u,u')$. We consider two cases about the sub-path $q_{xx'}$.
 
 The first case is that $q_{xx'}$ does not cross from $\mathcal{N}_2$ to $\mathcal{N}_1$ and the second case is that $q_{xx'}$  crosses from $\mathcal{N}_2$ to $\mathcal{N}_1$ exactly once. In the former case, following the definition of points $x$ and $x'$ the sub-path $q_{xx'}$ must consist of the single black edge $(x,x')$. Thus, according to Lemma \ref{lemma:onlyblack} the computation of the first step of $\mathcal{C}_2$, that is, the computation of algorithm $\mathcal{A}_1$, follows the sub-path $q_{xx'}$ since it traverses only one black edge.
 
 In the latter case, the sub-path $q_{xx'}$ must  consist of the following ordering of runs $(r_b,r,r^{\prime}_b)$ where $r_b$ and $r^{\prime}_b$ are non-chromatic runs and $r$ is a red chromatic run which includes the red edge $(u,u')$. Thus, we can decompose $q_{xx'}$ into three parts $(x,\pi^*_1), (\pi^*_1,\tau^*_1)$ and $(\tau^*_1,x')$ where $\pi^*_1$ and $\tau^*_1$ is the first and last point of the red chromatic run $r$. An example of this decomposition is shown in Figure \ref{fig:combk1}. 
 
 According to Lemma \ref{lemma:onlyblack},  the computation of the first step of $\mathcal{C}_2$, that is, the computation of algorithm $\mathcal{A}_1$, follows the path from $x$ to $\pi^*_1$  since it traverses only black edges. According to Lemma \ref{lemma:onlyred}, the computation of the second step of $\mathcal{C}_2$, that is, algorithm $\Delta(\mathcal{Y}_1)$ follows the path from $\pi^*_1$ to $\tau^*_1$ since it traverses only red edges. Finally, according to Lemma \ref{lemma:onlyblack} the computation of the third step of $\mathcal{C}_2$, that is, algorithm $\mathcal{A}_1$ follows the path from $\tau^*_1$ to $x'$ since it traverses only black edges.
 
 The running time of algorithm $\mathcal{C}_2$ when applied to a sub-network $\mathcal{N}[\beta_1,\beta_2]$ of size $n$ is given by the following relationship $T_{\mathcal{C}_2}(n)=2T_{\mathcal{A}_1}(n)+ T_{\Delta}(n)$. According to Lemma \ref{lemma:onlyblack} and Lemma \ref{lemma:onlyred} we have that $T_{\mathcal{A}_1}(n)=O(n\log^3 n)$ and $T_{\Delta}(n)=O(n)$ and therefore $T_{\mathcal{C}_2}(n)=O(n \log^3 n)$.
\end{proof}

\begin{proof}[Proof of Theorem \ref{Theorem:specialcasek=2}]
 When algorithm $\mathcal{A}_2$ is applied to sub-network  $\mathcal{N}[\beta_1,\beta_2]$ 
the computation consists of the following three phases: $\mathcal{A}_2(\mathcal{N}_1),\mathcal{C}_2(\mathcal{N}[\beta_1,\beta_2]),\mathcal{A}_2(\mathcal{N}_2)$ which  denote the recursive call of algorithm $\mathcal{A}_2$ on $\mathcal{N}_1$, the call of algorithm $\mathcal{C}_2$ on $\mathcal{N}[\beta_1,\beta_2]$ and the recursive call of algorithm $\mathcal{A}_2$ on $\mathcal{N}_2$, respectively.

Let $Q$ be a path in this sub-network and consider the decomposition of $Q$ into $Q_x,q_{xx'},Q_{x'}$. Recall that if sub-path $Q_x$ (resp. $Q_{x'}$) is not empty then it must include points only in $\mathcal{N}_1$ (resp. $\mathcal{N}_2$). By induction, Theorem \ref{Theorem:specialcasek=2} implies that when algorithm $\mathcal{A}_2$ is applied to $\mathcal{N}_1$ then the computation follows the sub-path $Q_x$ of $Q$ since it has points only in $\mathcal{N}_1$. According to Theorem \ref{lemma:coordinationfork=2} when algorithm $\mathcal{C}_2$ is applied to sub-network $\mathcal{N}[\beta_1,\beta_2]$ the computation follows the sub-path $q_{xx'}$ of $Q$. Finally, by induction Theorem \ref{lemma:coordinationfork=2} implies that when algorithm $\mathcal{A}_2$ is applied to sub-network $\mathcal{N}_2$ the computation follows the sub-path $Q_{x'}$ of $Q$ since it has points only in $\mathcal{N}_2$.

The running time of algorithm $\mathcal{A}_2$ when applied to a sub-network of size $n$ is given by the following recurrence relationship: $T_{\mathcal{A}_2}(n)=T_{\mathcal{A}_2}(\frac{n}{2})+T_{\mathcal{C}_2}(n)+T_{\mathcal{A}_2}(\frac{n}{2})$ where $T_{\mathcal{C}_2}(n)$ is the running time of coordination algorithm $\mathcal{C}_2$. According to Theorem \ref{lemma:coordinationfork=2}, we have that $T_{\mathcal{C}_2}(n)=O(n \log^3 n)$ and therefore by solving  the recurrence relationship we obtain that $T_{\mathcal{A}_2}(n)=O(n\log^4 n)$.
\end{proof}
\section{Shortest Path Algorithm $\mathcal{A}_k$ for $k \geq 3$}
\label{knopi}
For $k \geq 3$ when algorithm $\mathcal{A}_k$ is applied to a sub-network 
$(\mathcal{N}; \mathcal{Y}_1, \mathcal{Y}_2,
\ldots,\mathcal{Y}_{k-1})[\beta_1,\beta_2]$ or $\mathcal{N}[\beta_1,\beta_2]$ for short, the computation consists of three phases. The first phase and third phase are
the recursive calls of $\mathcal{A}_k$ on $\mathcal{N}_1$ and $\mathcal{N}_2$, respectively. The second, coordination
phase calls algorithm $\mathcal{C}_k$ to the sub-network $\mathcal{N}[\beta_1,\beta_2]$.

Algorithm 
$\mathcal{C}_k$ repeats for $(k-1)$ times the following three steps. The first and third step consist of the following $k-1$ calls of algorithm $\mathcal{A}_{k-1}$: $\mathcal{A}_{k-1}(\mathcal{N}[\beta_1,\beta_2] \setminus \mathcal{Y}_1),\ldots,$
$\mathcal{A}_{k-1}(\mathcal{N}[\beta_1,\beta_2] \setminus \mathcal{Y}_{k-1})$. Term $\mathcal{N}[\beta_1,\beta_2]\setminus 
\mathcal{Y}_i$ denotes the sub-network $\mathcal{N}[\beta_1,\beta_2]$ without the red edges of paths $\mathcal{Y}_{i}$. We group this sequence of calls to algorithm $\mathcal{A}_{k-1}$ in this order to facilitate analysis and for simplicity we denote this sequence by $\widehat{\mathcal{A}}_{k-1}$.
The second step of algorithm $\mathcal{C}_k$ calls algorithm $\mathcal{Z}_k$ which is applied to 
the sub-network $\mathcal{N}[\beta_1,\beta_2]$.   
  
Algorithm $\mathcal{Z}_k$ has a recursive structure 
similar to algorithm $\mathcal{A}_k$
except that it works in the opposite direction. 
In more detail, the computation of algorithm  $\mathcal{Z}_k$ 
consists of three phases. 
The first and third phase are recursive calls to
$\mathcal{Z}_k$ on $\mathcal{N}_2$ and $\mathcal{N}_1$, respectively
(so the first recursive call is to the top half of the 
sub-network).
The second, coordination phase consists of two steps, as explained below. 

The first step calls algorithm $\Delta(\mathcal{Y}_1,\mathcal{Y}_2,\ldots,\mathcal{Y}_{k-1})$ which is the natural generalisation of algorithm $\Delta(\mathcal{Y}_1)$. That is, algorithm $\Delta(\mathcal{Y}_1,\mathcal{Y}_2,\ldots,\mathcal{Y}_{k-1})$ traverses the red edges (short and long) of each path  $\mathcal{Y}_1,\mathcal{Y}_2,\ldots,\mathcal{Y}_{k-1}$ in the sub-network starting from the last edge and moving towards the first edge. When a red edge $(u,u')$ is considered it performs operation $\relaxop(u,u')$. The second step repeats for $(k-2)$ times 
two calls to the sequence  $\widehat{\mathcal{A}}_{k-1}$. That is, one iteration consists of $2(k-1)$ calls  to algorithm $\mathcal{A}_{k-1}$.
Algorithms $\mathcal{A}_k$, $\mathcal{C}_k$ and
$\mathcal{Z}_k$ are described in pseudocode as 
Algorithms \ref{PseudoAmainnoPi}, \ref{PseudoCmainnoPi} and 
\ref{PseudoZmainnoPi}, respectively.

\begin{algorithm}[H]
 \caption{Algorithm $\mathcal{A}_k$ for input $\mathcal{N}[\beta_1,\beta_2]$ }
\SetAlgoLined
 \begin{algorithmic}
 \State $\mathcal{N}_1 \leftarrow \mathcal{N}[\beta_1,(\beta_1+\beta_2)/2]$ $\mathcal{N}_2 \leftarrow \mathcal{N}[(\beta_1+\beta_2)/2,\beta_2]$
\State $\mathcal{A}_k (\mathcal{N}_1);$
  \State $\mathcal{C}_{k}(\mathcal{N}[\beta_1,\beta_2]);$
\State $\mathcal{A}_k (\mathcal{N}_2);$
 \end{algorithmic}
 \label{PseudoAmainnoPi}
\end{algorithm}

\begin{algorithm}[H]
 \caption{Algorithm $\mathcal{C}_k$ for input  $\mathcal{N}[\beta_1,\beta_2]$ }
\SetAlgoLined
 \begin{algorithmic}

 \State \textbf{Repeat} for $i=1,2,\ldots,k-1$
  \State $\mathcal{A}_{k-1}(\mathcal{N}[\beta_1,\beta_2] \setminus \mathcal{Y}_1);,\ldots, \mathcal{A}_{k-1}(\mathcal{N}[\beta_1,\beta_2] \setminus \mathcal{Y}_{k-1});$
 \State $\mathcal{Z}_k(\mathcal{N}[\beta_1,\beta_2]);$
  \State  $\mathcal{A}_{k-1}(\mathcal{N}[\beta_1,\beta_2]\setminus \mathcal{Y}_1);,\ldots, \mathcal{A}_{k-1}(\mathcal{N}[\beta_1,\beta_2] \setminus \mathcal{Y}_{k-1});$
 \end{algorithmic}
 \label{PseudoCmainnoPi}
\end{algorithm}

\begin{algorithm}[H]
 \caption{Algorithm $\mathcal{Z}_k$ for input  $\mathcal{N}[\beta_1,\beta_2]$}
\SetAlgoLined
 \begin{algorithmic}
 \State $\mathcal{N}_1 \leftarrow \mathcal{N}[\beta_1,(\beta_1+\beta_2)/2]$ $\mathcal{N}_2 \leftarrow \mathcal{N}[(\beta_1+\beta_2)/2,\beta_2]$
\State $\mathcal{Z}_k (\mathcal{N}_2);$
\State $\Delta(\mathcal{Y}_1,\mathcal{Y}_2,\ldots,\mathcal{Y}_{k-1});$
\State \For{$i=1,2,\ldots,(k-2)$}{
    \State  $\mathcal{A}_{k-1}(\mathcal{N}[\beta_1,\beta_2]\setminus \mathcal{Y}_1);,\ldots, \mathcal{A}_{k-1}(\mathcal{N}[\beta_1,\beta_2] \setminus \mathcal{Y}_{k-1});$
     \State  $\mathcal{A}_{k-1}(\mathcal{N}[\beta_1,\beta_2]\setminus \mathcal{Y}_1);,\ldots, \mathcal{A}_{k-1}(\mathcal{N}[\beta_1,\beta_2] \setminus \mathcal{Y}_{k-1});$}
 \State $\mathcal{Z}_k (\mathcal{N}_1);$
\end{algorithmic}
\label{PseudoZmainnoPi}
\end{algorithm}
For  a sub-network $\mathcal{N}[\beta_1,\beta_2]$ let $Q$ be a non-self-crossing path in this sub-network. Without loss of generality, we assume that $Q$ has points both in $\mathcal{N}_1$ and $\mathcal{N}_2$. Analogously as for $k=2$ and according to Definitions \ref{Def:x} and \ref{Def:x'}, path $Q$ can be decomposed into three parts $(Q_x,q_{xx'},Q_{x'})$ where $Q_x$ is the sub-path of $Q$ from its starting point to point $x$, $q_{xx'}$ is the sub-path of $Q$ from $x$ to $x'$ and $Q_{x'}$ is the sub-path of $Q$ from point $x'$ to its ending point.  Recall that if $Q_x$ (resp. $Q_{x'}$) is not empty then it must include points only in $\mathcal{N}_1$ (resp. $\mathcal{N}_2$).

For $k \geq 3$ the proof of Theorem \ref{Thm1noPi} is outlined below. We assume by induction that when algorithm $\mathcal{A}_k$ is applied to $\mathcal{N}_1$ the computation follows every non-self-crossing path that has only points in $\mathcal{N}_1$. Thus, the computation of the recursive call of algorithm $\mathcal{A}_k$ on $\mathcal{N}_1$ follows the sub-path $Q_x$ of $Q$. According to  Theorem \ref{Thm:coordinationnoPi} the computation of algorithm $\mathcal{C}_k$ follows the sub-path $q_{xx'}$ of $Q$. Finally, we assume by induction that  when algorithm $\mathcal{A}_k$  is applied to $\mathcal{N}_2$  the computation follows every non-self-crossing path that has only points in $\mathcal{N}_2$.  Thus, the computation of the second recursive call follows the sub-path $Q_{x'}$ of $Q$.

 The  remaining part of the section is organized in the following way:  In Subsection \ref{sec:algoCknoPi} we specify the structure of paths followed by algorithm $\mathcal{C}_k$. In Subsection \ref{sec:AlgorithmZnoPi} we specify the  structure of paths followed by algorithm $\mathcal{Z}_k$. Finally, based on the analysis of Subsection \ref{sec:AlgorithmZnoPi} in Subsection \ref{sec:knoPi} we show the proof of Theorems \ref{Thm1noPi} and \ref{Thm:coordinationnoPi} for $k \geq 3$.

\subsection{Algorithm $\mathcal{C}_k$}
\label{sec:algoCknoPi}
For a sub-network $\mathcal{N}[\beta_1,\beta_2]$ and a non-self-crossing path $Q$ in the sub-network, algorithm $\mathcal{C}_k$ is employed to follow the sub-path $q_{xx'}$ of $Q$ from $x$ to $x'$ . In this sub-section we outline the combinatorial structure of path $q_{xx'}$. 
To facilitate analysis we  first introduce "shades" of red colour 
to distinguish between red edges of different paths $\mathcal{Y}_1, \mathcal{Y}_2,..,\mathcal{Y}_{k-1}$. Specifically, the edges of path $\mathcal{Y}_i$ for $i=1,2,\ldots,k-1$ are coloured with red colour $c_i$. 

Recall that according to Definition \ref{def:run}, a run $r$ is a maximal sub-path of a path $Q$ such that all edges are of the same colour.  A run is \textit{non-chromatic} if it consists of black edges.  A run is \textit{chromatic} if all of its edges are of the same red colour $c_i$, for some  $1 \leq i \leq k-1$.

\begin{definition}
For $0\leq d\leq k-1$ we say that a non-self-crossing path $Q$ (or a sub-path of $Q$) is $d$-chromatic, if the number of red colours in all chromatic runs is equal to $d$.
A path $Q$ (or a sub-path of $Q$) that does not have a chromatic run, is 
$0$-chromatic, or 
non-chromatic, and traverses only black edges.
\end{definition} 

\begin{definition}
We say that a $(k-1)$-chromatic path $Q$ is short $(k-1)$-chromatic if all chromatic runs of colour $c_1$ appear before all chromatic runs of colour $c_{k-1}$ or vice versa.
\end{definition}

 Recall that for a sub-network $\mathcal{N}[\beta_1,\beta_2]$ we denote by $\widehat{\mathcal{A}}_{k-1}$ the algorithm which performs the following sequence  of $(k-1)$ calls $\mathcal{A}_{k-1}(\mathcal{N}[\beta_1,\beta_2]\setminus \mathcal{Y}_1),\ldots,\mathcal{A}_{k-1}(\mathcal{N}[\beta_1,\beta_2]\setminus \mathcal{Y}_{k-1})$ to algorithm $\mathcal{A}_{k-1}$.  The following describes the specification of $\widehat{\mathcal{A}}_{k-1}$ when applied to a sub-network.
\begin{lemma}
\label{followatmostnoPi}
For $k\geq 3$, assuming that Theorem \ref{Thm1noPi} holds for $k-1$, when algorithm $\widehat{\mathcal{A}}_{k-1}$ is applied to a sub-network  $\mathcal{N}[\beta_1,\beta_2]$  then the computation follows any non-self-crossing path $Q$ in the sub-network such that $Q$ is at most $(k-2)$-chromatic.
\end{lemma}
\begin{proof}
Sub-networks $\mathcal{N}[\beta_1,\beta_2]\setminus \mathcal{Y}_1,\ldots,\mathcal{N}[\beta_1,\beta_2]\setminus \mathcal{Y}_{k-1}$  do not have a negative cycle  because there are sub-networks of the residual network $\mathcal{N}_{k-1}$ (\textit{i.e.} the sub-network for $\beta_1=0$ and $\beta_2=n+1$) which does not have a negative cycle. Consider a non-self-crossing path $Q$ such that $Q$ is at most $(k-2)$-chromatic. This means that $Q$ can traverse red edges of all paths except one path $\mathcal{Y}_i$ where $i \in [1,k-1]$.

Thus, path $Q$ must be a non-self-crossing path in one of the sub-networks $\mathcal{N}[\beta_1,\beta_2] \setminus \mathcal{Y}_1,\ldots,\mathcal{N}[\beta_1,\beta_2] \setminus \mathcal{Y}_{k-1}$. Without loss of generality, we assume that $Q$ is a path on sub-network $\mathcal{N}[\beta_1,\beta_2] \setminus \mathcal{Y}_i$ where $i \in [1,k-1]$. Algorithm $\widehat{\mathcal{A}}_{k-1}$ consists of applying algorithm $\mathcal{A}_{k-1}$ on sub-networks $\mathcal{N}[\beta_1,\beta_2] \setminus \mathcal{Y}_1,\ldots,\mathcal{N}[\beta_1,\beta_2] \setminus \mathcal{Y}_{k-1}$. Thus, assuming that Theorem \ref{Thm1noPi} holds for $k-1$, when algorithm $\mathcal{A}_{k-1}$ is applied to sub-network $\mathcal{N}[\beta_1,\beta_2] \setminus \mathcal{Y}_i$ the computation follows every non-self-crossing path in the sub-network. This completes the proof.
\end{proof}

When algorithm $\widehat{\mathcal{A}}_{k-1}$ is applied twice on a sub-network $\mathcal{N}[\beta_1,\beta_2]$ (denoted by  $\widehat{\mathcal{A}}_{k-1}$  $\widehat{\mathcal{A}}_{k-1}$) the computation consists of the following calls to algorithm $\mathcal{A}_{k-1}$: $\mathcal{A}_{k-1}(\mathcal{N}[\beta_1,\beta_2] \setminus \mathcal{Y}_1),\ldots,\mathcal{A}_{k-1}(\mathcal{N}[\beta_1,\beta_2] \setminus \mathcal{Y}_{k-1})$ and $\mathcal{A}_{k-1}(\mathcal{N}[\beta_1,\beta_2] \setminus \mathcal{Y}_1),\ldots,\mathcal{A}_{k-1}(\mathcal{N}[\beta_1,\beta_2] \setminus \mathcal{Y}_{k-1})$ in this order.

\begin{lemma}
\label{followatmostshortnoPi}
For $k\geq 3$, assuming that Theorem \ref{Thm1noPi} holds for $k-1$, when algorithm $\widehat{\mathcal{A}}_{k-1}$ is applied twice to a sub-network  $\mathcal{N}[\beta_1,\beta_2]$  then the computation follows any non-self-crossing path $Q$ in the sub-network such that $Q$ which is short $(k-1)$-chromatic.
\end{lemma}
\begin{proof}

Consider a non-self-crossing path $Q$ from a point $u$ to a point $u'$ such that $Q$ is short $(k-1)$-chromatic. Without loss of generality, we assume that all chromatic runs of colour $c_1$ appear before all chromatic runs of colour $c_{k-1}$ in $Q$. Let $u''$ be the first point of the first chromatic run of colour $c_{k-1}$ in $Q$. The sub-path of $Q$ from $u$ to $u''$ and the sub-path of $Q$ from $u''$ to $u'$ can be at most $(k-2)$-chromatic. 

That is, there is no chromatic run of colour $c_{k-1}$ (resp. $c_1$) between $u$ and $u''$ (resp. between $u''$ and $u'$). According to Lemma \ref{followatmostnoPi} when algorithm $\widehat{\mathcal{A}}_{k-1}$ is applied on sub-network $\mathcal{N}[\beta_1,\beta_2]$ the computation follows every non-self-crossing path which is at most $(k-2)$-chromatic. Thus, the first call to $\widehat{\mathcal{A}}_{k-1}$ follows the sub-path of $Q$ from $u$ to $u''$. Similarly, the second call to algorithm $\widehat{\mathcal{A}}_{k-1}$ follows the sub-path of $Q$ from $u''$ to $u'$.
\end{proof}

Note that algorithm $\mathcal{C}_k$ includes at least two calls to algorithm $\widehat{\mathcal{A}}_{k-1}$ (see steps 1 and 3 in Algorithm \ref{PseudoCmainnoPi}). This  means that if the sub-path $q_{xx'}$ is at most $(k-2)$-chromatic or short $(k-1)$-chromatic then according to Lemmas \ref{followatmostnoPi} and \ref{followatmostshortnoPi},  the computation of algorithm $\mathcal{C}_k$ follows the sub-path $q_{xx'}$.
Thus, for the remaining part of the analysis we consider the case where $q_{xx'}$ is $(k-1)$-chromatic.

We need the following definitions to outline the combinatorial structure of a $(k-1)$-chromatic non-self-crossing path in a sub-network $\mathcal{N}[\beta_1,\beta_2]$, with respect to the two sub-networks $\mathcal{N}_1$ and $\mathcal{N}_2$. 
\begin{definition}
 We say that a non-chromatic run crosses from $\mathcal{N}_1$ to $\mathcal{N}_2$ if it traverses a black edge $(u,u')$ such that $u \in \mathcal{N}_1$ and $u' \in \mathcal{N}_2$. We say that a chromatic run $r$ of colour $c_j$ where $j\in [1,k-1]$ crosses from $\mathcal{N}_2$ to $\mathcal{N}_1$ if it  traverses a red edge $(u,u') \in \mathcal{Y}_j$ such that $u \in \mathcal{N}_2$ and $u' \in \mathcal{N}_1$.
\end{definition}

If a run $r$ has only points  in $\mathcal{N}_1$ (resp. $\mathcal{N}_2$) we say that $r$ is placed in $\mathcal{N}_1$  (resp. $\mathcal{N}_2$).
\begin{definition}
We say that a path $Q$  crosses from $\mathcal{N}_2$ to $\mathcal{N}_1$ if $Q$ has a chromatic run $r$ which crosses from $\mathcal{N}_2$ to $\mathcal{N}_1$. We say that a path $Q$  crosses from $\mathcal{N}_1$ to $\mathcal{N}_2$ if $Q$ has a non-chromatic run that crosses from $\mathcal{N}_1$ to $\mathcal{N}_2$.
\end{definition}

 Similarly as for $k=2$, if the sub-path $q_{xx'}$ is not empty and does not cross from $\mathcal{N}_2$ to $\mathcal{N}_1$ it must hold that $x \in \mathcal{N}_1$, $x' \in \mathcal{N}_2$ and the sub-path $q_{xx'}$ simply consists of the black edge $(x,x')$. If the sub-path $q_{xx'}$ crosses at least once  from $\mathcal{N}_2$ to $\mathcal{N}_1$, observe that for $k \geq 3$ there are exactly $(k-1)$ red edges which cross from $\mathcal{N}_2$ to $\mathcal{N}_1$ (one for each path $\mathcal{Y}_1,\mathcal{Y}_2,...,\mathcal{Y}_{k-1}$). 
 
 Therefore, the sub-path $q_{xx'}$ can cross at most $m \leq k-1$ times from $\mathcal{N}_2$ to $\mathcal{N}_1$ (as each such crossing must traverse one of the $k-1$ red 
edges from $\mathcal{N}_2$ to $\mathcal{N}_1$). It is easy to see that every red edge of $q_{xx'}$ crossing from $\mathcal{N}_2$ to $\mathcal{N}_1$  must appear after a black edge crossing from $\mathcal{N}_1$ to $\mathcal{N}_2$. 

\begin{definition}
We define $w_i$ for $i=1,2,\ldots,m$ to be the last point in $\mathcal{N}_1$ before the $i^{th}$ crossing of $q_{xx'}$ from $\mathcal{N}_1$ to $\mathcal{N}_2$. We denote by $q_i$  the sub-path of $q_{xx'}$ from $w_i$ to $w_{i+1}$. 
\end{definition}
 
  For clarity, we denote $x$  and $x'$ by $w_1$ and $w_{k}$, respectively, and w.l.o.g, we assume that $m=k-1$. For $i=1$ and the  special case where $w_1$ is on $\mathcal{N}_2$ then $p_1$ does not cross from $\mathcal{N}_1$ to $\mathcal{N}_2$. Similarly, for $i=k-1$ and the special case where $w_k$ is on $\mathcal{N}_2$ then $p_{k-1}$ does not cross from $\mathcal{N}_2$ to $\mathcal{N}_1$. Denote by $w'_i \in \mathcal{N}_2$ the successor of $w_i$ in $q_i$  for $i=1,2,\ldots,k-1$. The following corollary outlines the structure of $q_i$ for $i =1,2,\ldots,k-1$.
 \begin{corollary}
 \label{cor:CoordinationCk}
 For $i=1,2,\ldots,k-1$ the sub-path $q_i$ of $q_{xx'}$ from $w_i$ to $w_{i+1}$ crosses the boundary between $\mathcal{N}_1$ and $\mathcal{N}_2$ twice. The first crossing is from  $\mathcal{N}_1$ to $\mathcal{N}_2$ identified with the black edge $(w_i,w'_i)$. The second crossing is from  $\mathcal{N}_2$ to $\mathcal{N}_1$ identified with a red chromatic run $r^*$ of colour $c_j$ where $j \in [1,k-1]$.
 \end{corollary}

 For $i=1,2,\ldots,k-1$ the sub-path $q_i$ of $q_{xx'}$ can be either at most $(k-2)$-chromatic or $(k-1)$-chromatic. In the former case, according to  Lemma \ref{followatmostnoPi} the computation of the first step in the $i^{th}$ iteration of algorithm $\mathcal{C}_k$, follows the sub-path $q_i$ of $q_{xx'}$. For the latter case, we provide the following definition to facilitate analysis.
 
 \begin{definition}
 \label{def:pointssplit}
 For a sub-network $\mathcal{N}[\beta_1,\beta_2]$ consider a $(k-1)$-chromatic path $Q$ from a point $u$ to a point $u'$. Let $\pi,\tau$ be the points in $Q$ such that the sub-path of $Q$ from $u$ to $\pi$ (resp. from $\tau$ to $u'$) is a maximal $(k-2)$-chromatic path.
 \end{definition}
 Points $\pi$ and $\tau$ are always unique and well-defined for any $(k-1)$-chromatic path $Q$.   Following Definition \ref{def:pointssplit}, let   $\pi_i$ and $\tau_i$ for $i=1,2,\ldots,k-1$ be the points in the sub-path $q_i$ of $q_{xx'}$ from $w_i$ to $w_{i+1}$.   Consider the decomposition of sub-path $q_{xx'}$ into the following parts $(w_1,\pi_1,\tau_1,w_2),\ldots,(w_{k-1},\pi_{k-1},\tau_{k-1},w_k)$, as shown in Figure \ref{fig:description}. 
 
 Note that for $i=1,2,\ldots,k-1$ if the path from $\pi_i$ to $\tau_i$ is empty (\textit{i.e.} $\tau_i$ appears before $\pi_i$ in $q_i$) then according to Definition \ref{def:pointssplit} the path from $\pi_i$ to $w_{i+1}$ is also $(k-2)$-chromatic. In this case, based on Lemma \ref{followatmostnoPi} we will show that the computation of the first and third step in the $i^{th}$ iteration of algorithm $\mathcal{C}_k$, follows the sub-path $q_i$ of $q_{xx'}$. From now on we consider the case where the path from $\pi_i$ to $\tau_i$ is not empty (\textit{i.e.} $\tau_i$ appears after $\pi_i$ in $q_i$).

\begin{figure}
    \centering
    \includegraphics[scale=0.45]{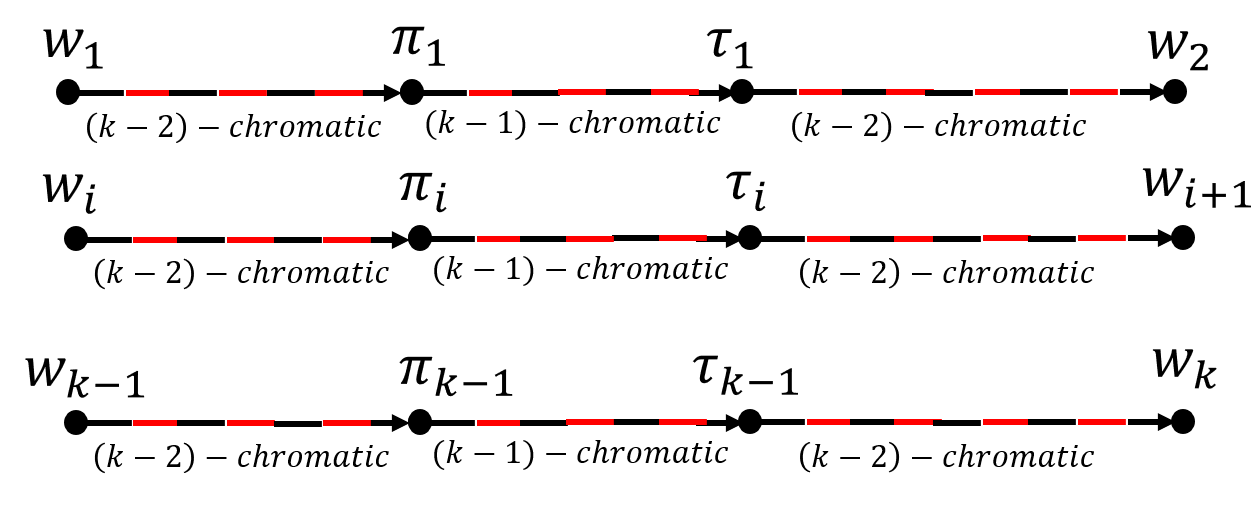}
    \caption{The combinatorial structure of the sub-path $q_{xx'}$ from $x$ to $x'$ of a non-self-crossing path $Q$, for $k \geq 3$. For clarity, we assume that $m=k-1$.}
    \label{fig:description}
\end{figure}

\begin{definition}
\label{def:closedterritory}
Consider the plane representation of the residual network $\mathcal{N}_{k-1}$.
We denote by $\Phi$  the closed subset of the plane whose boundary is described by the leftmost and rightmost path $\mathcal{Y}_1$ and $\mathcal{Y}_{k-1}$, respectively. 
\end{definition}

The exterior of $\Phi$ contains only black points. A black point $u$ in the exterior of $\Phi$ must be either on the left side of path $\mathcal{Y}_1$ or on the right side of path $\mathcal{Y}_{k-1}$. For the former case we say that $u$ is on the left exterior of $\Phi$, whereas in the latter case we say that $u$ is on the right exterior of $\Phi$. We say that a red point $u$ is a left (resp. right) boundary point of $\Phi$ if $u$ is a red point on path $\mathcal{Y}_1$ (resp. $\mathcal{Y}_{k-1}$). We say that a red or black point is in the interior of $\Phi$ if $u$ is a black point between two consecutive paths $\mathcal{Y}_j$ and $\mathcal{Y}_{j+1}$ where $j \in [1,k-2]$ or a red point on path $\mathcal{Y}_{j}$ where $j \in [2,k-2]$. A boundary point or a point in the interior of $\Phi$ is said to be in $\Phi$.

An edge $(u,u')$ is in $\Phi$ if the closed straight line segment $[u,u']$, corresponding to edge $(u,u')$) in the planar representation, is in  $\Phi$. Observe that all red edges of the paths $\mathcal{Y}_1,\mathcal{Y}_2,\ldots,\mathcal{Y}_{k-1}$ must be in $\Phi$. Thus, we have the following corollary.

\begin{corollary}
\label{cor:chromaticrunsinPhi}
For a sub-network $\mathcal{N}[\beta_1,\beta_2]$ and a non-self-crossing path $Q$ in this sub-network, all chromatic runs of  $Q$ are in $\Phi$.
\end{corollary}
    If a black edge is not in $\Phi$, denoted by $(u,u') \notin \Phi$, then the closed segment $[u,u']$ in the planar representation, must have a closed (sub)-segment in the exterior of $\Phi$.      A non-chromatic run $r$ is not in $\Phi$ if it has at least one black edge $(u,u')$ such that $(u,u') \notin \Phi$. 

 A black edge $(u,u') \notin \Phi$ is a \textit{boundary} edge if point $u$ is a right or left boundary point. A black edge $(u,u') \notin \Phi$ is a \textit{crossing} edge if point $u$ is in the interior of $\Phi$. Recall that paths $\mathcal{Y}_1,\mathcal{Y}_2,\ldots,\mathcal{Y}_{k-1}$ are non crossing pairwise and therefore a crossing edge $(u,u') \notin \Phi$  must necessarily cross at least one  red edge of path $\mathcal{Y}_1$ or path $\mathcal{Y}_{k-1}$.

For a sub-network $\mathcal{N}[\beta_1,\beta_2]$ and a non-self-crossing path $Q$ in the sub-network, consider the geometric representation of $Q$ with the set of points on the plane. Path $Q$ can be seen as a concatenation of straight line segments which represent the edges of $Q$ and  form a continuous segment  $\phi$ in the planar representation.

To facilitate analysis, we distinguish between points and space points. A  point $u$ in $\phi$ corresponds to node $u$ in the directed acyclic graph model (on which the sub-network $\mathcal{N}[\beta_1,\beta_2]$ is based on). A \textit{space} point $l$ in $\phi$ is a geometrical point on the closed segment $[u,u']$ of an edge $(u,u')\in Q$ and does not correspond to a node in the directed acyclic graph model.

\begin{definition}
\label{def:phisegment}
A sub-path $q$ of a non-self-crossing path $Q$ is a covering-path if the continuous segment $\phi$ corresponding to $q$  connects two space points (or points) on the right and left boundary of $\Phi$, respectively.
\end{definition}

Notice that the continuous segment $\phi$  of a covering path $q$ forms a boundary which splits $\Phi$ into two subsets, the bottom subset and the top subset.
Recall that for $i=1,2,\ldots,k-1$ we denote by $q_i$ the path from $w_i$ to $w_{i+1}$. Further according to Definition \ref{def:pointssplit}, for $i=1,2,\ldots,k-1$ path $q_i$ is decomposed into the following parts $(w_i,\pi_i)(\pi_i,\tau_i),(\tau_i,w_{i+1})$ (see Figure \ref{fig:description}).

 \begin{lemma}
 \label{lemma:CoordinationCknoPi}
 For $i=1,2,\ldots,k-1$, all runs (chromatic and non-chromatic) in the sub-path of $q_i$ from $\pi_i$ and $\tau_i$ are in $\Phi$.
 \end{lemma}
 \begin{proof}
 According to Corollary \ref{cor:chromaticrunsinPhi}, all chromatic runs in the sub-path of $q_i$ from $\pi_i$ and $\tau_i$ must be in $\Phi$. Thus, it remains to show that all non-chromatic runs are also in $\Phi$.
 Assume towards contradiction that for some $i \in [1,k-1]$ there is a non-chromatic run $r_b$ between $\pi_i$ and $\tau_i$ such that $r_b$ is not in $\Phi$. This means that $r_b$ must include at least one black edge which is not in $\Phi$.

 We denote by $(u,u')$ the first black edge in $r_b$ such that $(u,u') \notin \Phi$. Recall that a black edge which is not in $\Phi$ must be  either a crossing edge or a boundary edge.  Let $l$ be the space point which is defined in the following way. If edge $(u,u')$ is a boundary edge then $l$ is defined as point $u$. If edge $(u,u')$ is a crossing edge then $l$ is defined as the first crossing point on the closed segment $[u,u']$ with a red edge of path $\mathcal{Y}_1$ or path $\mathcal{Y}_{k-1}$.
 
 Without loss of generality, we assume that edge $(u,u')$ is a crossing edge and that space point $l$ is on path $\mathcal{Y}_1$. According to Corollary \ref{cor:CoordinationCk} there is exactly one chromatic run $r^*$ in $q_i$ which crosses from $\mathcal{N}_2$ to $\mathcal{N}_1$.   There are two possible cases: (1) edge $(u,u')$ appears before $r^*$ and (2) edge $(u,u')$ appears after $r^*$.
\paragraph{Case 1}(see Figure \ref{CoordinationCknoPia})\newline
According to Definition \ref{def:pointssplit}, the path from $w_i$ to $\pi_i$ is a maximal $(k-2)$-chromatic path. If $\pi_i$ is a red point on path $\mathcal{Y}_j$ where $j \in [1,k-2]$ then the path from $w_i$ to $\pi_i$ has at least one chromatic run of colour $c_{k-1}$. If $\pi_i$ is a red point on path $\mathcal{Y}_{k-1}$ then clearly the path from $w_i$ to $\pi_i$ has at least one chromatic run of colour $c_{k-1}$. Because $\pi_i$ appears before edge $(u,u')$ we conclude that there is at least one chromatic run of colour $c_{k-1}$ before edge $(u,u')$. Let $r$ be the last chromatic run of colour $c_{k-1}$ before the black edge $(u,u')$. 

Let $p$ be the path from the last point of run $r$ to point $u'$.
According to Definition \ref{def:phisegment} $p$ must be a covering path since there is a continuous segment $\phi$ which connects a right boundary point (the first point of run $r$ on $\mathcal{Y}_{k-1}$) and a left boundary point (the crossing point $l$ on path $\mathcal{Y}_1$). Notice that all runs in $p$ appear before $r^*$ and therefore $p$ has only points in $\mathcal{N}_2$. This means that the continuous segment $\phi$ is above the boundary separating $\mathcal{N}_1$ and $\mathcal{N}_2$. 

Let $\Phi \cap \mathcal{N}_2$ be the subset of $\Phi$ above the  boundary separating $\mathcal{N}_1$ and $\mathcal{N}_2$.  Consider the subset $\widetilde{\Phi}$ of $\Phi \cap \mathcal{N}_2$ which is described with the following two boundaries. The top boundary is the continuous segment $\phi$. The bottom boundary is the boundary separating $\mathcal{N}_1$ and $\mathcal{N}_2$. In  Figure \ref{CoordinationCknoPia}, the subset $\widetilde{\Phi}$ of $\Phi \cap \mathcal{N}_2$ is shown with the shaded area. 

The last point of run $r^*$ must be in $\mathcal{N}_1$  since $r^*$ crosses from $\mathcal{N}_2$ to $\mathcal{N}_1$. Since $\widetilde{\Phi}$ is a subset of $\Phi \cap \mathcal{N}_2$, the last point of run $r^*$ must be in the exterior of $\widetilde{\Phi}$. This means that the first point of $r^*$ must be between the top and bottom boundary of $\widetilde{\Phi}$ since otherwise run $r^*$  crosses with the continuous segment $\phi$  which implies a self-crossing. Thus, the first point of run $r^*$ must be in $\widetilde{\Phi}$.

Let $p'$ be the path from point $u$ to the first point of run $r^*$ and let $\phi'$ be the continuous segment (corresponding to $p'$) from the space point $l$ to the first point of run $r^*$.  All runs in $p'$ appear before $r^*$ which means that $p'$ has only points in $\mathcal{N}_2$. Thus, the continuous segment $\phi'$ is above the boundary separating $\mathcal{N}_1$ and $\mathcal{N}_2$. Since edge $(u,u') \notin \Phi$, the continuous segment $\phi'$ must have a closed segment $[l,l']$ in the exterior of $\Phi \cap \mathcal{N}_2$ and subsequently in the exterior of $\widetilde{\Phi}$. 

If the closed segment $[l,u']$ does not cross any red edges, then space point $l'$ is defined as point $u'$. If the closed segment $[l,u']$ crosses with at least one red edge, then the first crossing point on the closed segment $[l,u']$ must be with a red edge of $\mathcal{Y}_{1}$, since there are no red edges in the exterior of $\Phi$ and subsequently in the exterior of $\widetilde{\Phi}$. In this case, space point $l'$ is defined as the first crossing point on the closed segment $[l,u']$. 

The continuous segment from any arbitrary space point on the closed segment $[l,l']$ which is on the exterior of $\widetilde{\Phi}$ to the first point of run $r^*$ which is in $\widetilde{\Phi}$ must cross the top boundary of $\widetilde{\Phi}$. This implies, that path $p$ crosses with path $p'$, which makes a contradiction.

\paragraph{Case 2}(see Figure \ref{CoordinationCknoPib})\newline
Consider the path $p$ from the last point of run $r^*$ to point $u'$ and let $\phi$ be the continuous segment (corresponding to path $p$) from the last point of run $r^*$ to the space point $l$ on edge $(u,u')$. All runs in $p$ appear after $r^*$, which means that $p$ has only points in $\mathcal{N}_1$ and subsequently the continuous segment $\phi$ must be below the boundary separating $\mathcal{N}_1$  and $\mathcal{N}_2$.

Let $\Phi \cap \mathcal{N}_1$ be the subset of $\Phi$ below the  boundary separating $\mathcal{N}_1$  and $\mathcal{N}_2$. Consider the subset $\widetilde{\Phi}$ of $\Phi \cap \mathcal{N}_1$ which is described with the following top and bottom boundary. The bottom boundary of $\widetilde{\Phi}$  is described with the continuous segment $\phi$. The top boundary of $\widetilde{\Phi}$ is described with the boundary separating $\mathcal{N}_1$ and $\mathcal{N}_2$. In  Figure \ref{CoordinationCknoPib}, the subset $\widetilde{\Phi}$ of $\Phi \cap \mathcal{N}_1$ is shown with the shaded area.

Let $p'$ be the path from $u'$ to point $w_{i+1}$. All runs in $p'$ appear after $r^*$ and therefore $p'$ has only points in $\mathcal{N}_1$. Path $p'$ is non-self-crossing and therefore all chromatic runs in $p'$ must be in $\widetilde{\Phi}$.
The only red edges of path $\mathcal{Y}_{k-1}$ in $\widetilde{\Phi}$ (if any) are the red edges  traversed by path $p$. Therefore any red edges of path $\mathcal{Y}_{k-1}$ in  $\widetilde{\Phi}$  can not be traversed by $p'$ which means that $p'$ can be at most $(k-2)$-chromatic. 

Points $u$ and $u'$ are connected with a black edge. Hence, the path form $u$ to $w_{i+1}$ can also be at most $(k-2)$-chromatic. According to Definition \ref{def:pointssplit} the path from $\tau_i$ to $w_{i+1}$ is a maximal $(k-2)$-chromatic path. Therefore, point $\tau_i$ can not appear after point $u$ in the path from $\pi_i$ to $w_{i+1}$ (\textit{i.e.} either $\tau_i=u$ or $\tau_i$ precedes $u$). All non-chromatic runs in the path from $\pi_i$ to $u$ must be in $\Phi$ since edge $(u,u')$ is the first black edge such that $(u,u')\notin \Phi$.  Therefore,  all non-chromatic runs between $\pi_i$ and $\tau_i$ must also be in $\Phi$ since $\tau_i$ does not appear after $u$.
 \end{proof}
 
 \begin{figure}
\centering
\begin{subfigure}{.5\textwidth}
  \centering
  \includegraphics[scale=0.7]{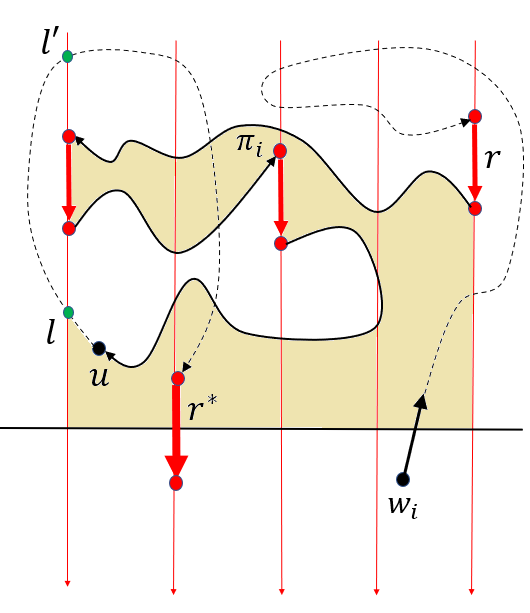}
    \caption{}
    \label{CoordinationCknoPia}
\end{subfigure}%
\begin{subfigure}{.5\textwidth}
  \centering
  \includegraphics[scale=0.8]{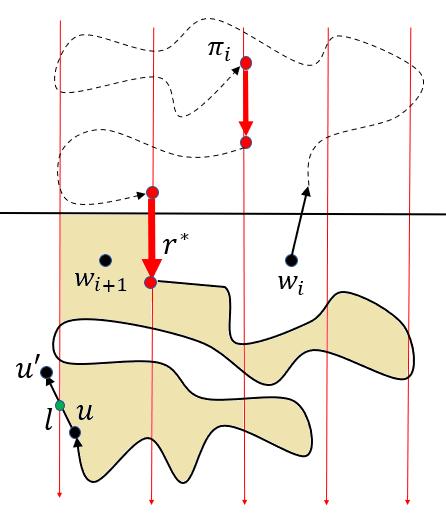}
    \caption{}
    \label{CoordinationCknoPib}
\end{subfigure}
\caption{Figure \ref{CoordinationCknoPia}: The schematic representation of Case 1 for the proof of Lemma \ref{lemma:CoordinationCknoPi}. Figure \ref{CoordinationCknoPib}: The schematic representation of Case 2 for the proof of Lemma \ref{lemma:CoordinationCknoPi}.}
\end{figure}
 
For a sub-network $\mathcal{N}[\beta_1,\beta_2]$ let $Q$ be a non-self-crossing path in this sub-network. Consider the sub-path $q_{xx'}$ of $Q$ from $x$ to $x'$ and more specifically its decomposition as shown in Figure \ref{fig:description} (for clarity we denote $x$ and $x'$ by $w_1$ and $w_{k}$, respectively, and assume that $m=k-1$). For $i=1,2,\ldots,k-1$ recall that $q_i$ denotes the sub-path of $q_{xx'}$ from $w_i$ to $w_{i+1}$. Without loss of generality, we assume that path $q_i$ for $i=1,2,\ldots,k-1$  is $(k-1)$-chromatic.    Let $\mathcal{C}^i_k$ for $i=1,2,\ldots,k-1$ denote the $i^{th}$ iteration of algorithm $\mathcal{C}_k$. The proof of Theorem \ref{Thm:coordinationnoPi} is outlined below. 

For $i=1,2,\ldots,k-1$ the first step of $\mathcal{C}^i_k$ calls algorithm $\widehat{\mathcal{A}}_{k-1}$ on sub-network $\mathcal{N}[\beta_1,\beta_2]$. According to Lemma \ref{followatmostnoPi} when algorithm  $\widehat{\mathcal{A}}_{k-1}$ is applied on a sub-network the computation follows every non-self-crossing path $Q$ such that $Q$ is at most $(k-2)$-chromatic. According to Definition \ref{def:pointssplit},  the  sub-path of $q_i$ from $w_i$ to $\pi_i$ is at most $(k-2)$-chromatic. Thus the computation of the first step follows the sub-path of $q_i$ from $w_i$ to $\pi_i$.

The second step of algorithm $\mathcal{C}^i_k$ calls algorithm $\mathcal{Z}_k$ on sub-network $\mathcal{N}[\beta_1,\beta_2]$. As we will show in the next sub-section when algorithm $\mathcal{Z}_k$ is applied on a sub-network the computation follows any non-self-crossing path $Q$ such that all runs (chromatic and non-chromatic) in $Q$ are in $\Phi$. According to Lemma \ref{lemma:CoordinationCknoPi}, all runs in the sub-path of $q_i$ from $\pi_i$ to $\tau_i$ are in $\Phi$. Thus, the computation of algorithm $\mathcal{Z}_k$ follows the sub-path of $q_i$ from $\pi_i$ to $\tau_i$. 

The third step  of $\mathcal{C}^i_k$ calls algorithm $\widehat{\mathcal{A}}_{k-1}$ on sub-network $\mathcal{N}[\beta_1,\beta_2]$. 
Similarly as for the first step,  according to Lemma \ref{followatmostnoPi} when algorithm  $\widehat{\mathcal{A}}_{k-1}$ is applied on a sub-network the computation follows every non-self-crossing path $Q$ such that $Q$ is at most $(k-2)$-chromatic. According to Definition \ref{def:pointssplit}, the  sub-path of $q_i$ from $\tau_i$ to $w_{i+1}$ is at most $(k-2)$-chromatic. Thus, the computation of the third step of follows the  sub-path of $q_i$ from $\tau_i$ to $w_{i+1}$.

\subsection{Algorithm $\mathcal{Z}_k$}
\label{sec:AlgorithmZnoPi}
In this section we outline the combinatorial structure of paths followed by algorithm $\mathcal{Z}_k$. 
\begin{definition}
\label{def:Phipath}
For a sub-network $\mathcal{N}[\beta_1,\beta_2]$ and a non-self-crossing path $Q$ in this sub-network we say that $Q$ is a $\Phi$-path if all runs (chromatic and non-chromatic) of $Q$ are in $\Phi$.
\end{definition}

Recall that when algorithm $\mathcal{Z}_k$ is applied to a sub-network $\mathcal{N}[\beta_1,\beta_2]$ the computation consists of three phases (see algorithm \ref{PseudoZmainnoPi}). The first and third phase call algorithm $\mathcal{Z}_k$ recursively on $\mathcal{N}_2$ and $\mathcal{N}_1$, respectively. The second phase, coordinates $\mathcal{N}_1$ and $\mathcal{N}_2$ and consists of two steps. The first step calls algorithm $\Delta(\mathcal{Y}_1,\mathcal{Y}_2,\ldots,\mathcal{Y}_{k-1})$. The second step consists of $(k-2)$ iterations where each iteration performs two calls of  algorithm $\widehat{\mathcal{A}}_{k-1}$  to the sub-network  $\mathcal{N}[\beta_1,\beta_2]$. Theorem \ref{algo:Znopi} describes the specification of algorithm $\mathcal{Z}_k$.

\begin{theorem}
\label{algo:Znopi}
Assuming that Theorem \ref{Thm1noPi} holds for $k-1$, when algorithm $\mathcal{Z}_k$ is applied to a sub-network $\mathcal{N}[\beta_1,\beta_2]$ the computation follows every $\Phi$-path in this sub-network.
\end{theorem}

For a sub-network $\mathcal{N}[\beta_1,\beta_2]$, let  $Q$ be a $\Phi$-path in the sub-network from a point $\pi$ to a point $\tau$. Consider the two sub-networks $\mathcal{N}[\beta_1,(\beta_1+\beta_2)/2]$ and $\mathcal{N}[(\beta_1+\beta_2)/2,\beta_2]$, denoted by $\mathcal{N}_1$ and $\mathcal{N}_2$, respectively. Without loss of generality we assume that $Q$ is $(k-1)$-chromatic and that is has points both in $\mathcal{N}_1$ and $\mathcal{N}_2$. There are only two possible cases: Path $Q$ does not cross from $\mathcal{N}_2$ to $\mathcal{N}_1$ or path $Q$ crosses at least once from $\mathcal{N}_2$ to $\mathcal{N}_1$.

\begin{lemma}
\label{nocrossinginZnoPi1}
For a sub-network $\mathcal{N}[\beta_1,\beta_2]$ and a $(k-1)$-chromatic $\Phi$-path $Q$ in the sub-network, if  path $Q$ has points both in $\mathcal{N}_1$ and $\mathcal{N}_2$ but does not cross from $\mathcal{N}_2$ to $\mathcal{N}_1$ then the starting point of $Q$ is in $\mathcal{N}_1$ and the ending point of $Q$ is in $\mathcal{N}_2$.
\end{lemma}
\begin{proof}
Let $\pi$ and $\tau$ be the starting and ending point of path $Q$. We  claim that if $Q$ has points both in $\mathcal{N}_1$ and $\mathcal{N}_2$ but does not cross from $\mathcal{N}_2$ to $\mathcal{N}_1$ then $\pi$ must be on $\mathcal{N}_1$ and $\tau$ must be on $\mathcal{N}_2$. Assume towards contradiction that our claim is not true. If point $\pi$ is on $\mathcal{N}_2$ and point $\tau$ is on $\mathcal{N}_1$ then $Q$ necessarily crosses from $\mathcal{N}_2$ to $\mathcal{N}_1$, which makes a contradiction. Similarly, if both $\pi$ and $\tau$ are on $\mathcal{N}_2$ (resp. $\mathcal{N}_1$) and $Q$ has points both in $\mathcal{N}_1$ and $\mathcal{N}_2$, then there is at least one point $x$ in $\mathcal{N}_1$ (resp. $\mathcal{N}_2$) between $\pi$ and $\tau$, which means that $Q$ crosses from $\mathcal{N}_2$ to $\mathcal{N}_1$. Again, this makes a contradiction. We conclude that the start point $\pi$ of $Q$ must be in $\mathcal{N}_1$ and the ending point $\tau$ of $Q$ must be in $\mathcal{N}_2$.
\end{proof}

Recall that a $(k-1)$-chromatic path $Q$ is short $(k-1)$-chromatic if all chromatic runs of colour $c_1$ appear before all chromatic runs of colour $c_{k-1}$, or vice versa.

\begin{lemma}
\label{nocrossinginZnoPi}
For a sub-network $\mathcal{N}[\beta_1,\beta_2]$ and a $(k-1)$-chromatic $\Phi$-path $Q$ in the sub-network, if  path $Q$ has points both in $\mathcal{N}_1$ and $\mathcal{N}_2$ but does not cross from $\mathcal{N}_2$ to $\mathcal{N}_1$ then $Q$ is short $(k-1)$-chromatic.
\end{lemma}
\begin{proof}
According to Lemma \ref{nocrossinginZnoPi1} the starting point $\pi$ of $Q$ must be in $\mathcal{N}_1$ and the ending point $\tau$ of $Q$ must be in $\mathcal{N}_2$. This implies that $Q$ has exactly one black edge $(u,u')$ which crosses from $\mathcal{N}_1$ to $\mathcal{N}_2$. Let $l$ be the space point corresponding to the crossing point of edge $(u,u')$ with the boundary separating $\mathcal{N}_1$ and $\mathcal{N}_2$ (shown with green in Figures \ref{fig:NoPiZ1} and \ref{fig:NoPiZ2}).

Since $Q$ is $(k-1)$-chromatic it must have at least one chromatic run of colour $c_1$ and at least one chromatic run of colour $c_{k-1}$. Without loss of generality, we assume that the first run of colour $c_{k-1}$ appears before the first chromatic run $r'$ of colour $c_{1}$. It is sufficient to show that there is no chromatic run of colour $c_{k-1}$ after $r'$ in $Q$.

Among all chromatic runs of colour $c_{k-1}$  before $r'$ let $r$ be the last chromatic run of colour $c_{k-1}$. We first claim that run $r$ must have all of its points in $\mathcal{N}_1$. Assume towards contradiction that our claim is not true. This means that run $r$ has all of its points in $\mathcal{N}_2$.  Notice that $r$ can not have points both in $\mathcal{N}_2$ and $\mathcal{N}_1$, since this implies that  $Q$ crosses from $\mathcal{N}_2$ to $\mathcal{N}_1$. If $r$ has all of its points in $\mathcal{N}_2$ then edge $(u,u')$ must appear before $r$. An example is shown in Figure \ref{fig:NoPiZ1}. 

Let $p$ be the sub-path of $Q$ from point $u$ to the first point of run $r$. Denote by $\phi$ the continuous segment (corresponding to path $p$) from the space point $l$ to the first point of run $r$. Notice that $\phi$ must be above the boundary separating $\mathcal{N}_1$ and $\mathcal{N}_2$. Consider the subset $\widetilde{\Phi}$ of $\Phi\cap \mathcal{N}_2$ which is described by the following two boundaries. The bottom boundary is the boundary separating $\mathcal{N}_1$ and $\mathcal{N}_2$. The top boundary is the continuous segment $\phi$.

Run $r$ appears before the first chromatic run $r'$ of colour $c_1$ in $Q$ which means that $p$ does not have a chromatic run of colour $c_1$.  By definition, all non-chromatic runs in path $Q$ and subsequently in $p$ are in $\Phi$. Therefore, there are no  red edges of path $\mathcal{Y}_1$ in $\widetilde{\Phi}$. Let $p'$ be the sub-path of $Q$ from the last point of run $r$ to the ending point $\tau$ of $Q$. Clearly, run $r'$ must be in path $p'$.

Path $p'$ can not cross the boundary from $\mathcal{N}_2$ to $\mathcal{N}_1$. This means that all runs (chromatic and non-chromatic) in $p'$ must be in $\Phi \cap \mathcal{N}_2$. Further, because $Q$ is non-self-crossing, all chromatic runs in $p'$ must be in $\widetilde{\Phi}$. There are no  red edges of path $\mathcal{Y}_1$ in $\widetilde{\Phi}$ which means that $p'$ can not have a chromatic run of colour $c_1$. However, this makes a contradiction because the chromatic run $r'$ of colour $c_1$ must be in path $p'$. We conclude that $r$ has all of its points in $\mathcal{N}_1$.

 We now claim that run $r'$ must have all of its points in $\mathcal{N}_2$. Assume towards contradiction that our claim is not true. Similarly, as before, it must be that $r'$ has all of its points in $\mathcal{N}_1$, otherwise we obtain a contradiction\footnote{If run $r'$ has point both in $\mathcal{N}_2$ and $\mathcal{N}_1$ this implies that $Q$ crosses from $\mathcal{N}_2$ to $\mathcal{N}_1$.}. This means that $r'$ (and subsequently $r$) appear before edge $(u,u')$ crossing from $\mathcal{N}_1$ to $\mathcal{N}_2$. According to Definition \ref{def:phisegment} the path from the last point of $r$ to the first point of $r'$ has a covering path since there is a continuous segment $\phi$ which connects two points on the left and right boundary of $\Phi$ (\textit{i.e.} the last point of $r$ on path $\mathcal{Y}_{k-1}$ and the first point of $r'$ on path $\mathcal{Y}_1$).
 
 The continuous segment $\phi$ is below the boundary separating $\mathcal{N}_2$ and $\mathcal{N}_1$, since $r$ and $r'$ appear before edge $(u,u')$. An example is shown in Figure \ref{fig:NoPiZ2}. All runs (chromatic and non-chromatic) in $Q$ are in $\Phi$. Further,  $Q$ is non-self-crossing. Thus, all runs  after run $r'$ in $Q$ must be below $\phi$ and subsequently below the boundary separating $\mathcal{N}_1$ and $\mathcal{N}_2$. However, this makes a contradiction since the ending point $\tau$ of $Q$ is in $\mathcal{N}_2$. We conclude that $r'$ must have all of its points in $\mathcal{N}_2$. 

We now show that $Q$ can not have a chromatic run of colour $c_{k-1}$ which appears after the first chromatic run $r'$ of colour $c_1$. Let $p$ be the sub-path  of $Q$ from point $u$ to the first point of run $r'$. Notice that path $p$ can not have a chromatic run of colour $c_{k-1}$ because run $r$ is the last chromatic run of colour $c_{k-1}$ before run $r'$ and appears before edge $(u,u')$. Let $\phi$ be the continuous segment (corresponding to path $p$)  from  space point $l$ to the first point of run $r'$.

Consider the subset $\widetilde{\Phi}$ of $\Phi \cap \mathcal{N}_2$ which is described by the following two boundaries. The bottom boundary is the boundary separating $\mathcal{N}_1$ and $\mathcal{N}_2$. The top boundary is the continuous segment $\phi$. Notice that there are not any red edges of path $\mathcal{Y}_{k-1}$ in $\widetilde{\Phi}$ since path $p$ does not have a chromatic run of colour $c_{k-1}$. Let $p'$ be the sub-path of $Q$ from the last point of run $r'$ to the ending point $\tau$. 

Path $p'$ can not cross the boundary from $\mathcal{N}_2$ to $\mathcal{N}_1$. This means that all runs (chromatic and non-chromatic) in $p'$ must be in $\Phi \cap \mathcal{N}_2$. Further, because $Q$ is non-self-crossing, all chromatic runs in $p'$ must be in $\widetilde{\Phi}$. However, there is no chromatic run of colour $c_{k-1}$ in $\widetilde{\Phi}$ and subsequently $p'$ can not have a chromatic run of colour $c_{k-1}$. Thus, there can not be a chromatic run of colour $c_{k-1}$ after run $r'$, which means that $Q$ is short $(k-1)$-chromatic.
\end{proof}

 \begin{figure}
\centering
\begin{subfigure}{.5\textwidth}
  \centering
  \includegraphics[scale=0.6]{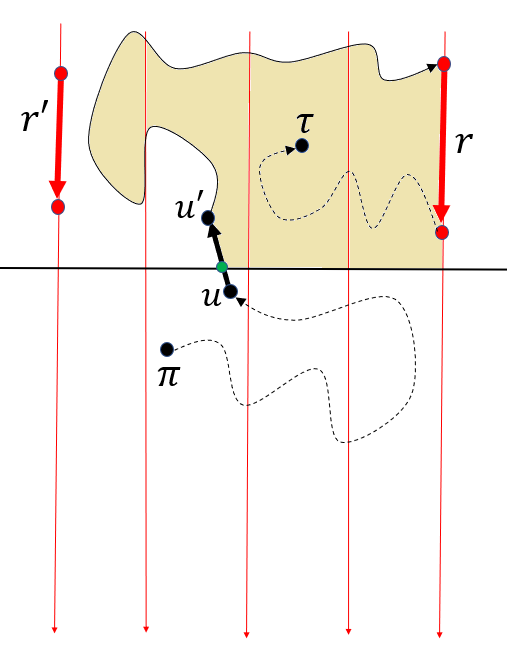}
    \caption{}
    \label{fig:NoPiZ1}
\end{subfigure}%
\begin{subfigure}{.5\textwidth}
  \centering
  \includegraphics[scale=0.6]{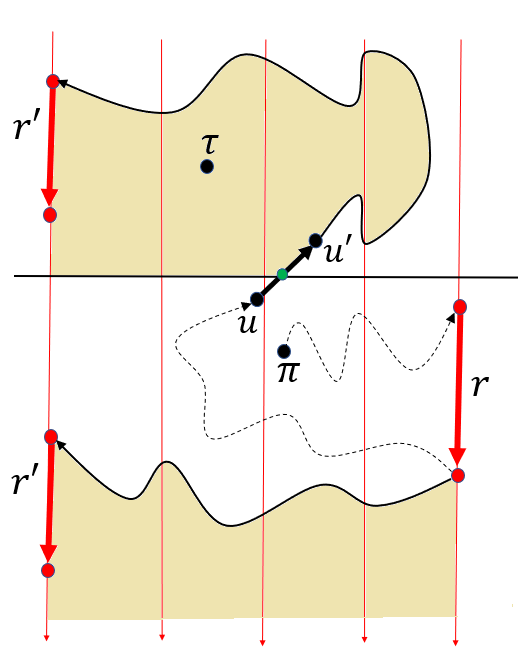}
    \caption{}
    \label{fig:NoPiZ2}
\end{subfigure}
\caption{Figures \ref{fig:NoPiZ1},\ref{fig:NoPiZ2}: The schematic representation of the proof for Lemma \ref{nocrossinginZnoPi}.}
\end{figure}
Lemma \ref{nocrossinginZnoPi} specifies the combinatorial structure of a $(k-1)$-chromatic $\Phi$-path $Q$ for the special case where $Q$ does not cross from $\mathcal{N}_2$ to $\mathcal{N}_1$. If path $Q$ crosses at least once and at most $m \leq k-1$ times from $\mathcal{N}_2$ to $\mathcal{N}_1$, we provide the following definition which will allow us to decompose $Q$ with respect to its crossings from $\mathcal{N}_2$ to $\mathcal{N}_1$.
\begin{definition}
\label{Def:winZ}
Define $w_i$  for $i=1,2,\ldots,m$ to be the last point in $\mathcal{N}_2$ before the $i^{th}$ crossing of $Q$ from $\mathcal{N}_2$ to $\mathcal{N}_1$. Let $w'_i$ be the successor point of $w_i$ in $Q$.
\end{definition} 
Without loss of generality, we assume that $m=k-1$. Let $\pi$ and $\tau$ be the starting and ending point of $Q$, respectively. We decompose path $Q$ into the following parts $(\pi,w'_1),(w'_1,w'_2),\ldots,(w'_{k-2},w'_{k-1}),(w'_{k-1},\tau)$.

For $i=1,2,\ldots,k-2$ the sub-path of $Q$ from $w'_i$ to $w_{i+1}$ does not cross from $\mathcal{N}_2$ to $\mathcal{N}_1$ since the red edge $(w_{i+1},w'_{i+1})$ denotes the next crossing of $Q$ from $\mathcal{N}_2$ to $\mathcal{N}_1$. Further, point $w'_i$ is in $\mathcal{N}_1$ and point $w_{i+1}$ is in $\mathcal{N}_2$ which means that the sub-path of $Q$ from $w'_i$ to $w_{i+1}$ has at least one point both in $\mathcal{N}_1$ and $\mathcal{N}_2$. Thus, according to Lemma \ref{nocrossinginZnoPi} we obtain the following corollary.

\begin{corollary}
\label{cor:Znopisubpath}
For $i=1,2,\ldots,k-2$ if the sub-path of $Q$ from $w'_i$ to $w_{i+1}$ is $(k-1)$-chromatic then it is short $(k-1)$-chromatic. 
\end{corollary}

\begin{lemma}
\label{lemma:ZnoPiCord}
For $i=1,2,\ldots,k-2$ if the sub-path of $Q$ from $w'_i$ to $w'_{i+1}$ is $(k-1)$-chromatic then it is short $(k-1)$-chromatic.
\end{lemma}
\begin{proof}
For any $i \in [1,k-2]$ let $q_i$ be the sub-path of $Q$ from $w'_i$ to $w'_{i+1}$ and let $q'_i=q_i \setminus \{(w_{i+1},w'_{i+1}\}$ be the sub-path of $q_i$ without the last red edge $(w_{i+1},w'_{i+1})$.  According to Corollary \ref{cor:Znopisubpath} if path $q'_i$ is $(k-1)$-chromatic then it must be short $(k-1)$-chromatic. Without loss of generality, we assume that all chromatic runs of colour $c_1$ appear before all chromatic runs of colour $c_{k-1}$ in $q'_i$. 

We show that if we construct $q_i$ from  $q'_i$  by adding the red edge $(w_{i+1},w'_{i+1})$, then the short $(k-1)$-chromatic condition is preserved. That is, all chromatic runs of colour $c_1$ appear before all chromatic runs of colour $c_{k-1}$ in $q_i$.

Let $c_j$ where $j \in [1,k-1]$ be the colour of the red edge $(w_{i+1},w'_{i+1})$. If $j \in [2,k-1]$ then clearly our claim is true.
That is, the red edge  $(w_{i+1},w'_{i+1})$ is not of colour $c_1$ and therefore all chromatic runs of colour $c_1$ appear before all chromatic runs of colour $c_{k-1}$ in path $q_i$. Thus, it remains to show that $j \neq 1$. Assume towards contradiction that $j=1$. 

 Notice that since $w'_i \in \mathcal{N}_1$ and $w_{i+1} \in \mathcal{N}_2$, path $q'_i$ must have exactly one black edge $(u,u')$ which crosses from $\mathcal{N}_1$ to $\mathcal{N}_2$. Let $l$ be the space point corresponding to the crossing point of edge $(u,u')$ with the boundary separating $\mathcal{N}_1$ and $\mathcal{N}_2$. Let $r$ be the last chromatic run of colour $c_1$ before the  first chromatic run $r'$ of colour $c_{k-1}$ in  $q'_i$.  Following the same methodology as in Lemma \ref{nocrossinginZnoPi}, we can obtain  that edge $(u,u')$ must appear between $r$ and $r'$. This means that $r$ has all of its points in $\mathcal{N}_1$ and run $r'$ has all of its points in $\mathcal{N}_2$. 

Let $p$ be the sub-path of  $q'_i$ from point $u$ to the first point of run $r'$ and let $\phi$ be the continuous segment (corresponding to this path) from the space point $l$  to the first point of run $r'$. All points  after $u$ in $q'_i$ must be in $\mathcal{N}_2$ and therefore the continuous segment $\phi$ is above the boundary separating $\mathcal{N}_1$ from $\mathcal{N}_2$. Further, run $r$ is the last chromatic run of colour $c_1$ before $r'$ and $r$ appears before edge $(u,u')$. Thus, path $p$ does not have a chromatic run of colour $c_1$.

   Consider the subset $\widetilde{\Phi}$ of $\Phi \cap \mathcal{N}_2$ which is described with the following two boundaries. The top   boundary  is the continuous segment $\phi$. The bottom boundary is the boundary separating $\mathcal{N}_1$ from $\mathcal{N}_2$. An example is shown in Figure \ref{fig:mainLemmaZnoPi1}.  There are not any red edges of path $\mathcal{Y}_1$ in  $\widetilde{\Phi}$ since path $p$ does not traverse any chromatic run of colour $c_1$. 
   
 Therefore, if $j=1$, that is, the red edge $(w_{i+1},w'_{i+1})$ is of colour $c_1$, then point $w_{i+1} \in \mathcal{N}_2$ must be in the exterior of $\widetilde{\Phi}$. Clearly, the last point of run $r'$ is in $\widetilde{\Phi}$. Let $p'$ be the sub-path of $q'_i$ from  the last point of run $r'$ to point $w_{i+1}$. All runs (chromatic and non-chromatic) in $p'$ must be in $\Phi$. Path $p'$ can not cross the boundary separating $\mathcal{N}_1$ and $\mathcal{N}_2$. Thus, all runs in $p'$ must be in $\Phi \cap \mathcal{N}_2$.  
 
 Therefore, path $p'$ must cross with the top boundary of $\widetilde{\Phi}$ (\textit{i.e.} the continuous segment $\phi$ corresponding to path $p$) since $w_{i+1}$ is on the exterior of  $\widetilde{\Phi}$ and the last point of run $r'$ is in $\widetilde{\Phi}$. This, implies that path $p'$ crosses with path $p$, which makes a contradiction.
\end{proof}

Figure \ref{fig:worstcaseZ} shows an example of a non-self-crossing $\Phi$-path $Q$ which crosses from $\mathcal{N}_2$ to $\mathcal{N}_1$ exactly $(k-1)$ times for $k=6$. Observe that the sub-path of $Q$ from $w'_i$ to $w'_{i+1}$ for $i=1,2,\ldots,k-2$ is short $(k-1)$-chromatic.
 \begin{figure}
\centering
\begin{subfigure}{.5\textwidth}
  \centering
   \includegraphics[scale=0.7]{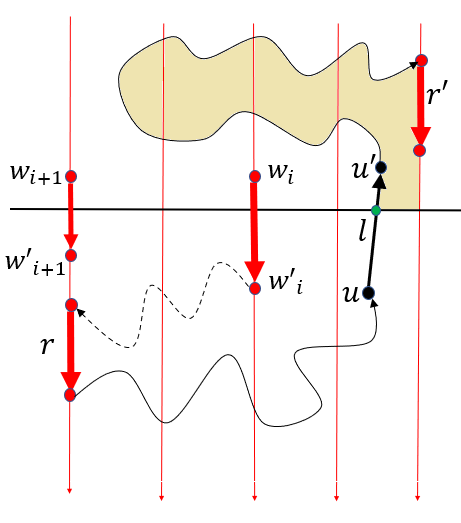}
    \caption{}
    \label{fig:mainLemmaZnoPi1}
\end{subfigure}%
\begin{subfigure}{.5\textwidth}
  \centering
  \includegraphics[scale=0.62]{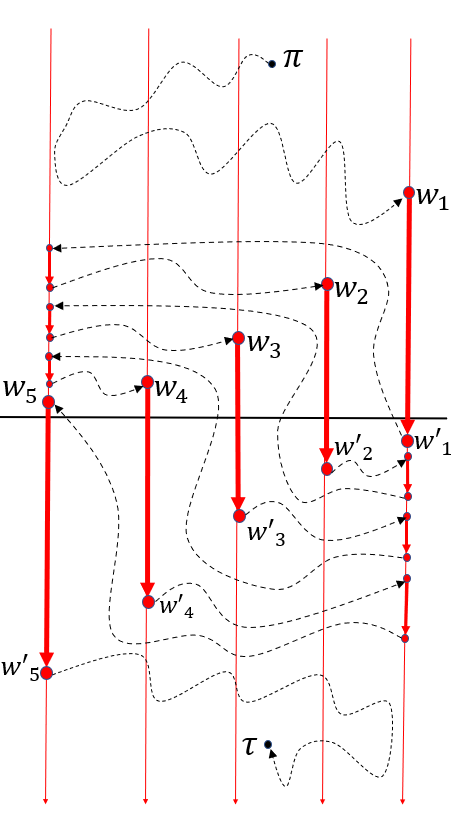}
    \caption{}
    \label{fig:worstcaseZ}
\end{subfigure}
\caption{Figure \ref{fig:mainLemmaZnoPi1}: The subset $\widetilde{\Phi}$ of $\Phi \cap \mathcal{N}_2$ is shown with the shaded territory.
Figure \ref{fig:worstcaseZ}: An example of a non-self-crossing $\Phi$-path $Q$ from a point $\pi$ to a point $\tau$ which crosses exactly $(k-1)$ times from $\mathcal{N}_2$ to $\mathcal{N}_1$ (for $k=6$).}
\end{figure}

\begin{proof}[Proof of Theorem \ref{algo:Znopi}]
For a sub-network $\mathcal{N}[\beta_1,\beta_2]$ consider a  $\Phi$-path $Q$ in this sub-network. Without loss of generality, we assume that $Q$ is $(k-1)$-chromatic and that $Q$ has points both in $\mathcal{N}_1$ and $\mathcal{N}_2$.  Recall that when algorithm $\mathcal{Z}_k$ is applied to sub-network $\mathcal{N}[\beta_1,\beta_2]$ the computation consists of three phases. 

The first and third phase call algorithm $\mathcal{Z}_k$ recursively on $\mathcal{N}_2$ and $\mathcal{N}_1$, respectively. The second phase consists of two steps. The first step calls algorithm $\Delta(\mathcal{Y}_1,\mathcal{Y}_2,\ldots,\mathcal{Y}_{k-1})$. The second step consists of $(k-2)$ iterations, where each iteration performs two calls of algorithm $\widehat{\mathcal{A}}_{k-1}$ on sub-network $\mathcal{N}[\beta_1,\beta_2]$. 

If $Q$ does not cross from $\mathcal{N}_2$ to $\mathcal{N}_1$, according to Lemma \ref{nocrossinginZnoPi} we have that $Q$ short $(k-1)$-chromatic. According to Lemma \ref{followatmostshortnoPi} when algorithm $\widehat{\mathcal{A}}_{k-1}$ is applied twice to sub-network $\mathcal{N}[\beta_1,\beta_2]$  the computation follows every non-self-crossing path which short $(k-1)$-chromatic. Algorithm $\mathcal{Z}_k$ includes at least two calls to algorithm $\widehat{\mathcal{A}}_{k-1}$ and therefore the computation follows path $Q$. 

If $Q$ crosses at least once and at most $m \leq k-1$ from $\mathcal{N}_2$ to $\mathcal{N}_1$, recall that according to Definition \ref{Def:winZ}, we define $w_i$ for $i=1,2,\ldots,m$ to be the last point in $Q$ before the $i^{th}$ crossing from $\mathcal{N}_2$ to $\mathcal{N}_1$. We denote by $w'_i$ the successor of $w_i$ in $Q$.  Consider the decomposition of $Q$ into the following parts $(\pi,w'_1),(w'_1,w'_2),\ldots,(w'_{k-2},w'_{k-1})(w'_{k-1},\tau)$ where $\pi$ and $\tau$ is the starting and ending point of $Q$, respectively. 

If the  the sub-path of $Q$ from $\pi$ to $w_1$ is not empty then it can only have points in $\mathcal{N}_2$. Similarly, if the  the sub-path of $Q$ from $w'_{k-1}$ to $\tau$ is not empty then it can only have points in $\mathcal{N}_1$. We now assume by induction that Theorem \ref{algo:Znopi} holds for $k$ and a sub-network of size less than $n$.

By the induction hypothesis, when algorithm $\mathcal{Z}_k$ is applied to sub-network $\mathcal{N}_2$ the computation follows every $\Phi$-path that has only points in $\mathcal{N}_2$. Thus, the computation of the  first phase (\textit{i.e.} first recursive call) follows follows the sub-path of $Q$ from $\pi$ to $w_1$. 

The second phase consists of two steps. The first step calls algorithm $\Delta(\mathcal{Y}_1,\mathcal{Y}_2,\ldots,\mathcal{Y}_{k-1})$ whose computation follows the red edge $(w_1,w'_1)$. We now claim that for $i=1,2,\ldots,k-2$ the computation in the $i^{th}$ iteration of the second step, follows the sub-path of $Q$ from $w'_i$ to $w'_{i+1}$. According to Lemma \ref{lemma:ZnoPiCord}, the sub-path of $Q$ from $w'_i$ to $w'_{i+1}$ can be either at most $(k-2)$-chromatic or short $(k-1)$-chromatic 

In the former case, according to Lemma \ref{followatmostnoPi} when algorithm $\widehat{\mathcal{A}}_{k-1}$ is applied to sub-network $\mathcal{N}[\beta_1,\beta_2]$ the computation follows every non-self-crossing path which is at most $(k-2)$-chromatic. In the latter case, according to  Lemma \ref{followatmostshortnoPi} when algorithm $\widehat{\mathcal{A}}_{k-1}$ is applied twice to sub-network $\mathcal{N}[\beta_1,\beta_2]$  the computation follows every non-self-crossing path which is short $(k-1)$-chromatic. 

The $i^{th}$ iteration  of the second step for $i=1,2,\ldots,k-2$ consists of two applications of algorithm  $\widehat{\mathcal{A}}_{k-1}$ to the sub-network. We conclude that the computation in the $i^{th}$ iteration for $i=1,2,\ldots,k-2$ follows  the sub-path of $Q$ from $w'_i$ to $w'_{i+1}$.

By the induction hypothesis, when algorithm $\mathcal{Z}_k$ is applied to sub-network $\mathcal{N}_1$ the computation follows every $\Phi$-path that has only points in $\mathcal{N}_2$. Thus, the computation in the  third  phase (\textit{i.e.} second recursive call) follows follows the sub-path of $Q$ from $w'_{k-1}$ to $\tau$.  This completes the proof. 
\end{proof}

\subsection{Proof of Theorem \ref{Thm:coordinationnoPi} and \ref{Thm1noPi} for $k \geq 3$}
\label{sec:knoPi}
\begin{proof}[Proof of Theorem \ref{Thm:coordinationnoPi} for $k \geq 3$]
Recall that when algorithm $\mathcal{C}_k$ is applied to sub-network $\mathcal{N}[\beta_1,\beta_2]$ the computation consists of $(k-1)$ iterations where each iteration has three steps (see algorithm \ref{PseudoCmainnoPi}). The first and third step call algorithm $\widehat{\mathcal{A}}_{k-1}$ on the sub-network  while the second step calls algorithm $\mathcal{Z}_k$ on the sub-network.

For a sub-network $\mathcal{N}[\beta_1,\beta_2]$ let $Q$ be a non-self-crossing path in the sub-network. Consider the sub-path $q_{xx'}$ of $Q$, according to Definitions \ref{Def:x} and \ref{Def:x'}. If  $q_{xx'}$ does not cross from $\mathcal{N}_2$ to $\mathcal{N}_1$ it must be that $x \in \mathcal{N}_1$,$x' \in \mathcal{N}_2$ and path $q_{xx'}$ consists of a single black edge $(x,x')$.

According to Lemma \ref{nocrossinginZnoPi} when algorithm $\widehat{\mathcal{A}}_{k-1}$ is applied to sub-network  $\mathcal{N}[\beta_1,\beta_2]$ the computation follows every non-self-crossing path which is at most $(k-2)$-chromatic. Algorithm $\mathcal{C}_k$ includes at least one call to algorithm $\widehat{\mathcal{A}}_{k-1}$ and therefore the computation of $\mathcal{C}_k$ follows the sub-path $q_{xx'}$ of $Q$. 

If $q_{xx'}$ crosses at least once from $\mathcal{N}_2$ to $\mathcal{N}_1$, then decompose $q_{xx'}$ into the following parts $(w_1,w_2),\ldots,(w_{k-1},w_k)$, where $w_i$ for $i=1,2,\ldots,k$ is the last point in $\mathcal{N}_1$ before the $i^{th}$ crossing of $q_{xx'}$ from $\mathcal{N}_1$ to $\mathcal{N}_2$ (see subsection \ref{sec:algoCknoPi}).

Without loss of generality, assume that the sub-path $q_i$ of $q_{xx'}$ from $w_i$ to $w_{i+1}$ for $i=1,2,\ldots,k-1$ is $(k-1)$-chromatic. Decompose $q_i$ into $(w_i,\pi_i,\tau_i,w_{i+1})$, according to Definition \ref{def:pointssplit}, such that the path from $w_i$ to $\pi_i$  (resp. $\tau_i$ to $w_{i+1}$) is a maximal $(k-2)$-chromatic path (see  Figure~\ref{fig:description}). 

We claim that the computation in the $i^{th}$ iteration $\mathcal{C}^i_k$ of algorithm $\mathcal{C}_k$ for $i=1,2,\ldots,k-1$ follows the sub-path of $q_{xx'}$ from $w'_i$ to $w'_{i+1}$. The first step in  $\mathcal{C}^i_k$  calls algorithm $\widehat{\mathcal{A}}_{k-1}$ on the sub-network, whose computation, follows any at most $(k-2)$-chromatic non-self-crossing path in the sub-network, according to Lemma \ref{followatmostnoPi}. Thus, the computation of the first step follows the path from $w_i$ to $\pi_i$.  

According to Lemma \ref{lemma:CoordinationCknoPi} the path from $\pi_i$ to $\tau_i$ is a $\Phi$-path (\textit{i.e.} all runs are in $\Phi$). The second step in  $\mathcal{C}^i_k$  calls algorithm $\mathcal{Z}_k$ whose computation follows every $\Phi$-path in the sub-network,  according to Theorem \ref{algo:Znopi}. Thus, the computation of the second step follows the path from $\pi_i$ to $\tau_i$. 

The third step in  $\mathcal{C}^i_k$  calls algorithm $\widehat{\mathcal{A}}_{k-1}$ on the sub-network, whose computation, follows any at most $(k-2)$-chromatic non-self-crossing path in the sub-network according to Lemma \ref{followatmostnoPi}. Thus, the computation of the third step follows the path from $\tau_i$ to $w_{i+1}$. This completes the proof.
\end{proof}

\begin{proof}[Proof of Theorem \ref{Thm1noPi} for $k \geq 3$]
The proof follows by induction. The base case of the induction is $k=2$, where according to  Theorem \ref{Theorem:specialcasek=2} when algorithm  $\mathcal{A}_2$ is applied to a sub-network $\mathcal{N}[\beta_1,\beta_2]$ the computation follows every path in the sub-network. We now assume by induction that Theorem \ref{Thm1noPi} holds for any $k'<k$. We also assume by induction that Theorem \ref{Thm1noPi} holds for $k$ and a sub-network of size less than $n$.

Let $Q$ be a non-self-crossing path in a sub-network $\mathcal{N}[\beta_1,\beta_2]$. Consider the decomposition of $Q$ into $(Q_x,q_{xx'},Q_{x'})$ according to Definitions \ref{Def:x} and \ref{Def:x'}. Recall that if path $Q_x$ (resp. $Q_{x'})$ is not empty then all points in this path are  in $\mathcal{N}_1$ (resp. $\mathcal{N}_2$).

Recall that when algorithm $\mathcal{A}_k$ is applied to sub-network $\mathcal{N}[\beta_1,\beta_2]$ the computation consists of three phases (see algorithm \ref{PseudoAmainnoPi}). The first and third phase call algorithm $\mathcal{A}_k$ recursively on $\mathcal{N}_1$ and $\mathcal{N}_2$, respectively. The second phase calls algorithm $\mathcal{C}_k$ on the sub-network.

By the induction hypothesis, when algorithm $\mathcal{A}_k$ is applied  recursively to $\mathcal{N}_1$ the computation  follows the sub-path $Q_x$ of $Q$ since $Q_x$ has only points in $\mathcal{N}_1$. According to Theorem \ref{Thm:coordinationnoPi} when algorithm $\mathcal{C}_k$ is applied to $\mathcal{N}[\beta_1,\beta_2]$ the computation follows the sub-path $q_{xx'}$ of $Q$. By the induction hypothesis,  when algorithm $\mathcal{A}_k$ is applied to $\mathcal{N}_2$ the computation follows the sub-path $Q_{x'}$ of $Q$ since it has only points in $\mathcal{N}_2$.

For $k \ge 2$, the running time of algorithm $\mathcal{A}_{k}$ over $n$ points is given by the following relationship $\mathcal{T}^k_A(n)=\mathcal{T}^k_A(\frac{n}{2})+\mathcal{T}^k_C(n)+ \mathcal{T}^k_A(\frac{n}{2})$, where $\mathcal{T}^k_C(n)$ is the running time of algorithm $\mathcal{C}_k$ over $n$ points. The running time of algorithm $\mathcal{C}_k$ is given by the following relationship $\mathcal{T}^k_C(n)=(k-1)[\mathcal{T}^{k-1}_A(n)+\mathcal{T}^k_Z(n)+\mathcal{T}^{k-1}_A(n)]$, where $\mathcal{T}^k_Z(n)$ is the running time of algorithm $\mathcal{Z}_k$ over $n$ points.

The running time of algorithm $\mathcal{Z}_k$ over $n$ points is given by the relationship $\mathcal{T}^k_Z(n)=\mathcal{T}^k_Z(\frac{n}{2})+k^2 \mathcal{T}^{k-1}_A(n)+\mathcal{T}^k_Z(\frac{n}{2})+O(n)$. By solving the recurrence relationship we obtain that $\mathcal{T}^k_Z(n)=O(k^2\log n) \mathcal{T}^{k-1}_A(n)$ and $\mathcal{T}^k_C(n)=O(k^3 \log n)\mathcal{T}^{k-1}_A(n)$ which results into $\mathcal{T}^k_A(n)=O(k^{3k}\log^{2k}n) \mathcal{T}^1_A(n)$. Putting everything together, we conclude that $\mathcal{T}^k_A(n)=O(k^{3k}n \log^{2k+3}n)$.
\end{proof}

\section{Proof of Theorem \ref{maintheoremuncrossing} for $k \geq 3$}
\label{noncrossingpathsfork}
In this section we show the proof of Theorem \ref{maintheoremuncrossing}. That is, for a set of points $\mathcal{P}_k$ such that all points can be covered with $k$ paths we show an algorithm $\widetilde{U}_k$ with the running time of $O(k n\log n)$ which computes a collection of $k$ node-disjoint non-crossing paths covering all points in $\mathcal{P}_k$.
For $k \geq 3$ consider the implicit plane $\alpha-\beta$ representation of the directed acyclic graph $G^*$ with the points in $\mathcal{P}$ (see sub-section \ref{sub:plane}). Recall that for an edge $(u,u')$ in $G^*$, point $u \in \mathcal{P}$ is (Pareto) dominated by point $u' \in \mathcal{P}$ ($\alpha_u \le \alpha_v$ and $\beta_u \le \beta_v$), which is denoted by $u \prec u'$. An edge $(u,u')$ in $G^*$ can be seen as a closed segment $[u,u']$ in the plane representation. An $s-t$ path in $G^*$, forms a continuous segment on the plane which consists of a concatenation of straight line segments, corresponding to the edges of this path.

Recall that two node-disjoint edges $(u,u')$ and $(x,x')$ in $G^*$
 cross, if the two (closed) segments $[u,u']$ and $[x,x']$ in the plane have a common point. All points are in general position and therefore  the common point $\pi$ on the closed segments $[u,u']$ and $[x,x']$ is a \textit{crossing} point, which is geometrically defined with coordinates $\alpha_{\pi}$ and $\beta_{\pi}$, but it does not correspond to a point in $\mathcal{P}$ and subsequently to a node in $G^*$. 

\begin{lemma}
\label{l5}
If two edges $(u,u')$ and  $(x,x')$ in $G^*$ cross then $G^*$ has edges $(u,x')$ and $(x,u')$.
\end{lemma}
\begin{proof}
It suffices to show that that $u \prec x'$ and $x \prec u'$. 
Clearly, for edges $(u,u')$ and $(x,x')$ we have $u \prec u'$ and $x \prec x'$, respectively.
Let $\pi$ be the crossing point of edge of the two closed segments $[u,u']$ and $[x,x']$. We have that $u \prec \pi \prec u'$ (resp. $x \prec \pi \prec x'$) since the crossing point $\pi$ is on the closed segment $[u,u']$ (resp. $[x,x']$). This implies that $u \prec \pi \prec x'$ and $x \prec \pi \prec u'$. 
\end{proof}

 We say that two paths in $G^*$ cross if there is an edge of the first path crossing with an edge of the second path. From now on we consider explicitly the plane representation of $G^*$ with the points in $\mathcal{P}$. 
 \begin{definition}
 For $k \geq 3$, let $\mathcal{P}_k\subseteq \mathcal{P} $ be a point set such that all points in $\mathcal{P}_k$ can be covered by $k$ node-disjoint $s-t$ paths $(Y_1,Y_2,\ldots,Y_k)$.
 \end{definition}

To facilitate analysis, we denote a collection of $k$ node-disjoint (but not necessarily non-crossing) paths by $(Y_1,Y_2,\ldots,Y_k)$ and a collection of $k$ node-disjoint, non-crossing paths by $(\mathcal{Y}_1,\mathcal{Y}_2,\ldots,\mathcal{Y}_k)$. 
We show an $O(kn\log n )$-time algorithm $\widetilde U_k$ for $k \geq 3$, which takes as an input a point set $\mathcal{P}_k$ and gives as an output  a collection $(\mathcal{Y}_1,\mathcal{Y}_2,..,\mathcal{Y}_k)$ $k$ node-disjoint, non-crossing $s-t$ paths. The paths $(\mathcal{Y}_1,\mathcal{Y}_2,\ldots,\mathcal{Y}_k)$ are given in left to right order in their planar representation, that is, $\mathcal{Y}_1$ is the leftmost path and $\mathcal{Y}_k$ is the rightmost path. 

The rest of the section is organized as follows: In Subsection \ref{collection}, based on a simple argument, we show that for a point set $\mathcal{P}_k$ there exists at least one collection  of $k$ node-disjoint, non-crossing paths $(\mathcal{Y}_1,\mathcal{Y}_2,..,\mathcal{Y}_k)$ that cover all points in $\mathcal{P}_k$. In Subsection \ref{selection} we discuss a subroutine algorithm $\mathcal{S}$ which is employed by algorithm $\widetilde U_k$. Given a point set $\mathcal{P}_k$, algorithm $\mathcal{S}$  computes a collection of $k$ node-disjoint (but not necessarily non-crossing) paths by $(Y_1,Y_2,\ldots,Y_k)$ covering all points in $\mathcal{P}_k$. The collection of paths obtained by algorithm $\mathcal{S}$ satisfies some geometrical properties which we use in the analysis of algorithm $\widetilde U_k$. In Subsection \ref{algoUK} we provide the detailed description of algorithm $\widetilde U_k$ and show the proof of Theorem \ref{maintheoremuncrossing}.

\subsection{Existence of $k$ non-crossing paths}
\label{collection}
For a point set $\mathcal{P}_k$ let $(Y_1,Y_2,..,Y_k)$ be a collection of $k$ node-disjoint $s-t$ paths (but not necessarily non-crossing) covering all points in $\mathcal{P}_k$. 

\begin{definition}
For two indexes $i,j \in [1,k]$ such that $i \neq j$ and two crossing edges $(u,u') \in Y_i$ and $(x,x') \in Y_j$, we define operation uncross which replaces edge $(u,u')\in Y_i$ and edge $(x,x') \in Y_j$ with edge $(u,x')$ and $(x,u')$, respectively.
\end{definition}

Essentially, for a collection  $Y=(Y_1,Y_2,..,Y_k)$ of $k$ paths covering all points in $\mathcal{P}_k$, operation uncross takes as an input two crossing edges $(u,u') \in Y_i$ and $(x,x') \in Y_j$ and outputs a  collection $Y'=Y \setminus{\{Y_i,Y_j\}} \cup \{Y'_i,Y'_j\}$ of $k$ paths. Path $Y'_i$ consists of the sub-path of $Y_i$ from $s$ to point $u$, edge $(u,x')$ and the sub-path of $Y_j$ from $x'$ to $t$. Similarly, path $Y'_j$ consists of the  sub-path of $Y_j$ from $s$ to point $x$, edge $(x,u')$ and the sub-path of $Y_i$ from $u'$ to $t$. 

Notice that $Y'$ is also a collection of $k$ node-disjoint $s-t$ paths, covering all points in $\mathcal{P}_k$. That is, for two crossing edges $(u,u') \in Y_i$ and $(x,x') \in Y_j$  operation uncross simply removes point $u'$ from $Y_i$ (resp. point $x'$ from $Y_{j}$) and adds point $u'$ to path $Y_j$ (resp. point $x'$ to path $Y_i$). Therefore, every point in $\mathcal{P}_k$ belongs to a path in $Y'$.

For two crossing edges $(u,u') \in Y_i$ and $(x,x') \in Y_j$ and an operation uncross, notice that the length of the closed segment $[u,x']$ is smaller than the length of the closed segments $[u,\pi]$ and $[\pi,x']$ because of the triangle inequality. Similarly, the length of the closed segment $[x,u']$ is smaller than the length of the closed segments $[x,\pi]$ and $[\pi,u']$. 

Therefore, starting from an arbitrary collection $(Y_1,Y_2,..,Y_k)$ of $k$ node-disjoint $s-t$ paths, covering all points in $\mathcal{P}_k$  we can select pairs of crossing edges among paths  in arbitrary order and perform operation uncross. Clearly, this procedure must terminate since the length of the edges appended by operation uncross is monotonically decreasing. Therefore, we have the following corollary. 

\begin{corollary}
\label{collectionwithPi}
Given a point set $\mathcal{P}_k$ such that all points $\mathcal{P}_k$ can be covered with $k$ paths, there is at least one collection $(\mathcal{Y}_1,\mathcal{Y}_2,..,\mathcal{Y}_k)$ of $k$ node-disjoint, non-crossing paths that covers all points in $\mathcal{P}_k$.
\end{corollary}

\subsection{Selection algorithm $\mathcal{S}$}
\label{selection}
Asahiro et al. \cite{DBLP:journals/dam/AsahiroHMOSY06} show an $O(n \log n)$-time algorithm $\mathcal{S}$ which takes as an input a point set $P_j$ such that all points in $P_j$ can be covered with $j$ paths, and gives as an output a collection of $j$ node-disjoint $s-t$ paths $(Y_1,Y_2,..,Y_j)$ that cover all points in $P_j$. Paths $(Y_1,Y_2,..,Y_j)$ computed by algorithm $\mathcal{S}$ may be crossing, but satisfy some useful geometrical properties. Algorithm $\widetilde{U}_k$ is based on these geometrical properties to obtain a collection of $k$ non-crossing $s-t$ paths.

\begin{definition}
For a point $x \in P_j$ we denote by $D^+_x$ all points $x^{\prime}$ in $P_j$ such that $x \prec x^{\prime}$. Similarly, we denote by $D^-_x$ all points $x^{\prime}$ in $P_j$ such that $x^{\prime} \prec x$.
\end{definition} 

 Algorithm $\mathcal{S}$ is an iterative process based on point variable $v$. At each iteration we select the point $x$ in $D^+_u$ such that $\alpha_x < \alpha_{x^{\prime}}$ for all points $x^{\prime}$ in $D^+_u \setminus{x}$. We append edge $(u,x)$, set $u \leftarrow x$ repeat the same  until the sink $t$ is selected. Upon termination of this process we have path $Y_1$. To obtain the next path $Y_2$ we repeat the same iterative process for all points in $P_j \setminus{P(Y_1)}$ where $P(Y_1)$ is the set of all points on path $Y_1$. 
 
 As shown in  Asahiro et al. \cite{DBLP:journals/dam/AsahiroHMOSY06}  after $j$ repetitions of this iterative process we have a collection of $j$ paths $(Y_1,Y_2,...Y_j)$ covering all points in $P_j$. For the remaining part of the section we denote by $(Y_1,Y_2,\ldots,Y_j)$ the collection of $j$ paths obtained by algorithm $\mathcal{S}$ over a set of points $P_j$ such that all points in $P_j$ can be covered with $j$ paths.

\begin{definition}
\label{def:rightleft}
We say that a point $v$ is on the left (resp. right) side of path $Y_i$ where $i \in [1,j]$ if the horizontal $\beta_v$ crosses with an edge $(u,u') \in Y_i$ and the crossing point $x$ on the closed segment $[u,u']$ satisfies  $\alpha_x \leq \alpha_v$ (resp. $\alpha_x \geq \alpha_v$).
\end{definition}

For an edge $(u,w) \in Y_i$ where $i \in [1,j]$ we denote by $B(u,w)=[\alpha_u,\alpha_w]\cdot [\beta_u,\beta_w]$ the rectangle formed by the verticals $\alpha_u$ and $\alpha_w$ and the horizontals $\beta_u$ and $\beta_w$. Following the description of algorithm $\mathcal{S}$ we obtain the following two corollaries.

\begin{corollary}
\label{nopointsinbox}
For an edge $(u,w) \in Y_i$ where $i \in [1,j]$ it holds that there are no points of paths $Y_{z>i}$ in the rectangle $B(u,w)$.
\end{corollary}

\begin{corollary}
\label{Scor2}
For two indexes $i,i' \in [1,j]$ such that $i< i'$ it holds that all points of path $Y_{i'}$ are on the right side of path $Y_i$.
\end{corollary} 

Notice that for two paths $Y_i$ and $Y_{i'}$ such that $i,i' \in [1,j]$ and $i<i'$, it is possible that path  $Y_i$ has points on the right side of path $Y_{i'}$ (as paths $Y_i$ and $Y_{i'}$ can cross).

Consider a set of points $P_j$ such that all points in  can be covered with $j$ paths. Let $(Y_1,Y_2,\ldots,Y_j)$ be a collection of $j$ paths obtained by algorithm $\mathcal{S}$, covering all points in $P_j$. Lemmas \ref{Sprop}, \ref{Sprop1} and \ref{sprop2} specify the geometrical properties satisfied by paths $(Y_1,Y_2,\ldots,Y_j)$.

\begin{lemma}
\label{Sprop}
For two indexes $i,i' \in [1,j]$ such that $i<i'$ and an edge $(x,x^{\prime}) \in Y_{i'}$, let $v_1,v_2,...,v_m$ be all points of path $Y_i$ (ordered in topological order)  within the horizontals $\beta_x$ and $\beta_{x'}$ on the right side of path $Y_{i'}$. It holds that $\alpha_x < \alpha_{v_i} < \alpha_{x^{\prime}}$ for $i=1,2,..,m$.
\end{lemma}
\begin{proof}
We refer the reader to Figure \ref{fig:withintriangle} for the schematic representation of the proof.
Notice that since points $v_1,v_2,\ldots,v_m$ are given in topological order we naturally have $\alpha_{v_i} < \alpha_{v_{i+1}}$ for $i=1,2,...,m-1$. Thus, it suffices to show that $\alpha_x < \alpha_{v_1}$ and $\alpha_{v_m} \leq \alpha_{x'}$. The inequality $\alpha_x < \alpha_{v_1}$  holds because point $v_1$ is on the right side of edge  $(x,x^{\prime})$ (see Definition \ref{def:rightleft}).

Assume towards contradiction that $\alpha_{v_m} > \alpha_{x^{\prime}}$. All points in $Y_i$ after $v_m$ must be above the horizontal $\beta_{v_m}$ and on the right of the vertical $\alpha_{v_m}$. Thus, the edge $(v_m,\sigma(v_m))$,  where $\sigma(v_m)$ is the successor of $v_m$ in $Y_i$ crosses the horizontal $\beta_{x'}$ at a crossing point $p$ such that $\alpha_p > \alpha_{x'}$. This means that $x'$ is on the left of $Y_i$. However, this contradicts Corollary \ref{Scor2} since $x^{\prime}$ is a point in path $Y_{i'>i}$ and therefore must be on the right side of $Y_i$.
\end{proof}

\begin{lemma}
\label{Sprop1}
For an edge $(x,x') \in Y_z$ where  $z \in [1,j]$, let $f_i$ and $f'_i$ be the first point of path $Y_i$ for $i=1,2,..,z-1$ above the horizontal $\beta_x$ and above the horizontal $\beta_{x'}$, respectively. It holds that for $i=1,2,..,z-1$ point $f_i$ is on the left of the vertical $\alpha_x$ and point $f'_i$ is on  the left of vertical $\alpha_{x'}$.
\end{lemma}
\begin{proof}

We refer the reader to Figure \ref{fig:pointsbelowabove} for the schematic representation of the proof. 
Consider an edge $(x,x') \in Y_z$ where $z \in [1,j]$ and assume towards contradiction that for some $i<z$ point $f_i$ is not placed  on the left of the vertical $\alpha_x$. Let $\pi(f_i)$ be the predecessor of $f_i$ in $Y_i$. Since $f_i$ is the first point of $Y_i$ above the horizontals  $\beta_x$ we have that edge $(\pi(f_i),f_i)$ must cross the horizontal  $\beta_x$. Thus, we have $\beta_{\pi(f_i)} \leq \beta_x \leq \beta_{f_i}$. 

Because $x$ is a point on path $Y_{z>i}$, according to Corollary \ref{Scor2} it must be on the right side of path $Y_i$. Therefore the horizontal $\beta_x$ must cross with edge $(\pi(f_i),f_i)$ at a point $p$ such that $\alpha_p < \alpha_x$ which means that point  $\pi(f_i)$ is on the left of the vertical $\alpha_x$.

If $f_i$ is not on the left of the vertical $\alpha_x$ then $x$ satisfies the inequality $\alpha_{\pi(f_i)} < \alpha_x < \alpha_{f_i}$. As discussed above, we have that $\beta_{\pi(f_i)} \leq \beta_x \leq \beta_{f_i}$ and therefore $\pi(f_i) \prec x \prec f_i$. However, this  contradicts Corollary \ref{nopointsinbox} because there is an edge $(\pi(f_i),f_i) \in Y_i$ and a point $x$ on $Y_{z>i}$ such that $x$ is within the rectangle $B(\pi(f_i),f_i)$.

Similarly, assume towards contradiction that $f^{\prime}_i$ is not on the left of the vertical $\alpha_{x^{\prime}}$. Let $\pi(f^{\prime}_i)$ be the predecessor of $f^{\prime}_i$ in $Y_i$. Since $f^{\prime}_i$ is the first point of $Y_i$ above the horizontal $\beta_{x^{\prime}}$ it holds that the horizontal $\beta_{x^{\prime}}$ crosses with edge $(\pi(f^{\prime}_i),f^{\prime}_i)$. Thus, we have that  $\beta_{\pi(f'_i)} < \beta_x' < \beta_{f'_i}$.

Because $x^{\prime}$ is a point on  $Y_{z>i}$,  according to Corollary \ref{Scor2} it must be on the right side of path $Y_i$. Thus, the horizontal  $\beta_{x^{\prime}}$  crosses with edge $(\pi(f^{\prime}_i),f^{\prime}_i)$ at a point $p$ such that $\alpha_p < \alpha_{x^{\prime}}$. Therefore, point $\pi(f^{\prime}_i)$ must be on the left of the vertical $\alpha_{x^{\prime}}$.

If $f^{\prime}_i$ is not on the left of the vertical  $\alpha_{x^{\prime}}$ then $x^{\prime}$ satisfies the inequality $\alpha_{\pi(f^{\prime}_i)} < \alpha_x^{\prime} < \alpha_{f^{\prime}_i}$. We also have $\beta_{\pi(f'_i)} < \beta_x' < \beta_{f'_i}$ which implies that $\pi(f'_i) \prec x' \prec f'_i$ which subsequently contradicts Corollary \ref{nopointsinbox} because $x'$ is a point on $Y_{z>i}$ and is within the rectangle $B(\pi(f'_i),f'_i)$ where $(\pi(f_i),f_i) \in Y_i$.
\end{proof}

\begin{figure}
\centering
\begin{subfigure}{.5\textwidth}
  \centering
  \includegraphics[scale=0.5]{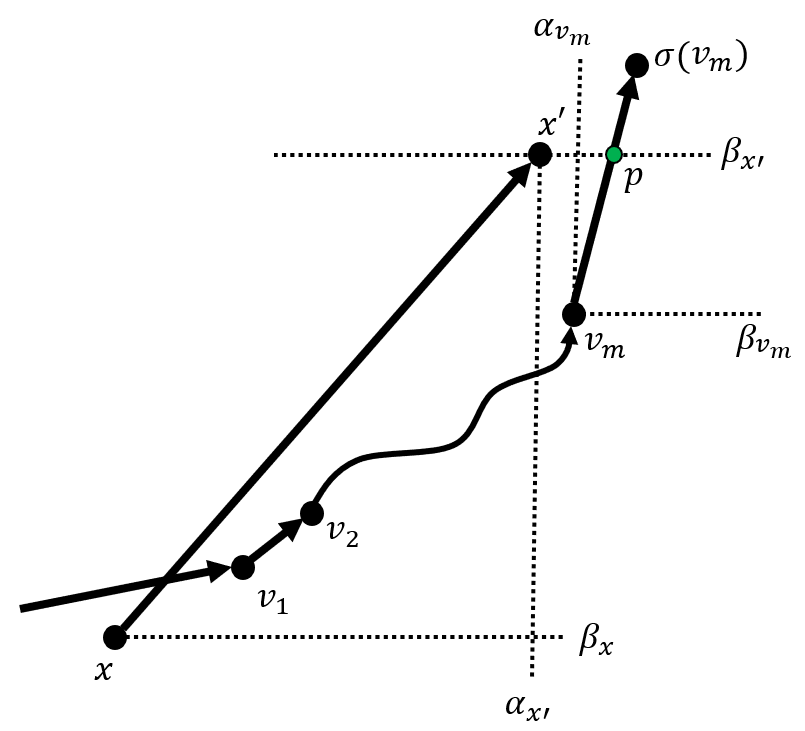}
        \caption{}
        \label{fig:withintriangle}
\end{subfigure}%
\begin{subfigure}{.5\textwidth}
  \centering
  \includegraphics[scale=0.5]{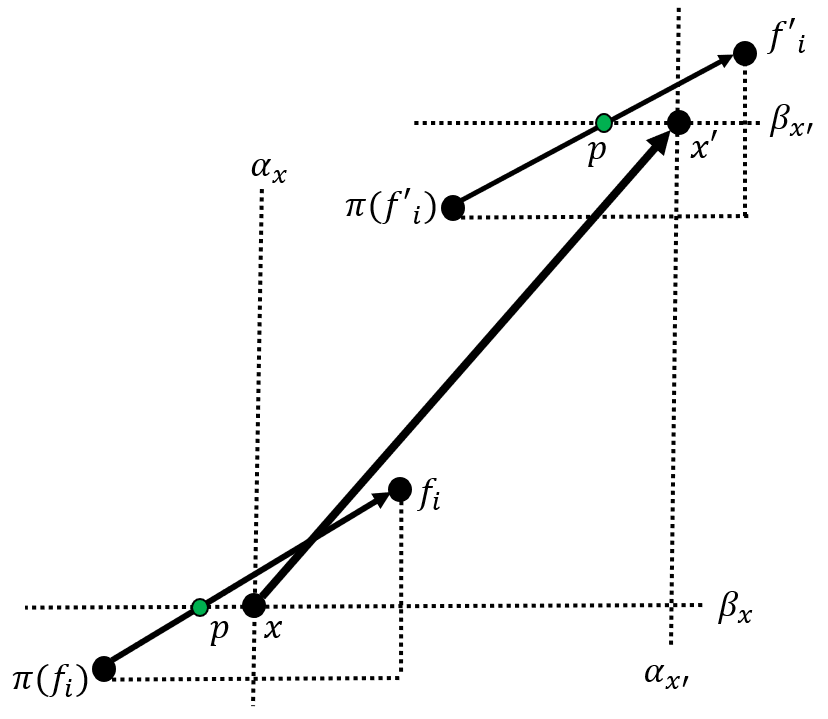}
        \caption{}
        \label{fig:pointsbelowabove}
\end{subfigure}
\caption{Figure \ref{fig:withintriangle} shows the schematic representation of the proof for Lemma \ref{Sprop}.  Figure \ref{fig:pointsbelowabove} the schematic representation of the proof for Lemma \ref{Sprop1}.}
\end{figure}

\begin{lemma}
\label{sprop2} 
 For any $i \in [1,j]$ and a point $x$ on path $Y_i$ there is no collection of $i-1$ paths that covers all points on paths $Y_1,Y_2,..,Y_{i-1}$ and point $x$.
\end{lemma}
\begin{proof}
 Assume towards contradiction that  for some $i \in [1,j]$ and a point $x$ on path $Y_i$, there is a collection of $i-1$ paths that includes all points on paths $Y_1,Y_2,..,Y_{i-1}$ and additionally point $x$. Notice that if such a collection exists then there must be at least one point $u$ on some path $Y_{z<i}$ such that $u \prec x $, since we need to append edge $(u,x)$ (\textit{i.e.} in order to cover point $x$). Let  $\sigma_u$ be the successor of point $u$ in $Y_z$. From corollary \ref{nopointsinbox} we have that there are no points of paths $Y_{i>z}$ within the rectangle $B(u,\sigma_u)$ which contradicts that $u \prec x$, since $x$ is on path $Y_i$.
\end{proof}
 \subsection{Algorithm $\widetilde U_k$ description}
\label{algoUK}
 For $k \geq 2$ algorithm  $\widetilde U_k$ takes as an input a set of points $\mathcal{P}_k$ such that all points in $\mathcal{P}_k$ can be covered with $k$ node disjoint $s-t$ paths and outputs a collection of $(\mathcal{Y}_1,\mathcal{Y}_2,..,\mathcal{Y}_k)$ of $k$ node-disjoint non-crossing paths. The description of  algorithm $\widetilde{U}_k$ for $k \geq 2$ is given below.
 
 For $k \geq 2$ the working of algorithm $\widetilde U_k$ is described as sequence of algorithms $U_k,U_{k-1},..,U_1$ such that for $j=k,k-1,\ldots,1$ algorithm $U_j$ computes the $j^{th}$ path $\mathcal{Y}_j$ of the collection $(\mathcal{Y}_1,\mathcal{Y}_2,..,\mathcal{Y}_k)$. That is, algorithm $\widetilde U_k$ outputs the paths $(\mathcal{Y}_1,\mathcal{Y}_2,..,\mathcal{Y}_k)$ from right to left, starting with the rightmost path  $\mathcal{Y}_k$ and ending with  the leftmost path $\mathcal{Y}_1$. 
 
 For $j =k,k-1,\ldots,1$ the input to algorithm $U_j$ is a set of points $P_j$ such that all points in $P_j$ can be covered with  $j$ paths. We call this condition the \textit{input condition} of algorithm $U_j$. For the base case $j=k$, the  input condition is satisfied since $P_k=\mathcal{P}_k$. For $j<k$ the input condition of algorithm $U_j$ will be satisfied inductively, as explained below.

For $j =k,k-1,\ldots,1$, algorithm $U_j$ gives as an output the $j^{th}$ path $\mathcal{Y}_j$ of the collection $(\mathcal{Y}_1,\mathcal{Y}_2,..,\mathcal{Y}_k)$. Let $P(\mathcal{Y}_j)$ be all points on path $\mathcal{Y}_j$.  For $j=k,k-1,..,1$ the  \textit{output condition} of algorithm $U_j$ is described by the two properties shown below:

\begin{itemize}
    \item Property 1: All points in $P_j \setminus{P(\mathcal{Y}_j)}$ can be covered with $(j-1)$ paths.
    \item Property 2: Any collection of $(j-1)$ paths that collects all points in $P_j \setminus{P(\mathcal{Y}_j)}$ does not cross with path $\mathcal{Y}_j$. 
\end{itemize}

 For $j=k,k-1,\ldots,2$ the input to algorithm $U_{j-1}$ is the set of points $P_j \setminus{P(\mathcal{Y}_j)}$. It is easy to see that for $j \leq k$ if algorithm $U_j$ satisfies Property 1 of the output condition, then the set of points $P_j \setminus{P(\mathcal{Y}_j)}$ satisfies the input condition of algorithm $U_{j-1}$. That is, all points in $P_{j}  \setminus{P(\mathcal{Y}_{j})}$ can be covered with $j-1$ paths.

For $j=k,k-1,\ldots,1$ and a point set $P_j$ such that all points in $P_j$ can be covered with $j$ paths, algorithm $U_j$ consists of two computational steps. The first step obtains a collection of $j$ paths $(Y_1,Y_2,..,Y_j)$ covering all points in $P_j$  by calling algorithm $S$, as described in   subsection \ref{selection}. 

The second step traverses the edges of the leftmost path $Y_j$ from $s$ to $t$ in order to  build the output path $\mathcal{Y}_j$. The output path $\mathcal{Y}_j$ has all points of path $Y_j$ and some additional points of paths $Y_1,Y_2,\ldots,Y_{j-1}$.  For an edge $(x,x') \in Y_j$ algorithm $U_j$ builds the segment of the output path $\mathcal{Y}_j$ from $x$ to $x'$ by performing either operation $A$ or operation $B$. Before we describe operation $A$ and operation $B$ for an edge $(x,x') \in Y_j$  we provide the following definition. 
  
  \begin{definition}
  \label{defprop2}
  Let $P_j \setminus{P(Y_j)}$ be the set of all points on paths $Y_1,Y_2,..,Y_{j-1}$.
  For an edge $(x,x') \in Y_j$ we denote by $P_{xx'}$ the set of all points in $P_j \setminus{P(Y_j)}$ within the horizontals $\beta_x$ and $\beta_{x'}$. We denote by $V_{xx'}$ the set of all points on the output  path $\mathcal{Y}_j$ within the horizontals $\beta_x$ and $\beta_{x'}$.
\end{definition}

Operation A,  sets $V_{xx'}= \emptyset$ and therefore the segment of $\mathcal{Y}_j$ from $x$ to $x'$ simply consists of edge $(x,x')$.  Operation B sets $V_{xx'}= (v_1,v_2,..,v_m)$  where point $v_i \in P_{xx'}$ for $i=1,2,..m$ and therefore the segment of $\mathcal{Y}_j$ from $x$ to $x'$ is a path $(x,v_1,v_2,..,v_m,x')$. We will shortly explain in detail how we select points $v_1,v_2,..,v_m$. 

Essentially, operation $A$ simply considers the next edge in $Y_j$ whereas operation $B$ can be seen as $m$ insertions where for $i=1,2,..,m$ the $i^{th}$ insertion  removes a point $v_i \in P_{xx'}$ on a path $Y_{z<j}$ and adds $v_i$ to the segment of $\mathcal{Y}_j$ from $x$ to $x'$. 

\subsubsection{Property 1}
\label{property1}
For $j=k,k-1,..,2,1$ and a point set $P_j$ such that all points can be covered with $j$ paths let $(Y_1,Y_2,\ldots,Y_j)$ be the collection of paths obtained by algorithm $\mathcal{S}$, covering all points in $P_j$. According to Lemma \ref{sprop2}, there is no collection of $(j-1)$ paths that covers all points in~$(P_j \setminus{P(Y_j)}) \cup x$ where $x$ is a point on path $Y_j$. This implies, that  for $j=k,k-1,\ldots,1$
if algorithm $U_j$ satisfies Property 1 of the {output condition}, then any point $x$ on path $Y_j$ is also a point on path $\mathcal{Y}_j$.

\begin{lemma}
\label{LemmaProperty1}
For $j=k,k-1,\ldots,2$ algorithm $U_j$ takes as an input a point set $P_j$ such that all points can be covered with $j$ paths and outputs a path $\mathcal{Y}_j$ such that all points in $P_j \setminus{P(\mathcal{Y}_j})$ can be covered with $j-1$ paths.
\end{lemma} 
\begin{proof}
 We show the Lemma by induction. For the base case $j=k$, we assume that the point set $P_k$ can be covered with $k$ path  and show that algorithm $U_k$ outputs path $\mathcal{Y}_k$ such that  all points in $P_k \setminus{P(\mathcal{Y}_k)}$ can be covered with $k-1$ paths.
 
 For $j=k$ the first step of algorithm $U_k$ obtains an initial collection of $k$ paths $(Y_1,Y_2,..,Y_k)$ using algorithm $S$. If all points in $P_k$ can be covered with $k$ paths, then naturally all points in $P_k \setminus{P(Y_k)}$ or any subset $X \subseteq P_k \setminus{P(Y_k)}$  can be covered with $k-1$ paths.
 
 Algorithm $U_k$ builds path $\mathcal{Y}_k$ by considering the edges of $Y_k$ from $s$ to $t$ and for an edge $(x,x') \in Y_k$ performs either operation $A$ or operation $B$ to build the segment of $\mathcal{Y}_k$ from $x$ to $x'$.  Notice that for an edge $(x,x') \in Y_k$, neither operation $A$ nor operation $B$ removes a point from path $Y_j$. 
 
  Therefore, for path $\mathcal{Y}_k$ we have that $P(\mathcal{Y}_k)=P(Y_k) \cup V $ where  $V \subseteq P_k \setminus{P(Y_k)}$. This implies that $P_k \setminus{P(\mathcal{Y}_k)} \subseteq P_k \setminus{P(Y_k)}$ since $|P(\mathcal{Y}_k)|>|P(Y_k)|$. Therefore, all points in $P_k \setminus{P(\mathcal{Y}_k)}$ can be covered with $k-1$ paths.
 
We now assume that our induction holds for $k,k-1,\ldots,j+1$ and show that it holds for $j$. If our induction holds for $j+1$ then this  means that the set of points $P_{j+1} \setminus{P(\mathcal{Y}_{j+1}})$ can be covered with $j$ paths. We show that when algorithm $U_j$ takes as input point set $P_j=P_{j+1} \setminus{P(\mathcal{Y}_{j+1}})$ and  outputs a path $\mathcal{Y}_j$ then the point set $P_j \setminus{P(\mathcal{Y}_j})$ can be covered with $j-1$ paths.  To complete the proof we simply follow the same methodology as for $j=k$.
\end{proof}

\subsubsection{Property 2}
\label{Property2}
For $j=k,k-1,..,2,1$ and a point set $P_j$ such that all points can be covered with $j$ paths let $(Y_1,Y_2,\ldots,Y_j)$ be the collection of paths obtained by algorithm $\mathcal{S}$. Recall that for an edge $(x,x') \in Y_j$, we denote by $f_i$ and $f'_i$ for $i=1,2,...,j-1$  the first point of path $Y_i$ above the horizontal $\beta_{x}$ and $\beta_{x'}$, respectively. Further, according to Definition \ref{defprop2}, $P_{xx'}$ is the set of all points in $P_j \setminus{P(Y_j)}$ within the horizontals $\beta_x$ and $\beta_{x'}$ and $V_{xx'}$ is the set of all points on the segment of $\mathcal{Y}_j$ from $x$ to $x'$.

\begin{definition}
\label{def:invariantIxx}
For an edge $(x,x') \in Y_j$ we say that invariant $I_{xx'}$ is satisfied if the edge between two points $u,u' \in(P_{xx'} \cup (f'_1,f'_2,..,f'_{j-1}) \setminus{V_{xx'}})$ can not cross the segment of $\mathcal{Y}_j$ from $x$ to $x'$.
\end{definition}

 For an edge $(x,x') \in Y_j$ let $(x,x',\pi)$ be the triangle formed by the closed segments $[x,x'],[x,\pi]$ and $[x',\pi]$ where $\pi$ is the crossing point of the vertical $\alpha_x$ with the horizontal $\beta_{x'}$. Notice that any point $u \in P_{xx'}$   within the triangle $(x,x',\pi)$ must satisfy $x \prec u \prec x'$ according to Lemma \ref{Sprop}.  Thus, for an edge $(x,x') \in Y_j$ there are only two  cases, as shown below.

\begin{itemize}
    \item Case 1: There is no point $u \in P_{xx'}$  within the triangle $(x,x',\pi)$.
    \item Case 2: There is at least one point $u \in P_{xx'}$  within the triangle $(x,x',\pi)$.
\end{itemize}

If an edge $(x,x') \in Y_j$ belongs to Case 1, algorithm $U_j$ performs operation $A$ which sets  $V_{xx'}=\emptyset$, whereas if edge $(x,x')$ belongs to Case 2, algorithm $U_j$ performs operation $B$, as defined below.

\begin{definition}
For an edge $(x,x') \in Y_j$ we define $\mathcal{C}_{xx'}=(v_1,v_2,...,v_m)$ to be the convex hull of all points in $P_{xx'}$ within the triangle $(x,x',\pi)$. For an edge $(x,x') \in Y_j$ operation $B$ sets $V_{xx'}=\mathcal{C}_{xx'}$.
\end{definition}

\begin{lemma}
\label{case1indedge}
If an edge $(x,x') \in Y_j$ belongs to Case 1, algorithm $U_j$ performs operation $A$ and invariant $I_{xx'}$ is satisfied.
\end{lemma}
\begin{proof}
Operation $A$ sets $V_{xx'}=\emptyset$ and therefore the segment of $\mathcal{Y}_j$ from $x$ to $x'$ simply consists of edge $(x,x')$. Thus, it is sufficient to show that the edge between two points $u,u' \in P_{xx'} \cup{(f'_1,f'_2,..,f'_m)}$ does not cross edge $(x,x')$. Because edge $(x,x')$ belongs to Case 1, we have that all points in $P_{xx'}$ are on the exterior of the triangle $(x,x',\pi)$. Thus, any point $u \in P_{xx'}$ is on the left of edge $(x,x')$.

Therefore, any edge $(u,u')$ such that $u,u' \in P_{xx'}$ can not cross edge $(x,x')$. 
For the special case where the first edge $(s,x) \in Y_j$ belongs to Case 1, then an edge $(s,u')$ where $u' \in P_{xx'}$ can not cross edge $(s,x)$ because they share a point.

We now consider an edge of the form $(u,f'_i)$ for any $i \in [1,j-1]$ such that $u \in P_{xx'}$. According to Lemma  \ref{Sprop1}, for $i=1,2,..,j-1$ we have that $f'_i$ is on the left of vertical $\alpha_{x'}$. Since every point in $P_{xx'}$ is on the left of edge $(x,x')$ we have that edge $(u,f'_i)$ can not cross edge $(x,x')$.  For the special case where the first edge $(s,x) \in Y_j$ belongs to Case 1, trivially an edge $(s,f'_i)$ can not cross edge $(s,x)$ because they share a point.

We conclude that any edge $(u,u')$ where $u,u' \in P_{xx'} \cup (f'_1,f'_2,..,f'_{j-1})$ does not cross the segment of $\mathcal{Y}_j$ from $x$ to $x'$ and therefore invariant $I_{xx'}$ is satisfied. This completes the proof.
\end{proof}

\begin{lemma}
\label{Case2ConvH}
If an edge $(x,x') \in Y_j$ belongs to Case 2, algorithm $U_j$ performs operation $B$, sets $V_{xx'}=\mathcal{C}_{xx'}$ and invariant $I_{xx'}$ is satisfied.
\end{lemma}
\begin{proof}
For an edge $(x,x') \in Y_j$ that belongs to Case 2 let $\mathcal{C}_{xx'}=(v_1,v_2,...,v_m)$ be the convex hull of all points in $P_{xx'}$. Because $\mathcal{C}_{xx'}=V_{xx'}$ it suffices to show that an edge $(u,u')$ such that $u,u'  \in P_{xx'}\cup (f'_1,f'_2,..,f'_m) \setminus{C_{xx'}}$ can not cross the segment of $\mathcal{Y}_j$ from $x$ to $x'$. That is, edge $(u,u')$ can not cross any of the edges $(x,v_1),\ldots,(v_i,v_{i+1}),\ldots, (v_m,x')$.

Notice that for $i=1,2,..,m-1$ we have that  $v_i \prec v_{i+1}$  because $v_i$ and $v_{i+1}$ are points in the convex hull $\mathcal{C}_{xx'}$. We also have $x,\prec v_1$ and $v_m \prec x'$ since according to Lemma \ref{Sprop} for every point $u \in P_{xx'}$ within the triangle $(x,x',\pi)$ we have $x \prec u \prec x'$. 

We refer the reader to Figure \ref{fig:convexhull} for the schematic representation of the proof.
According to the convex hull principle, an edge between two points $u,u' \in P_{xx'} \setminus{C_{xx'}}$ can not cross edge $(v_i,v_{i+1})$ for $i=1,2,..,m-1$, as this would imply that point $u'$ is in the exterior of the convex hull. For the same reason the edge between two points $u,u' \in P_{xx'} \setminus{C_{xx'}}$ can not cross edge $(x,v_1)$ or edge $(v_m,x')$.

We now consider an edge of the form $(u,f'_z)$ where $z \in [1,j-1]$ such that $u \in P_{xx'} \setminus{C_{xx'}}$. Consider the convex body which extends above the horizontal $\beta_{x'}$ by taking the vertical $\alpha_{x'}$, as shown in Figure \ref{fig:convexhull}. 

Point $u$ is in the interior of the convex body. According to Lemma \ref{Sprop1}, for $z=1,2,..,j-1$ we have that $f'_z$ is on the left of vertical $\alpha_{x'}$ and therefore $f'_z$ is also on the interior of the convex body. Thus, edge $(u,f'_z)$ can not cross any of the edges $(x,v_1),\ldots,(v_i,v_{i+1}),\ldots, (v_m,x')$. This completes the proof.
\end{proof}

\begin{figure}
    \centering
    \includegraphics[scale=0.46]{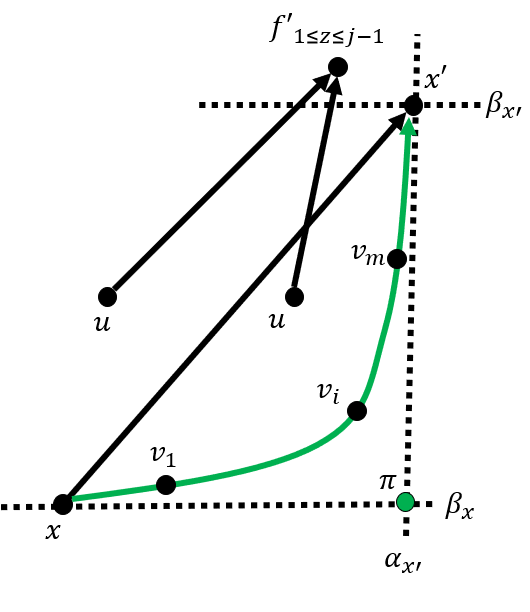}
    \caption{The schematic representation of the proof for Lemma \ref{Case2ConvH}.}
    \label{fig:convexhull}
\end{figure}

\begin{lemma}
\label{LemmaProperty2}
For $j=k,k-1,\ldots,2$ algorithm $U_j$ takes as an input a point set $P_j$ such that all points in $P_j$ can be covered with $j$ paths and outputs a path $\mathcal{Y}_j$ such that any collection of $j-1$ paths that covers all points in $P_j \setminus{P(\mathcal{Y}_j})$ does not cross with path $\mathcal{Y}_j$.
\end{lemma}
\begin{proof}
For $j=k,k-1,\ldots,2$, let $(Y_1,Y_2,\ldots,Y_j)$ be the collection of $j$ paths covering $P_j$ obtained by algorithm $\mathcal{S}$. Algorithm $U_j$ traverses the edges of path $Y_j$ from $s$ to $t$ and when an edge $(x,x') \in Y_j$ is considered, it performs operation $A$ if edge $(x,x')$ belongs to Case 1 and operation $B$ if edge $(x,x')$ belongs to Case 2. According to Lemmas \ref{case1indedge} and \ref{Case2ConvH}, for an edge $(x,x') \in Y_j$ when algorithm $U_j$ performs either operation $A$ or operation $B$, invariant $I_{xx'}$ is satisfied. 

According to Definition \ref{def:invariantIxx}, if invariant $I_{xx'}$ is satisfied for an edge edge $(x,x') \in Y_j$ then the edge between two points $u,u' \in(P_{xx'} \cup (f'_1,f'_2,..,f'_{j-1}) \setminus{V_{xx'}})$ can not cross the segment of $\mathcal{Y}_j$ from $x$ to $x'$. Summing over all edges $(x,x') \in Y_j$, invariant $I_{xx'}$ implies that any edge $(u,u)$ such that $u,u' \in P_j \setminus{P(\mathcal{Y}_j)}$ can not cross with an edge of path $\mathcal{Y}_j$. Thus, any collection of $j-1$ paths that covers all points in $P_j \setminus{P(\mathcal{Y}_j})$ can not cross with path $\mathcal{Y}_j$.
 \end{proof}

\begin{proof}[Proof of Theorem \ref{maintheoremuncrossing}]
We show the proof of Theorem \ref{maintheoremuncrossing} by induction using Lemmas \ref{LemmaProperty1} and \ref{LemmaProperty2}. For $k \geq 2$, algorithm $\widetilde U_k$ consists of the following sequence of algorithm $U_k,U_{k-1},\ldots,U_2,U_1$. We show that for $j=k,k-1,\ldots,2,1$ algorithm $U_j$ computes the  $j^{th}$ path $\mathcal{Y}_j$ of a collection of $k$ node-disjoint, non-crossing $s-t$ paths  $(\mathcal{Y}_1,\mathcal{Y}_2,\ldots,\mathcal{Y}_k)$. The paths are given in left to right order in their planar representation, meaning that $\mathcal{Y}_1$ is the leftmost path and $\mathcal{Y}_{k}$ is the rightmost path. Algorithm $\widetilde U_k$ iteratively outputs the paths from right to left (\textit{i.e.}. from $\mathcal{Y}_k$ to $\mathcal{Y}_1$).

For the base case $j=k$,  algorithm $U_k$ takes as input a set of points  $P_k$ such that all points in $P_k$ can be collected with $k$ paths and  gives as an output path $\mathcal{Y}_k$. According to Lemma \ref{LemmaProperty1} we have that all points in $P_k \setminus{P(\mathcal{Y}_k})$ can be covered with $k-1$ paths, which means that the set of points $P_k \setminus{P(\mathcal{Y}_k})$ is a valid input for algorithm  $U_{k-1}$. According to Lemma \ref{LemmaProperty2} we have that any collection of $k-1$ paths covering all points in $P_k \setminus{P(\mathcal{Y}_k})$ can not cross with path $\mathcal{Y}_k$.  

We now assume that our induction holds for $i=k,k-1,\ldots,j+1$ and show that it also holds for $j$. According to our inductive hypothesis, we have a collection of $(k-j)$ non-crossing $s-t$ paths $\mathcal{Y}_k,\mathcal{Y}_{k-1},\ldots, \mathcal{Y}_{j+1}$ and a set of points $P_j=\mathcal{P}_k \setminus{P(\mathcal{Y}_k,\mathcal{Y}_{k-1},\ldots, \mathcal{Y}_{j+1})}$ such that all points in $P_j$ can be covered with $j$ paths. 

Furthermore,  any collection of $j$ paths over $P_j$ can not cross with path $\mathcal{Y}_{j+1}$. This implies that  any collection of $j$ paths over $P_j$ can not cross with paths $\mathcal{Y}_{j+2},\ldots,\mathcal{Y}_{k-1},\mathcal{Y}_k$ since they are on the right side of $\mathcal{Y}_{j+1}$ and do not cross pairwise.

Algorithm $U_j$ takes point set $P_j$ as an input and outputs path $\mathcal{Y}_j$. According to Lemma  \ref{LemmaProperty1}, the set of points $P_j \setminus{P(\mathcal{Y}_j)}$ can be covered with $j-1$ paths, which means that point set $P_j \setminus{P(\mathcal{Y}_j)}$ is a valid input for algorithm $U_{j-1}$. According to Lemma  \ref{LemmaProperty2}, any collection of $j-1$ paths that covers all points in $P_j \setminus{P(\mathcal{Y}_j})$ can not cross with path $\mathcal{Y}_j$.

Notice that the computed path $\mathcal{Y}_j$ and any collection of $j-1$ paths that covers all points in $P_j \setminus{P(\mathcal{Y}_j)}$ is a collection of $j$ paths that covers all points in $P_j$. According to the induction hypothesis, any collection of $j$ paths that covers all points in $P_j$ can not cross with paths $\mathcal{Y}_k,\mathcal{Y}_{k-1},\ldots, \mathcal{Y}_{j+1}$. Thus, the computed path $\mathcal{Y}_j$ can not cross with paths $\mathcal{Y}_k,\mathcal{Y}_{k-1},\ldots, \mathcal{Y}_{j+1}$. 

To complete the proof it remains to show that algorithm $\widetilde U_k$ requires $O(k n \log n)$ time. We show that  for $j=k,k-1,\ldots,2,1$ algorithm $U_j$ requires $O(n \log n)$ time. For any $j \in [1,k]$ algorithm $U_j$ consists of two computational steps. The first step obtains a collection of $j$ paths $(Y_1,Y_2,\ldots,Y_j)$ using algorithm $\mathcal{S}$ which requires $O(n \log n)$ time, as shown in Asahiro et al.\cite{DBLP:journals/dam/AsahiroHMOSY06}. The second step, traverses the edges of the rightmost path $Y_j$ from $s$ to $t$ and for an edge $(x,x') \in Y_j$ performs either operation $A$ or operation $B$.

It is easy to see that operation $A$ requires $O(1)$ time since we simply consider the next edge of $Y_j$. We recall that for an edge $(x,x') \in Y_j$ we denote by  $P_{xx'}$ the set of all points on paths $Y_1,Y_2,\ldots,Y_{j-1}$ within the horizontals $\beta_x$ and $\beta_{x'}$. For an edge $(x,x') \in Y_j$, operation $B$ computes the convex hull of all points  in $P_{xx'}$ using Graham's scan algorithm \cite{Cormen2001introduction} and therefore the running time of operation $B$ is $O(|P_{xx'}| \log |P_{xx'}|)$. Summing over all edges of path $Y_j$, it is easy to see that the total running time of all operations $B$ is $O(n \log n)$ since the input point set $P_j$ can have at most $n$ points.
\end{proof}

\section{Conclusion and Future Work}
We study the optimisation problem of servicing $n$ timed requests on a line by $k$ robots, a generalisation of the Ball Collecting Problem \cite{DBLP:journals/dam/AsahiroHMOSY06} for arbitrary ball weights.  The optimisation problem is modelled as a minimum cost flow problem on a flow network $\mathcal{N}$, which can be implicitly represented by a set of points in the two-dimensional plane. 
We  show an algorithm with the running time of $O(k^{2k}n \log^{2k} n)$ for computing a minimum-cost  flow of value $k$
in $\mathcal{N}$, which improves the previous upper bound of $O(kn^2)$ if $k$ is considered constant.
For $k \geq 2$, a natural question is whether there exists  an algorithm with the running time of $O(kn \log^c n)$ for some constant $c \geq 1$ (or ideally independent of $k$), that computes a minimum-cost flow of value $k$ in the flow network $\mathcal{N}$.

For $k \geq 2$, we compute a minimum cost flow of value $k$ in $\mathcal{N}$ by iteratively finding an $s-t$ shortest path in the residual network $\mathcal{N}_i$ for $i=1,2,\ldots,k-1$. For $k=1$, an $s-t$ shortest path is computed in the standard way using appropriate data structures \cite{Lueker197828} for efficiency.
Our algorithm is based on the analysis of the geometric structure of non-self-crossing shortest paths, using the implicit representation of the flow network. 

We also rely on the fact that for $k \geq 2$, a minimum-cost flow  of value $k$ can be represented by $k$ non-crossing red paths $\mathcal{Y}_1,\mathcal{Y}_2,\ldots,\mathcal{Y}_{k}$ (Theorem \ref{maintheoremuncrossing}). We do not know how to find efficiently a non-self-crossing $s-t$ shortest path  in the residual network without the assumption that the red paths are non-crossing. We also do not know how to find efficiently an $s-t$ shortest path that does not cross the red paths, and subsequently ensures that the new $k$ paths do  not cross, if we obtain the flow in the usual way. This is also why we need the follow-up computation of 
un-crossing (algorithm $\widetilde{U}_k$).

Faster algorithms may depend on the existence of shortest paths with more specialised structures. For example,  non-self-crossing shortest paths which additionally do not cross the red paths. For $k =2$, we do not rely on the existence of non-self-crossing shortest paths (Theorem \ref{Theorem:specialcasek=2}) so these two conditions may not be useful. However, for $k \geq 3$, a $(k-2)$-chromatic path satisfying these two conditions, can traverse red edges of either path $\mathcal{Y}_1$ or path $\mathcal{Y}_{k-1}$, but not both, since the red paths are non-crossing. This observation can yield to more efficient algorithms, but it requires the existence of shortest paths with this specialised structure (not trivial).

Finally, an interesting research direction is generalising the problem of $n$ servicing timed requests with $k$ robots, in  three dimensions, where now  a request takes place at time $t_i$ at point $(x_i,y_i)$ in space. Considering the fact that the appropriate data structures for orthogonal search queries can be generalized to higher dimensions, this is an interesting research direction even for $k=1$ robot.

\bibliographystyle{unsrt}
\bibliography{foo.bib}
\end{document}